\documentclass[twoside]{article}

\PassOptionsToPackage{numbers, compress}{natbib}

\usepackage[preprint]{jmlr2e_modified}

\usepackage{verbatim}
\usepackage{algorithm}
\usepackage{algpseudocode} 
\usepackage{color,graphicx}  
\usepackage{amsfonts}   
\usepackage{amsmath}
\usepackage{amssymb}
\newtheorem{thm}{Theorem}[section]

\newtheorem{prop}{Proposition}[section]

\newtheorem{alg}[thm]{Algorithm}
\newtheorem{defn}{Definition}[section]
\newtheorem{ex}{Example}[section]
\newtheorem{rmk}{Remark}[section]

\newcommand{\RR}{\mathbb{R}}      
\newcommand{\CC}{\mathbb{C}}      
\newcommand{\NN}{\mathbb{N}}      

\newcommand{\vecc}{\boldsymbol}

\jmlrheading{-}{2020}{-}{-}{12/20}{-}{Soon Hoe Lim}

\ShortHeadings{Stochastic Processes: From Classical to Quantum}{Soon Hoe Lim}
\firstpageno{1}

\begin{document}

\title{Stochastic Processes: From Classical to Quantum}

\author{\name Soon Hoe Lim 
 \email shlim@kth.se \\
        \addr Department of Mathematics, KTH Royal Institute of Technology \\
       \addr Nordita, KTH Royal Institute of Technology and Stockholm University \\
       Stockholm 106 91, Sweden 
       }

\maketitle

\begin{abstract}
The main goal of these notes is to give an introduction to the mathematics of quantum noise and some of its applications in non-equilibrium statistical mechanics. We start with some reminders from the theory of classical stochastic processes.  We then provide a brief overview of quantum mechanics and quantum field theory, from the viewpoint of quantum probability and adopting the language of Hudson and Parthasarathy \citep{parthasarathy2012introduction}. We introduce quantum stochastic processes on a boson Fock space and their calculus. Whenever possible, we make connections with the relevant concepts in classical probability theory. As an application of the theory, we introduce the theory of open quantum systems, with emphasis on the physics and modeling aspects of these systems. \\

\end{abstract}

\tableofcontents

\clearpage
\section{Theory of Classical Stochastic Processes}

To understand quantum mechanics, it is important to first have a firm grasp of the classical world and its laws, in particular how randomness comes into play in classical systems. This section gives an overview (and a bit more) of open classical systems and their stochastic modeling, keeping in mind that the ideas for these classical systems can often be extended to study the quantum ones.  

\subsection{Motivation from Non-Equilibrium Statistical Mechanics} \label{sect_introocs}

A central topic in non-equilibrium statistical mechanics is the study of open classical systems. An open system is a component of a larger closed system and interacts with other components of the larger system. We first briefly review the mechanics of a closed classical system, in the Hamiltonian formulation. The {\it states} of the system are specified by points in a {\it phase space} (or state space), which can be a finite or infinite dimensional manifold. Let us assume for now that its phase space is finite dimensional and is taken to be the $2n$-dimensional manifold, $P = O \times \RR^n$, where $n$ is an integer and $O \subset \RR^n$ is an open set.  In this case, points in the phase space can be represented as $(\vecc{x}, \vecc{p})$, where the point $\vecc{x} \subset O$ describes the configuration of the various objects in the system and the point $\vecc{p} \subset \RR^n$ describes the momenta of the various objects in the system. 

The time evolution of the system is specified by a {\it trajectory}, a function $\vecc{z}(t) := (\vecc{x}(t), \vecc{p}(t))$ from an interval in $\RR$ to $P$. The fundamental dynamical principle says that the allowed trajectories are determined, uniquely in terms of the initial condition $\vecc{z}(0)$, by the solutions of the well-known {\it Hamilton’s equations}. The Hamilton's equations is a system of $2n$ ordinary differential equations of the form
\begin{equation}
\frac{d x_i}{dt} = \frac{\partial H}{\partial p_i}, \ \ \  \frac{d p_i}{dt} = -\frac{\partial H}{\partial x_i}, 
\end{equation}
for $i=1,\dots, n$, where $H = H(\vecc{x}, \vecc{p})$ is a function on the phase space called the {\it Hamiltonian}, whose choice depending on the system one wishes to describe. The Hamiltonian can be interpreted as the total energy on the phase space and the trajectory defines a {\it flow} $\vecc{z}(t) = \vecc{\Phi}(t) \vecc{z}(0)$ on the phase space that leaves invariant the total energy $H(\vecc{z}(t))$ of the system, conveying conservation of energy of the system. 

Often times, one would like to have a dynamical description for $F(\vecc{z})$, an arbitrary smooth function on the phase space, called a {\it classical observable}. If $\vecc{z}(t)$ satisfies the Hamilton's equations, then $F(t):=F(\vecc{z}(t))$ evolves according to: 
\begin{equation}
\frac{dF}{dt} = \{ F, H \},
\end{equation}
where the right hand side above denotes the Poisson bracket of $F$ and $H$, the function evaluated at $(\vecc{x}(t),\vecc{p}(t))$ as follows:
\begin{equation}
\{ F, H \} := \sum_{i=1}^n \frac{\partial F}{\partial x_i} \frac{\partial H}{\partial p_i} -  \frac{\partial F}{\partial p_i} \frac{\partial H}{\partial x_i}.
\end{equation}
In particular, we have $\{x_i, x_j \} = \{p_i, p_j\}=0$ and $\{x_i, p_j\} = \delta_{ij}$, where $\delta_{ij}$ denotes the Kronecker delta. For properties of the Poisson bracket and details on structures and symmetries in Hamiltonian dynamics, see \cite{marsden1995introduction}.

Generally, a system has, possibly infinitely, many states and one usually does not have a priori knowledge about all the initial values of these states\footnote{Typically one deals with systems containing Avogadro's number $N_A = 6 \times 10^{23}$ of particles. Interestingly, this number was computed by Perrin experimentally in his attempt to test Einstein's explanation of Brownian motion in terms of atoms and thereby confirming Dalton's atomic theory of matter. Perrin received a Nobel prize in 1926 for this.}. Treating the system in a statistical manner allows us to bring the problem into a tractable one. An important notion in statistical mechanics is that of {\it thermal equilibrium}, which, roughly speaking, is a condition that the states of the system be statistically distributed according to a probability distribution on the phase space that is invariant for the Hamiltonian flow $\Phi(t)$. In particular, the distribution of the state variable $\vecc{z}$ at time zero remains the same for all times. The introduction of such invariant probability measure randomizes the phase space, allowing one to simplify the description of the system at the cost of introducing uncertainty.

It is well known that in a finite dimensional phase space, any absolutely continuous $\Phi(t)$-invariant probability measure admits a Boltzmann-Gibbs type density, $\rho(\vecc{z}) = e^{-\beta H(z)}/Z$, where $Z$ is the normalization constant (partition function) and $\beta = 1/(k_B T) > 0$, $k_B$ is the Boltzmann's constant and $T>0$ is the temperature. As probabilistic framework takes its place in the thermal equilibrium setting, a natural question is whether a sufficiently smooth function on the phase space can be regarded as a stochastic process. It turns out that this is only true when the underlying phase space is infinite dimensional and, moreover, the Hamiltonian must have continuous spectrum \citep{picci1992generation}. In this case, the process is stationary and Gaussian (if the system is linear). We will keep in mind the above result when building a stochastic model for an open system later. The above question is related to the problem of stochastic aggregation of \cite{picci1989stochastic,picci1988hamiltonian}, where it was found that any stationary Gaussian process with a rational spectral density can be represented as the output of a linear infinite dimensional Hamiltonian system in a thermal equilibrium.

\subsection{A Hamiltonian Model for Open Classical Systems}
\label{sect_hammodelocs}

We consider a toy model describing a Brownian particle in contact with a heat bath, which is initially in a thermal equilibrium. The particle is modeled as a Hamiltonian system and it moves in a potential $U$. The heat bath is modeled as a system of non-interacting harmonic oscillators whose initial energy is distributed according to the Gibbs distribution at temperature $T$. The Brownian particle is coupled to each harmonic oscillator in the bath. This model is used widely to study many systems in statistical physics \citep{ford1965statistical, mori1965transport}. Our goal  is to derive, at a formal level, a stochastic integro-differential equation (SIDE) for the position and momentum variable of the particle from a specified Hamiltonian. This derivation serves to motivate the class of SIDEs that we are studying in this subsection. We emphasize that our derivation here is certainly not original and follows closely that in \cite{hanggi1997generalized} (see also an abstract approach in \cite{Zwanzig1973}). In fact, the derivation follows along the line of the program of Gibbs \citep{ford1965statistical}. 

One approach to derive the equations is to assume first that there are finitely many harmonic oscillators in the bath (Kac-Zwanzig model \citep{zwanzig2001nonequilibrium,ariel2008strong}). We then takes the thermodynamic limit by sending the number of oscillators to infinity in the resulting equations (replacing finite sum over oscillator frequencies by an integral), arguing that the set of frequencies must be dense to allow dissipation of energy from the system to the bath and to eliminate Poincar\'e recurrence. Another approach, which is more technically involved, is to replace the finite system of oscillator equations by an infinite system modeled by a wave equation \citep{rey2006open,pavliotis2014stochastic} (see Remark \ref{wave_eqn_model}). See also \cite{leimkuhler2018ergodic,chorin2009stochastic} and the references therein for an approach based on the Mori-Zwanzig projection formalism. We will derive the SIDEs by adopting the former approach in the multi-dimensional case.

We consider the situation where the coupling is nonlinear in the particle's position  and linear in the bath variables. Let $\vecc{\hat{x}} = (\vecc{x}, \vecc{x}_{1}, \dots, \vecc{x}_{N}) \in \RR^{d+d_1+\dots+d_N}$ and $\vecc{\hat{p}} = (\vecc{p}, \vecc{p}_{1}, \dots, \vecc{p}_{N}) \in \RR^{d+d_1+\dots+d_N}$. Hereafter, the subscript $^*$ denotes transposition and $|\vecc{b}|^2 := \vecc{b}^* \vecc{b} = \sum_{k=1}^{n} b_k^2$ denotes the norm of vector $\vecc{b} := (b_1, \dots, b_n) \in \RR^n$. 

The Hamiltonian of the system plus bath is:
\begin{equation}H(\vecc{\hat{x}},\vecc{\hat{p}}) = \frac{|\vecc{p}|^2}{2m} + U(\vecc{x}) + \sum_{k=1}^{N} \left( \frac{|\vecc{p}_{k}|^2}{2} + \frac{1}{2} \omega_{k}^2 \left| \vecc{x}_{k}-\frac{\vecc{c}^*_{k}}{\omega_{k}^2}\vecc{f}(\vecc{x}) \right|^2 \right) , \end{equation}
where $m$ is the particle's mass, $\vecc{x} \in \RR^{d}$ and $\vecc{p} \in \RR^{d}$ are respectively the particle's position  and  momentum, $\vecc{x}_{k} \in \RR^{d_k}$, $\vecc{p}_{k} \in \RR^{d_k}$  and $\omega_{k} \in \RR^{+}$ $(k=1,\dots,N)$ are respectively the position, momentum and frequency of the $k$th bath oscillator, $\vecc{f}(\vecc{x}) := (f_{1}(\vecc{x}), \dots, f_{r}(\vecc{x})) \in \RR^r$ is a vector function of $\vecc{x} := (x^{(1)}, \dots, x^{(d)})$ and the $\vecc{c}_{k} \in \RR^{r \times d_k}$ (so $\vecc{c}_k^* \in \RR^{d_k \times r}$) are  coupling matrices that specify the coupling strength between the system and the $k$th bath oscillator. We assume each $f_{j}$ $(j=1,\dots,r)$ is continuously differentiable in $x^{(k)}$, for every $k=1,\dots,d$.

To derive an equation for the particle's position and momentum, we start by plugging the expression for $H(\hat{\vecc{\vecc{x}}}, \hat{\vecc{\vecc{p}}})$ into the Hamilton's equations to obtain: 
\begin{align}
\dot{\vecc{x}} &= \frac{\vecc{p}}{m}, \\ 
\dot{\vecc{p}} &= -\vecc{\nabla}_{\vecc{x}} U(\vecc{x})  + \vecc{g}(\vecc{x}) \sum_{k} \vecc{c}_{k} \left(\vecc{x}_{k}-\frac{\vecc{c}^*_{k}}{ \omega_{k}^2} \vecc{f}(\vecc{x}) \right), \label{momeq}\\
\vecc{x}_{k} &= \vecc{p}_{k}, \ \ \ k = 1, \dots, N,\\
\vecc{p}_{k} &= -\omega_{k}^2 \vecc{x}_{k} + \vecc{c}^*_{k} \vecc{f}(\vecc{x}), \ \ \ k = 1,\dots,N,
\end{align}
where $\vecc{g}(\vecc{x}) \in \RR^{d \times r}$ denotes the Jacobian matrix $\left(\frac{\partial f_{i}}{\partial x^{(j)}} \right)_{ij}.$

Next, we eliminate the bath variables $\vecc{x}_{k}, \vecc{p}_{k}$, $k=1,\dots,N$, from the system's dynamics. Solving for $\vecc{x}_{k}(t)$ in terms of $\vecc{x}(t)$:
\begin{equation}\vecc{x}_{k}(t) = \vecc{x}_{k}(0)\cos(\omega_{k}t) + \frac{\vecc{p}_{k}(0)}{\omega_{k}} \sin(\omega_{k}t) + \frac{\vecc{c}^*_{k}}{ \omega_{k}} \int_{0}^{t} \sin(\omega_{k}(t-s)) \vecc{f}(\vecc{x}(s)) ds. \end{equation}

Plugging this into $\eqref{momeq}$, we obtain:
\begin{align}
\dot{\vecc{p}}(t) &= -\vecc{\nabla}_{\vecc{x}} U(\vecc{x}(t))  + \vecc{g}(\vecc{x}(t)) \sum_{k} \frac{\vecc{c}_{k} \vecc{c}_{k}^{*}}{ \omega_{k}^2} \left( \int_{0}^{t}  \omega_{k} \sin(\omega_{k}(t-s)) \vecc{f}(\vecc{x}(s)) ds  - \vecc{f}(\vecc{x}(t)) \right) \nonumber \\ 
&\ \ \ \  \ + \vecc{g}(\vecc{x}(t)) \vecc{F}(t), \end{align}
where \begin{equation}\vecc{F}(t) = \sum_{k} \vecc{c}_{k} \left(\vecc{x}_{k}(0) \cos(\vecc{w}_{k}t) + \frac{\vecc{p}_{k}(0)}{ \omega_{k}} \sin(\omega_{k}t) \right).\end{equation}

In the integral term above, we integrate by parts to obtain:
\begin{equation} \int_{0}^{t}  \omega_{k} \sin(\omega_{k}(t-s)) \vecc{f}(\vecc{x}(s)) ds = \vecc{f}(\vecc{x}(t)) - \cos(\omega_{k}t) \vecc{f}(\vecc{x}(0)) - \int_{0}^{t} \cos(\omega_{k}(t-s)) \vecc{g}^*(\vecc{x}(s)) \dot{\vecc{x}}(s) ds.\end{equation}

Using this, the equation for $\vecc{p}(t)$ becomes the {\it generalized Langevin equation (GLE)}:
\begin{equation} \label{deri_side}
\dot{\vecc{p}}(t) = -\vecc{\nabla}_{\vecc{x}} U(\vecc{x}(t))  - \vecc{g}(\vecc{x}(t))  \int_{0}^{t}  \vecc{\kappa}(t-s) \vecc{g}^*(\vecc{x}(s)) \dot{\vecc{x}}(s) ds    +  \vecc{g}(\vecc{x}(t)) \vecc{\xi}(t),\end{equation}
where 
\begin{equation} \label{memory_dis}
\vecc{\kappa}(t) = \sum_{k} \frac{\vecc{c}_{k} \vecc{c}_{k}^{*}}{ \omega_{k}^2} \cos(\omega_{k}t) \in \RR^{r \times r} \end{equation} 
and 
\begin{equation} \label{noise_dis}
\vecc{\xi}(t) = \vecc{F}(t) - \vecc{\kappa}(t) \vecc{f}(\vecc{x}(0)) =  \sum_{k} \vecc{c}_{k} \left( \left( \vecc{x}_{k}(0) - \frac{\vecc{c}^*_{k}}{\omega_{k}^2} \vecc{f}(\vecc{x}(0)) \right) \cos(\omega_{k} t) + \frac{\vecc{p}_{k}(0)}{ \omega_{k}} \sin(\omega_{k}t) \right). \end{equation}

Note that $\vecc{\xi}(t) \in \RR^r$ is expressed in terms of the initial values of the variables $\vecc{x}'_{k}(0) := \vecc{x}_{k}(0) - \frac{\vecc{c}^*_{k}}{\omega_{k}^2} \vecc{f}(\vecc{x}(0)) \in \RR^{p} $ and $\vecc{p}_{k}(0)\in \RR^{p}$. If all these initial values are known, then $\vecc{\xi}(t)$ is a deterministic force. However,  one rarely has a complete information about these initial values and this is where the introduction of randomness can help to simplify the model. In view of this, we assume that the variables $\vecc{x}'_k(0)$ and $\vecc{p}_k(0)$ are random and are distributed according to a Gibbs measure, with the density:
\begin{equation}\rho((\vecc{x}_{k}, \vecc{p}_{k}) \ | \ \vecc{x}(0) = \vecc{x})  = Z^{-1} \exp{\left(-\beta \left(\sum_{k=1}^{N} \frac{|\vecc{p}_{k}|^2}{2} + \frac{1}{2} \omega_{k}^2 \left| \vecc{x}_{k}-\frac{\vecc{c}^*_{k}}{\omega_{k}^2}\vecc{f}(\vecc{x}) \right|^2 \right) \right)},\end{equation} where $\beta = 1/(k_{B}T)$ and $Z$ is the partition function. Taking the averages of the bath variables with respect to the above density:
\begin{align}
&E_{\rho}\left[ \vecc{x}'_{k}(0)  \ |  \ \vecc{x}(0) = \vecc{x} \right] = 0, \ \ \  \ E_{\rho}[ \vecc{p}_{k}(0) \ | \ \vecc{x}(0) = \vecc{x} ] = 0, \\
&E_{\rho} [ \vecc{x}'_k(0) (\vecc{x}'_{k}(0))^* \ | \ \vecc{x}(0) = \vecc{x} ] = \frac{k_{B}T}{\omega_{k}^2} \vecc{I}, \ \  \ \ \ E_{\rho}[ (\vecc{p}_{k}(0) (\vecc{p}_k(0))^* \ | \ \vecc{x}(0) = \vecc{x} ] = k_{B}T \vecc{I}, 
\end{align}
where $E_{\rho}$ denotes mathematical expectation with respect to $\rho$ and $\vecc{I} \in \RR^{p \times p}$ is identity matrix. 

Note that $\vecc{\xi}(t)$ is a stationary, Gaussian process, if it is conditionally averaged with respect to $\rho$ \citep{zwanzig2001nonequilibrium}. It follows from this distribution of the bath variables that we have the  {\it fluctuation-dissipation relation}:
\begin{equation}E_{\rho}[\vecc{\xi}(t)] = 0, \ \ E_{\rho}[\vecc{\xi}(t)\vecc{\xi}(s)^{*}]=k_{B}T\vecc{\kappa}(t-s),\end{equation} 
where $\vecc{\kappa}(t-s)$ is the memory kernel whose formula is given in \eqref{memory_dis}. Later, we will generalize the resulting covariance of the process $\vecc{\xi}(t)$ to an integral expression. 
We remark that the memory function $\vecc{\kappa}(t)$ and  the ``color" of the noise $\vecc{\xi}(t)$ are determined by the bath spectrum and the system-bath coupling. 

Now we pass to the continuum limit by replacing the sum over $k$ in $\vecc{\kappa}(t)$ by an integral $\int_{\RR^{+}} d\omega n(\omega)$, where $n(\omega)$ is a density of states.  Then, if the $\vecc{c}_k$ are replaced by $\vecc{c}(\omega) \in \RR^{r \times p}$, the memory function $\vecc{\kappa}(t)$ becomes the function:
\begin{equation} \label{memory_cont}
\vecc{\kappa}(t) = \int_{\RR^{+}} d\omega n(\omega) \frac{\vecc{c}(\omega) \vecc{c}(\omega)^{*}}{\omega^2} \cos(\omega t), \end{equation}
where $\hat{\vecc{\kappa}}_c(\omega) :=  n(\omega) \vecc{c}(\omega) \vecc{c}(\omega)^*/\omega^2 \in L^1(\RR^+)$.  Repeating the same procedure for the noise process and also replacing the $\vecc{x}'_k(0)$ and $\vecc{p}_k(0)$ by $\vecc{x}'(\omega)$ and $\vecc{p}(\omega)$ respectively, $\vecc{\xi}(t)$ becomes:
 \begin{equation} \label{limit_noise}
 \vecc{\xi}(t) = \int_{\RR^+} d\omega n(\omega) \vecc{c}(\omega) \left(  \vecc{x}'(\omega)  \cos(\omega t) + \frac{\vecc{p}(\omega)}{ \omega} \sin(\omega t) \right). \end{equation}

The choice of the $n(\omega)$ and $\vecc{c}(\omega)$ specifies the memory function and therefore (by the fluctuation-dissipation relation) the statistical properties of the noise process. We write $\vecc{\kappa}(t)$ as an inverse Fourier transform of a measure:
\begin{equation} \label{memory}
\vecc{\kappa}(t) = \frac{1}{2\pi} \int_{\RR} \vecc{S}(\omega) e^{i \omega t} d\omega,  
\end{equation}
where the measure is absolutely continuous with respect to the Lebesgue measure, with the density
$\vecc{S}(\omega) = \pi \hat{\vecc{\kappa}}_c(\omega) \geq 0$. The density $\vecc{S}(\omega)$ is known as the {\it spectral density} of the bath.

\begin{rmk} \label{wave_eqn_model}
As the heat bath has infinitely many degrees of freedom, it is physically more correct to derive the SIDEs starting from a an infinite dimensional Hamiltonian with continuous spectrum (as discussed in the introduction). We briefly mention how this can be done at the cost of higher level of difficulties. We restrict to one dimension for simplicity. This discussion will be useful when we consider a quantum version of open system. For details, we refer to Chapter 2 in \cite{attal2006open}. The model for the heat bath, in accordance to the classical field theory, is modeled by the wave equation in $\RR$:
\begin{equation} \label{wave_eqn}
\frac{\partial^2}{\partial t^2} \varphi_t(\omega) = \Delta_{\omega} \varphi_t(\omega),  
\end{equation}
where $\omega \in \Omega =  \RR$, $t \in \RR$ and $\Delta_{\omega}$ denotes the Laplacian. The wave equation is a second-order equation which can be written as the following system of first-order equations:
\begin{align}
d \varphi_t(\omega) &= \pi_t(\omega) dt, \\
d \pi_t(\omega) &=  \Delta_{\omega} \varphi_t(\omega) dt, 
\end{align}
with initial conditions $\varphi_0(\omega) = \varphi(\omega)$ and $\pi_0(\omega) = \pi(\omega) $ to be specified. These are the Hamiltonian equations of motion whose dynamics are specified by the Hamiltonian function:
\begin{equation}
H_B(\varphi, \pi) = \frac{1}{2} \int_{\Omega} (|\nabla_{\omega} \varphi(\omega)|^2 + |\pi(\omega)|^2)  d\omega.
\end{equation}
Note that in the above, we have taken the frequency representation for modeling the environment, in contrast to the dual representation (via Fourier transformation) in spatial domain as adopted in Chapter 2 of \cite{attal2006open}.  The total system of the particle plus its environment is described by the Hamiltonian:
\begin{equation} \label{Hamiltonian_cos}
H = H_S + H_B + H_I,
\end{equation}
where $H_I$ describes the interaction between the particle and its environment and is assumed to be of the following dipole form:
\begin{equation} 
H_I =  -f(\vecc{x}) \int_{\vecc{\Omega}} \nabla_{\omega} \varphi(\omega) c(\omega) d\omega,  
\end{equation}
where $c(\omega)$ is a coupling function and $f(\vecc{x})$ is a (generally nonlinear) function of the particle's position $\vecc{x} \in \RR^d$.

So far we have not specified the initial conditions for the above model for the total system. Our goal is to model a situation where the particle not only dissipates energy into the environment but also ``fluctuates'' and eventually its dynamics relax to a ``stationary state''. To allow this, we assume that the model is initially in thermal equilibrium at positive temperature $T > 0$, i.e. the initial conditions, $\{(\varphi(\omega), \pi(\omega)): \omega \in \Omega\}$, of the wave equations are distributed according to a Gibbs measure at this temperature. In this case, the environment has infinite energy (or heat capacity) and therefore provides sufficient energy to achieve our goal. In fact, one expects ``return to equilibrium'' for the system, i.e. an initial distribution of the system will converge to a stationary state which is given by the Gibbs distribution, $Z^{-1} e^{-\beta H_S(\vecc{x},\vecc{p})} d\vecc{x} d\vecc{p}$, where $Z$ is the partition function and $\beta = 1/k_B T$. This introduces randomness\footnote{Our notation of $\Omega$ resembles that for a sample space of elementary events in probability theory but this is actually not a deliberate choice of ours here!} into the otherwise deterministic Hamiltonian dynamics. 

Now, technical difficulties arise when one tries to extend the Gibbs distribution to the infinite-dimensional setting as the following expression for the Gaussian measure, 
\begin{equation}
\mu_{\beta}(d\pi d\varphi) = Z^{-1} e^{-\beta H_B(\varphi, \pi)} \prod_{\omega \in \RR} d\pi d\varphi,
\end{equation}
does not define a Gaussian density with respect to the Lebesgue measure (which does not exist in infinite dimensions!).  One has to deal with this using the theory of Gaussian measures in Hilbert spaces (see \cite{bogachev1998gaussian, janson1997gaussian}). Once these are dealt with carefully, one can then show that the dynamics specified by the Hamiltonian above are equivalent to those by a SIDE of the form \eqref{deri_side}. 
\end{rmk}


\begin{rmk}
An alternative approach, natural from the approximation theory and as outlined before Remark \ref{wave_eqn_model}, is to justify rigorously the passage from \eqref{memory_dis} and \eqref{noise_dis} to \eqref{memory_cont} and \eqref{limit_noise} respectively. In our context, such limiting procedure constructs stationary Gaussian processes from deterministic ODEs with random initial conditions. Such construction can indeed be done and justified rigorously by adapting techniques from \cite{tupper2002constructing,kupferman2002long,kupferman2004fractional,ariel2008strong}, in which one of the key techniques is the Skorokhod embedding theorem \citep{gikhman1996introduction}. Other approximation methods such as Monte Carlo approximation of \eqref{memory_cont} can also be considered \citep{huisinga2003extracting}. See also \cite{kupferman2004fitting,melbourne2011note} for related results concerning derivation of stochastic models from deterministic ones and \cite{stuart1999analysis,hald2002asymptotic} for numerical experiments. Note that, in this approach \eqref{noise_dis} is viewed as a random trigonometric series\footnote{This viewpoint is an old one, dated as early as the days of Paley, Wiener and Zygmund. We refer to \cite{kahane1968some,filipsmooth,kahane1997century} for review, history and connection to other areas of mathematics.}, with the randomness coming from the initial variables $\vecc{x}_k'(0)$ and $\vecc{p}_k(0)$. 
\end{rmk}

\subsection{SDE Representation and Quasi-Markov Stationary Gaussian Processses} \label{sect_sderep}

In this section, we define a class of stationary Gaussian processes \citep{doob1953stochastic,yaglom1952introduction,dym2008gaussian} known as the quasi-Markov processes  and characterize these processes in terms of Markovian representations in the form of It\^o  type stochastic differential equations (SDEs) on Euclidean space $\RR^n$ ($n$ is a positive integer) driven by Wiener processes\footnote{One could consider, for instance, SDEs on manifolds driven by continuous semimartingales. We will not treat classes of SDEs that are more general than the one defined here in this thesis.} . We assume familiarity\footnote{We refer to \cite{arnold2013stochastic, evans2012introduction,oksendal2003stochastic}  for mathematically rigorous introductions and \cite{gardiner2009stochastic,schuss1988stochastic} for applications.} with basic concepts in the theory of probability and stochastic processes in the following.

Fix a probability space $(\Omega, \mathcal{F}, P)$ and denotes expectation with respect to $P$ by $\mathbb{E}$.

\begin{defn} \label{def_sderep}
We say that a stochastic process, $\{\vecc{y}_t \in \RR^m;  t \in [0,T] \}$, admits a {\it SDE representation} if it can be represented as $\vecc{y}_t = \vecc{c}(\vecc{x}_t)$, where $\vecc{x}_t \in \RR^n$ is the solution of the following {\it It\^o SDE}: 
\begin{equation} \label{sde_repr}
d\vecc{x}_t = \vecc{a}(t,\vecc{x}_t) dt + \vecc{b}(t,\vecc{x}_t) d\vecc{W}_t,
\end{equation} 
with initial condition $\vecc{x}_0 = \vecc{x}$. In the above, $\vecc{a}: [0,T] \times \RR^n \to \RR^n$, $\vecc{b}: [0,T] \times \RR^n \to \RR^{n \times r}$,  $\vecc{c}: \RR^n \to \RR^m$ are  measurable functions, and $\{\vecc{W}_t \in \RR^r: t \geq 0\}$ is a {\it Wiener process}, i.e. a family of random variables which is  Gaussian with mean $\mathbb{E} \vecc{W}_{t} = 0$ and has covariance $\mathbb{E} \vecc{W}_{t}\vecc{W}^*_{s} = \min(t,s) \vecc{I}$, where $\vecc{I} \in \RR^{r \times r}$ is identity matrix.  The initial condition $\vecc{x}$ can be either deterministic or a random variable that is independent of the Wiener process.
\end{defn}

The formal derivative of the Wiener process, $\vecc{\zeta}_t := \frac{d \vecc{W}_t}{dt}$, is a {\it white noise process}, i.e. a (generalized\footnote{See \cite{ito1954stationary,gel2014generalized,hida2008lectures} for alternative approaches to study white noise. Note that the SDE is written in a differential form due to the generalized nature of the white noise.}) mean-zero Gaussian vector-valued stochastic process with correlation function
$\mathbb{E} \vecc{\zeta}_t \vecc{\zeta}_s^* = \delta(t-s) \vecc{I}.$ It serves as an idealized model for a random disturbance imparted to the otherwise deterministic ordinary differential equation, $d\vecc{x}_t = \vecc{a}(t,\vecc{x}_t) dt$. We call $\vecc{a}$ the {\it drift coefficient} and the amplitude of the noise, $\vecc{b}$, the {\it noise coefficient} or {\it diffusion coefficient} of the SDE. In the context of stochastic modeling, it is useful to distinguish two different natures of the driving noise in the SDE. If $\vecc{b}$ is a constant, we will say that the noise is additive. If $\vecc{b}$ depends on the state $\vecc{x}_t$ of the system, we will say that the noise is multiplicative.

The correct interpretation for the SDE \eqref{sde_repr} is as the stochastic integral equation:
\begin{equation} \label{stoch_int_eqn}
\vecc{x}_t = \vecc{x} + \int_0^t \vecc{a}(s,\vecc{x}_s)ds + \int_0^t \vecc{b}(s,\vecc{x}_s) d\vecc{W}_s,
\end{equation} almost surely (a.s.), where the last integral term above is an It\^o integral, as carefully\footnote{Recall that the Wiener paths $t \mapsto W_t(\omega)$ ($\omega \in \Omega$) are continuous but nowhere differentiable a.s.. Moreover, it is not of bounded
variation, so the integral cannot be defined as a Riemann-Stieltjes integral in a unique way. Different Riemann-Stieltjes approximations
lead to different stochastic integrals (the It\^o integrals and Stratonovich integrals are two important ones), which, in turn, lead to SDEs with different properties. This is an important lesson from the point of view of stochastic modeling. } defined in the It\^o theory of stochastic integration. A precise notion of solution to the SDE \eqref{def_sderep} involves the above integral interpretation, specification of appropriate classes of functions which $\vecc{a}$ and $\vecc{b}$ belong to, as well as desired properties of the solution $\vecc{x}_t$.  There are different notions of solution for SDE of type \eqref{def_sderep}.


In the Definition \ref{def_sderep}, the functions $\vecc{a}$ and $\vecc{b}$ are assumed to be belong to appropriate classes of functions  such that the SDE \eqref{sde_repr} has a pathwise unique solution up to the time $T$.  A general sufficient condition for existence and uniqueness of solution, up to its explosion time\footnote{i.e. the maximum stopping time up to which a solution of the SDE can be defined.}, of the SDE is, roughly speaking\footnote{These conditions depend on the notion of ``solution'' to the SDE that one introduces.}, when $\vecc{a}$ and $\vecc{b}$ are locally Lipschitz (see Theorem 1.1.8 in \cite{hsu2002stochastic} for precise formulation; see also \cite{ikeda2014stochastic,GeorgeLowtherBlog}). The global existence and uniquess result (i.e. for all time $t \geq 0$) can be obtained by imposing additional assumptions on $\vecc{a}$ and $\vecc{b}$. Typically, one additionally imposes a linear growth type condition (see \cite{mao2007stochastic}) or assumption on a Lyapunov function associated to the SDE (see Theorem 5.9 in \cite{bellet2006ergodic} or the text \cite{khasminskii2011stochastic}).

We now elaborate more on the above discussions through a simple mathematical formulation\footnote{At the introductory level of, for instance, \cite{evans2012introduction}.}. We emphasize that in the following we are not including various extensions of the formulation to keep the technicalities involved to a bare minimum. Let us start by defining It\^o stochastic integrals. For simplicity, we restrict to scalar processes. Generalization to vector-valued processes is straightforward. Let $b(t)$ be a random process on $[0,T]$ which is adapted to the filtration $\mathcal{F}_t = \sigma(\{W_s:  s \leq t\})$ generated by the Wiener process $\{W_s: s \leq t\}$, i.e. it is an $\mathcal{F}_t$-measurable function for all $t \in [0,T]$. Assume that $b(\cdot)$ is square-integrable, i.e. $\mathbb{E} \int_0^T |b(s)|^2 ds < \infty$. We define the stochastic integral:
\begin{equation}
B(t) = \int_0^t b(s) dW_s,
\end{equation}
as the $L_2(\Omega)$ limit of the Riemann sum approximation: 
\begin{equation}
B_a(t) = \lim_{N \to \infty} \sum_{n=0}^{N-1} b(\tau_n) (W(t_{n+1})-W(t_n)), 
\end{equation} 
where we have introduced a partition of the interval $[0,t]$ by letting $t_n = n \Delta t$, $n=0,1,\dots,N-1$ and $N \Delta t = t$, and for $a \in [0,1]$, $\tau_n := (1-a)t_n + a t_{n+1}$ for $n=0,1,\dots,N-1$.  The {\it It\^o stochastic integral} is defined as the integral $B_0(t)$, corresponding to the choice of $a=0$ in the formula for $B_a(t)$ above. It is a diffusion process (i.e. it is a Markov\footnote{Recall that an adapted process $x_t$ (with respect to  $\{\mathcal{F}_{t}: t \geq 0\}$) is {\it Markov} if for any $s, t \geq 0$ and any bounded continuous real-valued function $f$, $
\mathbb{E}[f(x_{s+t}) | \mathcal{F}_s] = E[f(x_{s+t}) | x_s] \ \text{ a.s..}$} process with a continuous sample path), satisfies the It\^o isometry: 
\begin{equation}
\mathbb{E} \left(\int_0^T b(s) dW_s \right)^2 = \int_0^T \mathbb{E}|b(s)|^2 ds,
\end{equation}
and has the quadratic variation:
\begin{equation}
\langle B \rangle_t := \lim_{|P| \to 0} \sum_{n=0}^{N-1} |B_{t_{n+1}}-B_{t_n}|^2 =  \int_0^t b(s)^2 ds,
\end{equation}
where $P$ ranges over the partition $\{t_0, \dots, t_{N-1}\}$ of $[0,t]$ and $|P|$ is the mesh of the partition $P$.

We now define a notion of solution to SDE \eqref{def_sderep} and provide a result on the existence and uniqueness of its solution.
 
\begin{defn} {\bf Strong solution of SDE.} A process $\vecc{x}_t$ with continuous paths defined on the probability space $(\Omega, \mathcal{F}, P)$ is called a {\it strong solution}\footnote{There is also a notion of weak solution. Throughout this paper, a solution to a SDE is meant in the strong sense.} to the SDE \eqref{def_sderep} if:
\begin{itemize}
\item[(i)]  $\vecc{x}_t$ is almost surely continuous and adapted to the filtration $\mathcal{F}_t$ generated by  the Wiener process $\{\vecc{W}_s, s \leq t\}$ and $\vecc{x}$ (independent of the Wiener process)\footnote{Condition (i) implies that $\vecc{x}_t$ is progressively measurable with respect to $\mathcal{F}_t$, so our definition here is a bit less general than the one in page 81-82 of \cite{evans2012introduction}.}; 
\item[(ii)] $\vecc{b}(\cdot, \vecc{x}_{\cdot}) \in L^1([0,T]; \RR^n)$ and $\vecc{\sigma}(\cdot, \vecc{x}_{\cdot}) \in L^2([0,T]; \RR^{n \times r})$ almost surely;
\item[(iii)] For every $t \geq 0$, the stochastic integral equation \eqref{stoch_int_eqn}, with $\vecc{x}_0 = \vecc{x}$, holds almost surely. 
\end{itemize}
\end{defn}

A simple approach to obtain result on existence and uniqueness of solution for SDEs is to impose the Lipshitz continuity assumption and a linear growth condition (which prevents the solution from exploding in finite time), familiar from the existence and uniqueness theory of ordinary differential equations.

\begin{thm} {\bf Global existence and uniqueness of solution to SDE.} \label{thm_exisuniq}
Assume that $\vecc{a}(\cdot, \cdot)$ and $\vecc{b}(\cdot, \cdot)$ satisfy the following (global) Lipschitz continuity assumption and linear growth condition, i.e. there exists a positive constant $C$ such that for all $\vecc{x}, \vecc{y} \in \RR^n$ and $t \in [0,T]$,  
\begin{equation}
|\vecc{a}(t,\vecc{x})-\vecc{a}(t,\vecc{y})| + |\vecc{b}(t,\vecc{x}) - \vecc{b}(t, \vecc{y})|_{F} \leq C|\vecc{x}-\vecc{y}|, 
\end{equation}
and 
\begin{equation}
|\vecc{a}(t,\vecc{x})| + |\vecc{b}(t,\vecc{x})|_F \leq C (1 + |\vecc{x}|),
\end{equation}
where $| \cdot |_F$ denotes Frobenius norm. Moreover, assume that $\vecc{x}_0 = \vecc{x}$ is a random variable independent of the Wiener process $\vecc{W}_t$ with finite second moment, $\mathbb{E} |\vecc{x}|^2 < \infty$. 

Then the SDE \eqref{def_sderep} has a unique (strong) solution $\vecc{x}_t$  with $\mathbb{E} \int_0^t |\vecc{x}_s|^2 ds < \infty$ for every $t > 0$. 
\end{thm}
\begin{proof}
See page 90-94 in \cite{evans2012introduction}.
\end{proof}

The uniqueness of the solution is in the sense that if there exists two solutions $\vecc{x}_t$ and $\vecc{y}_t$ satisfying the SDE, then $\vecc{x}_t = \vecc{y}_t$ for all $t$ almost surely. Note that when the above globally Lipschitz condition on the drift and noise coefficient holds, the linear growth condition above is equivalent to the condition that both $|\vecc{a}(t,\vecc{0})|$ and $|\vecc{b}(t,\vecc{0})|$ be bounded for every $t \geq 0$. It can also be shown that the solution $\vecc{x}_t$ of the SDE is a semimartingale and a diffusion process.

Usually one is interested in the case when the SDE is autonomous, i.e. $\vecc{a} = \vecc{a}(\vecc{x}_t)$ and $\vecc{b} = \vecc{b}(\vecc{x}_t)$ do not show explicit dependence on the time $t$. In this case, recall that $\vecc{x}_t$ is a diffusion process and there is an associated second-order differential operator, the infinitesimal generator of the process: 
\begin{equation} \label{markov_semigroup}
\mathcal{L} = \sum_{i=1}^n a_j(\vecc{x}) \frac{\partial}{\partial x_j} + \frac{1}{2} \sum_{i,j=1}^n \Sigma_{i,j} \frac{\partial^2}{\partial x_i \partial x_j} =: \vecc{a}(\vecc{x}) \cdot \vecc{\nabla} + \frac{1}{2} \vecc{\Sigma}(\vecc{x}) : \vecc{D}^2,
\end{equation}
where the $\Sigma_{i,j}$ denote the matrix elements of the matrix $\vecc{\Sigma}(\vecc{x}) = \vecc{b}(\vecc{x})\vecc{b}(\vecc{x})^{T}$, $\vecc{\nabla}$ denotes the gradient operator, $\vecc{D}^2$ denotes the Hessian matrix and $:$ denotes inner product between two matrices defined by $\vecc{A}: \vecc{B} = Tr(\vecc{A}^{T} \vecc{B}) = \sum_{i,j} A_{i,j} B_{i,j}$ (the superscript $^{T}$ denotes transposition. 

Of interest is a formula that allows us to compute the rate of change in time of a sufficiently nice function $F: [0,T] \times \RR^n \to \RR$ evaluated at the solution $\vecc{x}_t \in \RR^n$ of the SDE. This is an important result at the heart\footnote{There are many important applications of It\^o's formula. Perhaps an intriguing example is the computation of the stochastic integral $\int_0^t W_s dW_s \in \RR$. More generally, one compute $\int_0^t h_n(W_s, s) dW_s = h_{n+1}(W_t,t)$, where $h_n(x,s)$ $(n=0,1,\dots)$ is the $n$th Hermite polynomial. This result can also be derived using the exponential martingale (generating function of the Hermite polynomial), which has connections to the quantum stochastic calculus (see Example \ref{expmart_iso}).} of stochastic calculus.

\begin{thm} {\bf It\^o's formula\footnote{It\^o's formula also holds for stopping times.}.} Assume that the assumptions in Theorem \ref{thm_exisuniq} hold and let $\vecc{x}_t$ be the solution of the autonomous SDE \eqref{def_sderep}. Assume that $F \in C^{1,2}([0,T] \times \RR^n)$. Then the process $F(\vecc{x}_t)$ satisfies:
\begin{equation}
F(t, \vecc{x}_t) = F(\vecc{x}_0) + \int_0^t \frac{\partial F}{\partial s}(s, \vecc{x}_s) ds + \int_0^t \mathcal{L} F(s, \vecc{x}_s) ds + \int_0^t (\vecc{\nabla} F(s, \vecc{x}_s)) \cdot \vecc{b}(\vecc{x}_s) d\vecc{W}_s,
\end{equation}
where $\mathcal{L}$ is the infinitesimal generator for the process $\vecc{x}_t$,
or, in the differential form:
\begin{equation}
dF(t,\vecc{x}_t) = \frac{\partial F}{\partial t} dt + \sum_{i=1}^n  \frac{\partial F}{\partial x_i} dx_i + \frac{1}{2} \sum_{i,j=1}^n  \frac{\partial^2 F}{\partial x_i \partial x_j} dx_i dx_j, 
\end{equation}
where the convention $dW_i(t) dW_j(t) = \delta_{i,j}dt$, $dW_i(t) dt = 0$ for $i,j=1,\dots,n$ is used. 
\end{thm}
\begin{proof}
See page 78-79 in \cite{evans2012introduction}.
\end{proof}

Using the It\^o's formula, one can obtain a partial differential equation (PDE) describing the expectation of a functional, $u(\vecc{x},t) := \mathbb{E}[\phi(\vecc{x}_t) | \vecc{x}_0 = \vecc{x}]$, $\phi \in C^2(\RR^n)$, of the solution $\vecc{x}_t$ to the SDE. Indeed, applying the It\^o's formula to $u$, using the martingale property of stochastic integrals and differentiating with respect to time gives the {\it backward Kolmogorov equation}: 
\begin{equation}
\frac{\partial u}{\partial t} = \mathcal{L} u, \ \ \ \ \  u(\vecc{x},0) = \phi(\vecc{x}),
\end{equation} 
and its solution can be expressed via the semigroup generated by $\mathcal{L}$, i.e. $u(\vecc{x},t) = (e^{t \mathcal{L}} \phi)(\vecc{x})$. The adjoint equation of the backward Kolmogorov equation is the {\it forward Kolmogorov equation} ({\it Fokker-Planck equation}) for the probability density $\rho(\vecc{x},t) \in C^{2,1}(\RR^n \times (0,\infty))$ of the process $\vecc{x}_t$ (with initial density $\rho_0(\vecc{x})$):
\begin{equation}
\frac{\partial \rho}{\partial t} = \mathcal{L}^* \rho, \ \ \ \ \  \rho =  \rho_0 \text{ for } \vecc{x} \in \RR^n \times \{0\},
\end{equation}
where $\mathcal{L}^* \cdot = \vecc{\nabla} \cdot (-\vecc{a}(\vecc{x}) \cdot + \frac{1}{2} \vecc{\nabla} \cdot (\vecc{\Sigma}(\vecc{x}) \cdot))$ is the $L_2$-adjoint of the generator $\mathcal{L}$. 

In the context of stochastic modeling of the noise process $\vecc{\xi}(t)$ in \eqref{deri_side}, we are interested in real-valued processes that are mean-square continuous, mean-zero Gaussian and stationary ($\mathbb{E} \vecc{y}_t = \vecc{0}$ and $\mathbb{E} \vecc{y}_t^2 < \infty$) \citep{cramer2013stationary}. The Gaussianness and stationarity of $\vecc{y}_t$ should be inherited by the Markov process $\vecc{x}_t$. Therefore, we are led to consider linear time-invariant SDE representation of the type 
\begin{align}
d\vecc{x}_t &= \vecc{A}\vecc{x}_t dt + \vecc{B} d\vecc{W}_t, \label{q_m1} \\ 
\vecc{y}_t &= \vecc{C}\vecc{x}_t, \label{q_m2}
\end{align} 
where $\vecc{A} \in \RR^{n \times n}$, $\vecc{B} \in \RR^{n \times r}$, $\vecc{C} \in \RR^{m \times n}$ are constant matrices. Therefore, 
\begin{equation}
\vecc{y}_t = \vecc{C}e^{\vecc{A}t}\vecc{x} + \int_0^t \vecc{C}e^{\vecc{A}(t-s)} \vecc{B} d\vecc{W}_s, 
\end{equation}
where the initial time is taken to be $t = 0$, and we see that the SDE representation is a particular time representation of the process $\vecc{y}_t$. 

We assume that $\vecc{A}$ is {\it Hurwitz stable} (or $-\vecc{A}$ is {\it positive stable}), i.e. the real parts of all eigenvalues of $\vecc{A}$ are negative, and the initial condition $\vecc{x}$ is a mean-zero Gaussian random variable with covariance matrix $\vecc{M}$ satisfying the {\it Lyapunov equation}\footnote{By our assumption on $\vecc{A}$, there exists a unique solution to the Lyapunov equation. Furthermore, the solution is given by $\vecc{M} = \int_0^\infty e^{\vecc{A} y } \vecc{B} \vecc{B}^{T} e^{\vecc{A}^{T} y} dy$ (also known as the controllability gramian) \citep{bellman1997introduction}.}:
\begin{equation}
\vecc{A} \vecc{M} + \vecc{M} \vecc{A}^T = -\vecc{B}\vecc{B}^T.
\end{equation}
Then one computes that the covariance function of $\vecc{y}_t$ equals:
\begin{equation}
\vecc{R}(t-s) := \mathbb{E} \vecc{y}_t \vecc{y}^T_s = \vecc{C}e^{\vecc{A}(t-s)}\vecc{M} \vecc{C}^T,
\end{equation}
where $t > s \geq 0$. Note that since $\vecc{y}_t$ is stationary, the covariance function depends only on the time difference $\tau = t-s$. 

Denote by $\mathcal{G}$ the class of the real-valued, mean-square continuous, mean-zero stationary Gaussian processes. A subclass of $\mathcal{G}$ known as the quasi-Markov processes is of particular interest to us.

\begin{defn} \label{defn_qm}
Let $\{\vecc{y}_t \in \RR^m: t \in [0,T]\}$ be a stochastic process in $\mathcal{G}$. We say that $\vecc{y}_t$ is a {\it quasi-Markov process in $\mathcal{G}$} if it has a SDE representation of the form \eqref{q_m1}-\eqref{q_m2}, specified by the triple $(\vecc{A}, \vecc{B}, \vecc{C})$ of matrices of appropriate dimensions. Here $\vecc{A}$ is Hurwitz stable and the initial condition $\vecc{x}$ is a mean-zero Gaussian random variable with covariance matrix $\vecc{M}$ such that the Lyapunov equation 
$\vecc{A} \vecc{M} + \vecc{M} \vecc{A}^{T} = -\vecc{B} \vecc{B}^T$ is satisfied.  
\end{defn}

We now look at an equivalent representation in the frequency domain.  It is useful to view the Gaussian process, $t \to \vecc{y}_t = (y_t^1, \dots, y_t^m)$ ($t \in [0,T]$), as a curve in the real Hilbert space of $L^2(\Omega, \mathcal{F},P)$. In particular, all the probability information about $\vecc{y}_t$ is encoded in the subspace $\mathcal{H}(\vecc{y})  \subset L^2(\Omega, \mathcal{F},P)$, where $\mathcal{H}(\vecc{y})$ is the closed subspace spanned by the process $\vecc{y}_t$, i.e. $\mathcal{H}(\vecc{y}) = span\{y^k_t: t \in [0,T], \ k=1,\dots,m \}$. Recall that for a (mean-square continuous) stationary stochastic process, $t \mapsto \vecc{y}_t$, there exists a (strongly continuous) one-parameter group $\{\vecc{U}_t: t \in \RR\}$ of unitary operators on $\mathcal{H}(\vecc{y})$ such that for any time $t > 0$, one has $\vecc{y}_t = \vecc{U}_t \vecc{y}$. Note that the group preserves expectation. 

By the Stone-von Neumann theorem \citep{reed1972methods}, $\vecc{y}_t$ can be written uniquely as 
\begin{equation}
\vecc{U}_t \vecc{y} = \int_\RR e^{i \omega t} d\vecc{E}(\omega) \vecc{y},
\end{equation}
where $\vecc{E}(\cdot)$ is the spectral measure mapping Borel subsets of real line into orthogonal projection operators on $\mathcal{H}(\vecc{y})$ and $d\vecc{E}(\omega)\vecc{y}$ can be viewed as a stochastic measure. In this way, one obtains a spectral representation for the stationary Gaussian process $\vecc{y}_t$. More generally: 

\begin{thm}
Every stationary process $\{\vecc{y}_t: t \in \RR\}$, continuous in
mean-square, admits a representation
\begin{equation}
\vecc{y}_t = \int_\RR e^{i \omega t} d\hat{\vecc{y}}(\omega), \ \ \ t \in \RR,
\end{equation}
where $d\hat{\vecc{y}}$ is a finite vector-valued orthogonal stochastic measure uniquely determined by the process, and satisfies
\begin{equation}
\mathbb{E} d\hat{\vecc{y}}(\omega) = \vecc{0}, \ \ \ \mathbb{E} |d \hat{\vecc{y}}(\omega)|^2 = d\vecc{F}(\omega),
\end{equation}
where $\vecc{F}$ is the spectral distribution function of $(\vecc{y}_t)$.  The orthogonal stochastic measure $\hat{\vecc{y}}(\omega)$ is called the Fourier transform of the stationary process $\vecc{y}_t$. 
\end{thm}
\begin{proof} See Theorem 3.3.2 in \cite{lindquist2015linear}.  \end{proof}

In particular, consider the one-dimensional stationary linear process of form 
\begin{equation} \label{time_rep}
\xi_t = \int_\RR w(t-s) \eta(ds), 
\end{equation}
$t \in \RR$, where $\eta(dt)$ is the standard stochastic measure with orthogonal values on $t \in \RR$ such that $\mathbb{E} \eta(dt) = 0$, $\mathbb{E} |\eta(dt)|^2 = dt$, with the weight function $w(t)$ satisfying $\int_\RR |w(t)|^2 dt < \infty$ (for instance, $\xi_t$ is a stochastic integral). Then $\xi_t$ has a spectral representation:
\begin{equation} \label{spec_rep}
\xi_t = \int_\RR e^{i\omega t} \phi(\omega) d\hat{W}(\omega),
\end{equation}
where $d\hat{W}(\omega)$ is a stochastic measure with orthogonal values on $\RR$ such that $\mathbb{E} \hat{W}(\omega) = 0$ and $\mathbb{E} |d\hat{W}(\omega)|^2 = d\omega/(2 \pi)$, i.e.  $\hat{W}$ is the Fourier transform of the Wiener process $(W_t)$, and $\phi(\omega)$ is a non-random function expressible as the Fourier transform of the weight function $w(t)$:
\begin{equation} 
\phi(\omega) = \int_\RR e^{-i\omega t} w(t) dt.
\end{equation}
Equation \eqref{spec_rep} represents the harmonic oscillations $\phi(\omega) e^{i\omega t}$ of frequency $\omega$  and the spectral density $S(\omega) = |\phi(\omega)|^2$ characterizes the weight of the different harmonic components of the process depending on the frequency $\omega$ \citep{rozanov1987stationary}. 


One natural task is to characterize all  processes that admit a (finite-dimensional) SDE representation in terms of their statistical properties (i.e. their covariance function and spectral distribution function). 

\begin{thm} \label{TFAE} The following statements are equivalent. 
\begin{itemize}
\item[(i)] There exists finite-dimensional SDE representations of $\vecc{y}_t \in \mathcal{G}$; 
\item[(ii)] The spectral distribution function, $\vecc{F}$, of the process is absolutely continuous with a rational spectral density $\vecc{S}$, i.e. $\vecc{S}(\omega) = \frac{d}{d\omega} \vecc{F}(\omega)$; 
\item[(iii)] The covariance function, $\vecc{R}(t) = \frac{1}{2 \pi} \int_\RR \vecc{S}(\omega) e^{i \omega t} d\omega$, of the process is a {\it Bohl function}, i.e. its matrix elements are finite linear combination of products of an exponential, a polynomial, a sine or cosine function. 
\end{itemize}
\end{thm}
\begin{proof} To show (i) is equivalent to (ii), see Corollary 10.3.4 in \cite{lindquist2015linear}. To show (ii) is equivalent to (iii), apply Theorem 2.20 in \cite{trentelman2002control} to $\vecc{R}(t)$.
\end{proof}

In experimental situations, one typically only has spectral information  about a noise process; for instance its spectral density \citep{gittes1997signals}. It is then important to be able to model the noise process based on this information. The construction of a SDE representation (i.e. identification of the triple ($\vecc{A}, \vecc{B},\vecc{C})$) for a quasi-Markov process in $\mathcal{G}$ given its spectral density or covariance function is the problem of stochastic realization, which has interesting connections to the Lax-Phillips scattering theory and can be formulated in a coordinate-free approach (see the monograph \cite{lindquist2015linear}), that are worth mentioning. The following result solves the problem in our case (see \cite{lindquist2015linear,picci1992generation} for details).

\begin{alg}
A SDE representation of a quasi-Markov process $\vecc{y}_t$, given its spectral density $\vecc{S}$, can be computed via the following procedures:
\begin{itemize}
\item[(1)] Find a spectral factorization of $\vecc{S}(\omega)$, i.e. find a rational $m \times r$ matrix function $\vecc{\Phi}$ such that $\vecc{S}(\omega) = \vecc{\Phi}(i\omega) \vecc{\Phi}(i\omega)^*$, where $\vecc{\Phi}$ is an analytic spectral factor (i.e. all its poles lie in the left half plane) and $^*$ denotes conjugate transpose. For simplicity, restrict to left-invertible factors, with the rank of $\vecc{S}$ equals to $r$.
\item[(2)] For each such spectral factor $\vecc{\Phi}$, define the Gaussian process $\vecc{W}_t$ by specifying its Fourier transform $\hat{\vecc{W}}(\omega)$ as:  
\begin{equation}
d\hat{\vecc{W}}(\omega) = \vecc{\Phi}^{-L}(i\omega) d\hat{\vecc{y}}(\omega),
\end{equation}
 where $-L$ denotes left inverse. Then it is easy to see that $\vecc{W}_t$ is a $\RR^d$-valued Wiener process and $\vecc{y}_t$ admits the spectral representation:
 \begin{equation}
 \vecc{y}_t = \int_\RR e^{i \omega t} \vecc{\Phi}(i \omega) d\hat{\vecc{W}}(\omega).
 \end{equation}
 \item[(3)] Compute a minimal realization of the spectral factor $\vecc{\Phi}(i\omega)$ of the form
 \begin{equation}
\vecc{\Phi}(i\omega) = \vecc{C} (i \omega \vecc{I} - \vecc{A})^{-1} \vecc{B}, 
 \end{equation}
where $\vecc{A} \in \RR^{n \times n}$ is a Hurwitz stable matrix, $\vecc{B} \in \RR^{n \times r}$ and $\vecc{C} \in \RR^{m \times n}$ are constant matrices such that $\vecc{A}\vecc{M} + \vecc{M} \vecc{A}^T = -\vecc{B} \vecc{B}^T$ , with $n$ as small as possible. 
\end{itemize}
Therefore, corresponding to every spectral factor $\vecc{\Phi}$, $\vecc{y}_t$ admits a SDE representation of the form as defined in Definition \ref{defn_qm}. The representation obtained is unique up to a change of basis on the state space and an orthogonal transformation on the Wiener process $\vecc{W}_t$.  
\end{alg}

Now we return to our earlier discussions of open systems.  Observe that after taking the thermodynamic limit the noise process $\vecc{\xi}(t)$ in \eqref{limit_noise} can be seen to be already in a form of spectral representation, with the initial ``field variables'' $(\vecc{x}'(\omega), \vecc{p}(\omega))$ (conditionally) distributed according to a Gibbs measure. This justifies our stochastic modeling of the noise process. 

We now apply the above algorithm by factorizing $\vecc{S}(\omega) = \vecc{\Phi}(i\omega) \vecc{\Phi}^*(i\omega)$, where 
$$\vecc{\Phi}(i\omega) = \sqrt{\pi n(\omega)/\omega^2} \vecc{c}(i\omega)$$ is a spectral factor of the spectral density. The following examples give realization of a few noise proceses. We take $n(\omega) =  2\omega^2/\pi$ (Debye-type spectrum for phonon bath) in all these examples.

\begin{ex}
If we choose $\vecc{c}(\omega) \in \RR^{d \times d}$ to be a scalar multiple of the identity matrix $\vecc{I}$, then $\vecc{\kappa}(t)$ is proportional to $\delta(t) \vecc{I}$. This leads to a Langevin equation driven by white noise, in which the damping term is instantaneous. In this case, we have the SDE system for $(\vecc{x}_t, \vecc{v}_t) \in \RR^{d \times d}$: 
\begin{align}
d \vecc{x}_t &= \vecc{v}_t dt, \label{side_derived1} \\ 
m d \vecc{v}_t &= -\vecc{\nabla}_{\vecc{x}} U (\vecc{x}_t) dt - \vecc{g}^2(\vecc{x}_t) \vecc{v}_t dt + \vecc{g}(\vecc{x}_t) d\vecc{W}_t, \label{side_derived2}  
\end{align}
where $\vecc{W}_t$ is the Wiener process. 

The above SDE system can also be obtained as a Markovian limit of the GLE \citep{pavliotis2014stochastic}. 

In the special case where $\vecc{g}(\vecc{x}_t) = \vecc{g}$ is a constant matrix and $U(\vecc{x}) = \frac{1}{2} k \vecc{x}^2$ (harmonic potential) or $U(\vecc{x}) = \vecc{0}$ (free particle case), both the GLE and the SDE system \eqref{side_derived1}-\eqref{side_derived2} are exactly solvable. Interestingly, in this special case and in one dimension ($d=1$), the SDE system \eqref{side_derived1}-\eqref{side_derived2} can be derived from the Lamb's model \citep{lamb1900peculiarity} and constructed using a dilation procedure (see \cite{maassen1982class,lewis1984hamiltonian} for details and other Hamiltonian models for open systems).     
\end{ex}


\begin{ex} If we  choose $\vecc{c}(\omega) \in \RR^{d \times d}$ to be the diagonal matrix with the $k$th entry 
\begin{equation}
\frac{\alpha_k}{\sqrt{\alpha^2_k+\omega^2}} ,\end{equation} 
where the $\alpha_{k} > 0$, then we have:
\begin{equation} \vecc{\kappa}(t) = \vecc{A} e^{-\vecc{A}t}, \end{equation}
where $\vecc{A}$ is the constant diagonal matrix with the $k$th entry equal $\alpha_{k}$.  On the other hand, choosing $\vecc{c}(\omega)$ to be the diagonal matrix with the $k$th entry
\begin{equation}
\left(\frac{\omega_{kk}}{\tau_{kk}}\right)^2 \frac{1}{\sqrt{\omega^2 (\omega_{kk}^2/\tau_{kk})^2+(\omega^2-(\omega_{kk}/\tau_{kk})^2)^2}}\end{equation} allows us to obtain the covariance function of a harmonic noise process, where the $\omega_{kk}$ and $\tau_{kk}$ are the diagonal entries of the matrix $\vecc{\Omega}$ and $\vecc{\tau}$ respectively. In the general case where $\vecc{\kappa}(t)$  is written as $\vecc{C}_1 e^{-\vecc{\Gamma}_1 t}\vecc{M}_1 \vecc{C}_1^*$, one may take $\vecc{M}_1 = \vecc{I}$, $\vecc{\Gamma}_1$ to be positive definite, in which case the Lyapunov equation gives $\vecc{\Gamma}_1 = \vecc{\Sigma}_1 \vecc{\Sigma}_1^*/2$, and choose 
\begin{equation}
\vecc{c}(\omega) = \frac{1}{\sqrt{2}} \vecc{C}_1 (\vecc{\Gamma}^2_1 +  \omega^2 \vecc{I})^{-1/2} \vecc{\Sigma}_1.\\
\end{equation}
\end{ex}

To summarize, stochastic integro-differential equations (SIDEs) of the following form appear naturally from the studies of open classical systems:
\begin{equation} \label{gle_mostgeneral}
m \ddot{\vecc{x}}(t) = \vecc{F}(\vecc{x}(t)) - \vecc{g}(\vecc{x}(t))\int_0^t \vecc{\kappa}(t-s) \vecc{h}(\vecc{x}(s)) \dot{\vecc{x}}(s) ds + \vecc{\sigma}(\vecc{x}(t)) \vecc{\xi}(t),
\end{equation}
where $m>0$ is the mass of the particle, $\vecc{F}$ represents the external force,  $\vecc{\kappa}(t)$ is a memory function, $\vecc{\xi}(t)$ is a stationary Gaussian process, and $\vecc{g}$, $\vecc{h}$ and $\vecc{\sigma}$ are state-dependent coefficients.   

Let us consider the special case where the memory function and the covariance function of the driving noise are Bohl. An immediate consequence of Theorem \ref{TFAE} allows us to embed the resulting process $\vecc{x}(t)$, satisfying \eqref{gle_mostgeneral}, as a component of a higher dimensional process which admits a SDE representation. This approach makes available various tools and techniques from the Markov theory of stochastic processes and can be exploited to study homogenization of GLEs.

To end this section, we give a brief literature review on works related to the GLEs. A basic form of the GLEs \eqref{gle_mostgeneral} was first introduced by Mori in \cite{mori1965transport} and subsequently used to model many systems in statistical physics \citep{Kubo_fd,toda2012statistical,goychuk2012viscoelastic}. As remarked by van Kampen in \cite{van1998remarks}, ``Non-Markov is the rule, Markov is the exception". Therefore, it is not surprising that non-Markovian equations (including those of form \eqref{gle_mostgeneral}) find numerous applications and thus have been studied widely in the mathematical, physical and engineering literature (see \cite{luczka2005non,samorodnitsky1994stable} for surveys of non-Markovian processes). In particular, GLEs have been widely used as models to study many systems and have attracted increasing interest in recent years. We refer to, for instance, \cite{PhysRevB.89.134303,mckinley2009transient,lysy2016model,adelman1976generalized,siegle2010markovian,hartmann2011balanced,cordoba2012elimination,2016arXiv160602596L}  for various applications of GLEs and \cite{Ottobre,mckinley2017,GlattHoltz2018,leimkuhler2018ergodic,nguyen2018small, lim2019homogenization, lim2020homogenization}  for asymptotic analysis of GLEs.

\section{Mathematical Concepts and Formulations of Quantum Mechanics} \label{sect_concepts}
We now switch our attention to the quantum theory formulated in the language of quantum probability. Quantum probability is a version of noncommutative probability theory that not only extends Kolmogorov's classical probability theory, but also provides a natural framework to study quantum mechanical systems. In fact, its development was aided by statistical ideas and concepts from quantum theory. It is the foundation for construction of quantum stochastic calculus (QSC) and quantum stochastic differential equations (QSDEs), which extend classical stochastic calculus and SDEs. The basic rigorous construction of QSC and QSDEs was first laid out in the seminal\footnote{For an account of the developments that preceded the publication of this seminal paper, see \cite{applebaum2010robin}.}  work of Hudson and Parthasarathy (H-P) \citep{hudson1984quantum}.  Modeling quantum mechanics as a noncommutative probability theory is a fruitful mathematical approach.  In the last few decades, the quantum probability formalism has been widely applied to study open quantum systems.  On the other hand, the classical stochastic calculus of It\^o has deep connection with objects such as the Fock space and the Heisenberg uncertainty principle \citep{biane2010ito}. 


In this section, we give a quick overview of quantum probability. We follow closely the notations and expositions in \cite{parthasarathy2012introduction}. For comprehensive accounts of quantum probability, we refer to the monographs \cite{parthasarathy2012introduction} and \cite{meyer2006quantum}.  For recent developments, perspectives and applications of the calculus to the study of open quantum systems, we refer to \cite{Hudson1985,fagnola1999quantum,attal2006open, bouten2007introduction,nurdinlinear,emzir2016physical,barchielli2015quantum,gregoratti2001hamiltonian,gough2009quantum,sinha2007quantum}.  \\

\noindent {\bf Notation.}  $[A,B] := AB - BA$ and $\{A, B\} := AB + BA$ denote, respectively, the commutator and anti-commutator of the operators $A$ and $B$. The symbol $I$ denotes identity operator on an understood space. We denote by $\mathcal{B}(\mathcal{H})$ the algebra of all bounded operators on the Hilbert space $\mathcal{H}$, with the inner product $\langle \cdot | \cdot \rangle$, which is linear in the second argument and antilinear in the first. We are using Dirac's bra-ket notation, so we will write, for instance, the vector $u \in \mathcal{H}$ as the ket $|u\rangle$. For $X \in \mathcal{B}(\mathcal{H})$, $u, v \in \mathcal{H}$, we write $\langle u | X | v \rangle = \langle u| X v\rangle = \langle X^* u|v\rangle$.  We recall that unbounded operators are defined only on a linear manifold in $\mathcal{H}$ (the domain of the operators). Two unbounded operators $X$ and $Y$ are equal if their domains coincide and both of them agree on the common domain. The adjoint operator of the unbounded operator $X$ is denoted as $X^\dagger$ (whenever it exists). Any projections considered hereafter are orthogonal projections.

\subsection{Postulates of Quantum Mechanics}
The essence of quantum probability is best illustrated in line with the mathematical formulation of non-relativistic quantum mechanics, which is based on a set of commonly accepted postulates \citep{neumann1932mathematische}. We will do so in this subsection and make connection to relevant concepts in classical probability along the way. For simplicity, we focus on description of a single, isolated particle in the following. 



Roughly speaking, classical mechanics describes the dynamical state variables of a particle as functions of position and momentum on a phase space. Quantum mechanics describes the state of a particle by an abstract ``wave function'' obeying wave mechanics\footnote{The wave nature of the particle in the theory is consistent with the observation in the double-slit experiment -- see Chapter 1 in \cite{gustafson2011mathematical,sakurai2017modern} for a brief description.}.  More precisely, following closely the Dirac-von Neumann axioms\footnote{Here we are following the orthodox version of quantum mechanics. Interpretations of quantum mechanics belong to the foundations of quantum mechanics and we will not discuss them (however, see \cite{accardi2006could}
for a fun digression).}, quantum mechanics is formulated based on the following principles. 
\begin{itemize}
\item[(A1)] {\bf Spaces.} For every quantum system, there is an associated complex separable Hilbert space $\mathcal{H}$ (with an inner product $\langle \cdot | \cdot \rangle$) on which an algebra of linear operators, $\mathcal{A}$, is defined. 
\item[(A2)] {\bf States.} Given an algebra of operators $\mathcal{A}$ on $\mathcal{H}$ for a quantum system, the space of {\it quantum states}, $\mathcal{S}(\mathcal{A})$ of the  system consists of all positive trace class operators $\rho \in \mathcal{A}$ with unit trace, i.e. $Tr(\rho)=1$. The pure states are projection operators (rays) onto one-dimensional subspaces of $\mathcal{H}$, with $Tr(\rho^2)=1$. All other states, with $Tr(\rho^2) < 1$, are called mixed states.  For instance, if $|\vecc{u}\rangle$ is a unit vector in $\mathcal{H}$, then the density operator defined by $|u\rangle \langle u|$ is a pure state. For our convenience, we will also refer to $|u\rangle$ as the state. In finite-dimensional spaces, the general density operator representing a mixed state is a statistical mixture of pure states of the form $\rho = \sum_{j} p_j |u_j \rangle \langle u_j|$, where $\sum_j p_j = 1$.  
\item[(A3)] {\bf Observables.} An {\it observable} of the quantum system is represented by a self-adjoint linear (not necessarily bounded) operator, $X$, on $\mathcal{H}$. By  von Neumann's spectral theorem \citep{reed1981functional},  it admits a spectral representation,

\begin{equation}
X = \int_{\RR} x P^{X}(dx),
\end{equation}
where $P^X$ is a spectral measure on the Borel $\sigma$-algebra of $\RR$. For instance, for a self-adjoint operator $X \in \mathcal{B}(\mathcal{H})$ with $\mathcal{H}$ finite-dimensional, the spectral representation becomes $X = \sum_{\lambda \in Spec(X)} \lambda E_{\lambda},$ where the $\lambda$ are the eigenvalues of $X$ with $E_{\lambda}$ the orthogonal  projection on the corresponding eigenspace such that $E_{\lambda} E_{\lambda'} = 0$ for $\lambda \neq \lambda'$ and $\sum_{\lambda \in Spec(X)} E_{\lambda} = I$.
\item[(4)] {\bf Measurements and Statistics.} 
Let $X$ represent an observable, $|u \rangle$ be a state\footnote{Any (mixed) state can be purified, i.e. it can be written as a partial trace of a pure state on an enlarged Hilbert space. Purification is a central idea in the theory of quantum information \citep{nielsen2010quantum}.} and $E \subset \RR$ be a Borel subset. Then the spectral projection $P^X(E)$ is the {\it event} that the value of the observable $X$ lies in $E$ and the probability that the event $P^X(E)$ occurs in the state $|u\rangle$ is given by $\langle u | P^X(E)| u \rangle$. From this, one sees that the {\it probability measure} $\mu_X$ is given by
\begin{equation}
\mu_X(E) = \langle u | P^X(E) | u \rangle,
\end{equation}    
on the Borel $\sigma$-algebra of $\RR$. 
A process of {\it measurement}\footnote{For the theory of quantum measurement from a physicist's perspective, see \cite{wiseman2009quantum}.} on a quantum system is the correspondence between the observable-state pair $(X, \rho=|u \rangle \langle u|)$ and the probability measure $\mu_{X}$. In other words, for any Borel subset $E \in \mathcal{B}(\RR)$, the quantity $\mu_{X}(E) \in [0,1]$ is the probability that the result of the measurement of the observable $X$ belongs to $E$ when a quantum system is in the state $\rho$. We say that $\mu_X$ is the {\it distribution} of the observable $X$ in the state $|u\rangle$ and define the {\it expectation} of the observable $X$ in this state  by 
\begin{equation}
\langle X \rangle = Tr(X\rho) = \int_{\RR} x \mu_{X}(dx) = \langle u | X | u \rangle,
\end{equation}
whenever it is finite. 
If $f$ is a Borel function (real or complex-valued), then $f(X)$ is an observable with expectation $\langle f(X) \rangle = \int_\RR f(x) \mu_X(dx) = \langle u | f(X) | u \rangle$, if it is finite. An important example is  the {\it characteristic function} of $X$ in the state $|u\rangle$, defined as:
\begin{equation}
\langle e^{it X} \rangle = \int_\RR e^{itx} \mu_X(dx) = \langle u | e^{itX} |u \rangle.
\end{equation}
\end{itemize}

From the above discussion, one can therefore view a quantum state as quantum analogue of probability distribution in classical probability and a quantum observable as  quantum analogue of random variable (classical observable). There are important distinctions between these notions in the classical versus quantum case. In particular, the set of quantum observables generally forms a noncommutative algebra while the set of classical observables forms a commutative one. The noncommutativity of quantum observables leads to notable departure of quantum mechanics from its classical counterpart, in particular:

\begin{itemize}
\item Non-commuting operators do not, in general, admit a `'joint distribution''\footnote{Non-commuting observables cannot be simultaneously realized classically. Generally, there is no sensible notion of joint quantum probability distribution for them.} in a particular state \citep{BreuerBook}. 
\item Noncommutativity of operators also gives rise to interesting inequalities for statistical quantities of observables. We mention one such inequality in the following.


Define the {\it covariance} between $X,Y \in \mathcal{B}(\mathcal{H})$ in the state $\rho$ as: $$cov_{\rho}(X,Y) = Tr(\rho X^{*}Y)-Tr(\rho X^{*})Tr(\rho Y),$$ which might be a complex number if the two observables are non-commuting (interference) \citep{parthasarathy2012introduction}. 

\begin{prop}(Uncertainty Principle) Let $var_{\rho}(X) = cov_{\rho}(X,X)$, then for any pure state $|u \rangle$ and observables $X,Y$, we have the following inequality: 
$$var_{u}(X) var_{u}(Y) \geq \frac{1}{4}|\langle u |  i[X,Y] u \rangle|^2.$$ 
\end{prop}
\begin{proof}
See Proposition 5.1 in \cite{parthasarathy2012introduction}. 
\end{proof}

This is an abstract version of the Heisenberg uncertainty principle in quantum mechanics. It conveys the impossibility of measuring both $X$ and $Y$ with total precision in the pure state $|u\rangle$ (see also Appendix A in \cite{lampo2016lindblad}). 
\end{itemize}

A comparison between notions arising in classical probability and the notions in the above postulates for quantum mechanics is summarized in Table \ref{table}. 

We demonstrate the notions introduced so far and the above consequences of noncommutativity in an example. These consequences should be kept in mind when we study open quantum systems in the next section.

\begin{ex}
The following describes a quantum particle moving in a three dimensional space.
The wave function, i.e. the complex-valued
function of the particle's position, $|\psi(x) \rangle \in \mathcal{H} = L^2(\RR^3) := \{ |\psi(x)\rangle : \RR^3 \to \CC : \int_{\RR^3} |\psi(x)|^2 dx < \infty \}$ (with the inner product, $\langle \phi | \psi \rangle = \int_{\RR^3} \overline{\phi(x)} \psi(x) dx $), determines the pure state, $|\psi(x) \rangle \langle \psi(x)|$.  

Two important observables are the position $q$ and momentum $p$, defined as
\begin{equation}
q \psi(x) = x \psi(x), \  \ \ p \psi(x) = -i\hbar \vecc{\nabla}_{x} \psi(x), 
\end{equation}
where $\psi(x)$ belongs to a dense domain\footnote{It is useful to consider wave functions $\psi(x)$ that live in a suitable test function space, i.e. a linear subspace $\mathcal{D} \subset \mathcal{H}$.  For example, $\mathcal{D}$ can be the set of all smooth functions with support in some compact subset $K \subset \RR$. Physically, this choice of $\mathcal{D}$ says that the particle is confined to the region $K$ in space. } of $\mathcal{H}$ and $\hbar \approx 1.05 \times 10^{-34}$ joule-second is the (reduced) Planck constant. They are unbounded self-adjoint operators and are noncommuting, since $[q, p] \psi(x)= i\hbar \psi(x)$, which is the canonical commutation relation (CCR) between $q$ and $p$.  By Nelson's theorem, which states that two observables have a joint probability distribution (in the sense as described in Section 2 of \cite{BreuerBook}) if and only if they commute, it follows that $q$ and $p$ do not admit a joint distribution in the same state. Applying the uncertainty principle, we see that formally $var_u(q) var_u(p) \geq \hbar^2/4$ when the particle is in the normalized state $u$. 

Later we will describe a system of infinitely many identical quantum particles. In that case,  we will see that the CCRs among their observables specify the quantum statistics and the properties of the representation of these CCRs will be exploited to build a theory of quantum stochastic integration.  
\end{ex}

\begin{table}[h] 
\centering
\caption{Notions in Classical and Quantum Probability}
\label{table}
\begin{tabular}{|l|l|l|}
\hline
\multicolumn{1}{|c|}{\textbf{Notions}} & \multicolumn{1}{c|}{\textbf{Classical Probability}}                                        & \multicolumn{1}{c|}{\textbf{Quantum Probability}}                                          \\ \hline
\textbf{State Space}                        & Set of all possible outcomes, $\Omega$                                                   & Complex separable Hilbert space, $\mathcal{H}$                                           \\ \hline
\textbf{Events}                             & $\mathcal{F}$, set of all indicator functions in $\Omega$  &  $\mathcal{P}(\mathcal{H})$, set of all projections in $\mathcal{H}$ \\ \hline
\textbf{Observables}                        & Measurable functions                                                   & Self-adjoint operators in $\mathcal{H}$                                                    \\ \hline
\textbf{States}                             & Probability measure, $\mu$                                                       & Positive operators of unit trace,  $\rho$                            \\ \hline
\textbf{Prob. Space}                 & Measure space $(\Omega, \mathcal{F},\mu)$                                                  & The triple $(\mathcal{H},\mathcal{P}(\mathcal{H}), \rho)$                                  \\ \hline
\end{tabular}
\end{table}

\begin{rmk} {\bf Interplay between classical and quantum probability.}
Quantum probability, in its algebraic formulation, can be seen as a generalization of classical probability as follows. The algebra $\mathcal{B}(\mathcal{H})$ contains many $\sigma$-algebras of mutually commuting projectors. Consider for instance, $\mathcal{H} = L^2(\Omega, \mathcal{F}, P)$, then $L^{\infty}(\Omega, \mathcal{F},P)$ is a commutative (von Neumann) algebra acting on $\mathcal{H} = L^2(\Omega, \mathcal{F},P)$ by multiplication. The projectors (events) in $L^{\infty}(\Omega, \mathcal{F},P)$ are the operators of multiplication by indicator functions of elements of $\mathcal{F}$. The spectral theorem below provides the crucial link between classical and quantum probability.

\begin{thm} There exists a probability space $(\Omega, \mathcal{F},P)$ and an $*$-isomorphic map $\phi$ from a commutative $*$-algebra of operators on $\mathcal{H}$ onto the set of measurable functions on $\Omega$ (i.e. a linear bijection with $\phi(AB) = \phi(A) \phi(B)$ and $\phi(A^*) = \phi(A)^*$). 
\end{thm}
\begin{proof}
See \cite{Bouten2008}.
\end{proof}

The probability measure $P$ induces a state $\rho$ on the commutative algebra by $\rho(f) = \int_{\Omega} f(w) dP(w)$.   A classical random variable $X$ can be described as a quantum random variable (observable) by the $*$-homomorphism
$J : L^{\infty}(E, \mathcal{E}, P) \to  L^{\infty}(\Omega, \mathcal{F}, P)$, $J(f) = f\circ X$, where $(E, \mathcal{E})$ is  a measurable space. However, two non-commuting self-adjoint operators cannot be represented as multiplication
operators on the same Hilbert space $L^2(\Omega, \mathcal{F}, P)$.
\end{rmk}


\begin{itemize}
\item[(A5)] {\bf Dynamics.} The (reversible) time evolution of a pure state (wave function) in $\mathcal{H}$ is determined by a unitary operator $U: \mathcal{H} \to \mathcal{H}$. By Stone's theorem, if $t \mapsto U_{t}$ is a  strongly continuous, one-parameter unitary group, then there exists a unique linear self-adjoint operator $H$, the {\it Hamiltonian}, such that $U(t) = e^{-itH}$.  In this case, the pure state $|u\rangle$ evolves according to $|u(t)\rangle = U(t) |u\rangle$, $U_0 = I$, $\langle u | u \rangle = 1$, and satisfies the celebrated {\it Schrodinger equation}: 
\begin{equation}
i \hbar \frac{d}{dt} |u(t)\rangle = H |u(t)\rangle, \ \ |u(0)\rangle = |u\rangle, \ \ t \geq 0.
\end{equation}
Physically, the Schrodinger equation is sensible, as the linearity of the dynamics ensures the superposition principle for states is satisfied and the fact that the Schrodinger equation is first order in time guarantees causality of the states. The self-adjointness (in particular, the symmetry) of $H$ is important to conserve the probability at all times, i.e. $\| u(t) \| = 1$ for all $t \geq 0$.

For general states which are represented by density operators, they evolve via the map $\rho \mapsto \rho(t) = U_{t} \rho U_{t}^{*} $. This is the so-called {\it Schrodinger picture}, where the states evolve in time while the observables are fixed. The dual picture, where the states are fixed while the observables evolve in time, is called the {\it Heisenberg picture}. In this picture, an observable evolves according to the map ($^*$-automorphism) $X \mapsto X(t) = \tau_t(X) = U_{t}^{*} X U_{t}$ and satisfies the {\it Heisenberg equation of motion}:
\begin{equation}
\frac{d}{dt} X(t) = \frac{i}{\hbar} [H, X(t)], \ \ X(0) = X.
\end{equation}
The  two pictures are related via $Tr(\rho X(t)) = Tr(\rho(t) X)$ (by the cyclic property of trace). Taking the expectation of $X$ with respect to the state $|u\rangle$, one  obtain the {\it Ehrenfest equation}, a quantum analogue of the classical Hamilton's equation:
\begin{equation}
\frac{d}{dt} \langle X(t) \rangle = \frac{i}{\hbar} \langle [H, X(t)] \rangle, \ \ \langle X(0) \rangle = \langle X \rangle \in \RR.
\end{equation}
\end{itemize}

\begin{ex} Suppose that the single particle in the previous example is a quantum harmonic oscillator and so is described by the Hamiltonian: 
\begin{equation}
H_{ho} = \frac{p^2}{2m} + \frac{1}{2} m \omega^2 q^2,
\end{equation}
where $m$ and $\omega$ denote the mass and frequency of the particle respectively. Then the Heisenberg equation of motions for its position and momentum give
\begin{equation}
\dot{q}(t) = p(t)/m, \  \ \  \dot{p}(t) = -m\omega^2 q(t),
\end{equation}
which can be seen to be quantum analogue of the Newton's second law, $\ddot{q}(t) = -\omega^2 q(t)$.
\end{ex}

\subsection{Composite Systems}
We need to be able to describe a family of independent systems (recall that so far we have focused on a system -- that of a single particle).  Let $\times$ and $\otimes$ denote the Cartesian and tensor product\footnote{For definitions of tensor product, see \cite{parthasarathy2012introduction} (for a coordinate-free approach based on positive definite kernels) and \cite{reed1981functional} (for the usual approach).} respectively in the following. 

In classical probability, if $(\Omega_{1}, \mathcal{F}_1, P_{1}), \dots, (\Omega_{n}, \mathcal{F}_n, P_{n})$ are probability spaces describing $n$ independent statistical systems, then  the product probability space $(\Omega, \mathcal{F}, P)  := (\Omega_{1} \times \dots \Omega_{n}, \mathcal{F}_1 \times \dots \times \mathcal{F}_n, P_{1} \times \dots \times P_{n})$ describes a single system consisting of  the $n$ systems, and  $$P(F_1 \times \dots \times F_n) = P_1(F_1) \cdots P_n(F_n),$$ for any event $F_i \in \mathcal{F}_i$. 
In quantum mechanics,  the Hilbert space of a composite system consisting of independent component systems is the Hilbert space  tensor product of the state spaces associated with the component systems. More precisely, if $(\mathcal{H}_{1}, \rho_{1}), \dots, (\mathcal{H}_{n}, \rho_{n})$ describe $n$ independent quantum systems, then $(\mathcal{H}, \rho) := (\mathcal{H}_{1} \otimes \dots \otimes \mathcal{H}_{n},  \rho_{1} \otimes \dots \otimes  \rho_{n})$ describes a single quantum system consisting of the $n$ systems, where $\rho_{1} \otimes \dots \otimes \rho_{n}$ is a state on $\mathcal{H}_{1} \otimes \dots \otimes \mathcal{H}_{n}$. 

For $i=1,\dots,n$, let $X_i$ be an observable on $\mathcal{H}_i$ and  the spectral projection $P^{X_i}(E_i)$ be the event that the value of the observable $X_i$ lies in the Borel subset $E_i \in \RR$. Then $X = X_1 \otimes \dots \otimes X_n$ is an observable on $\mathcal{H}$ and the probability that the value of the observable $X$ lies in $E = E_1 \times \dots \times E_n$ is given by 
\begin{align} 
\mu_{X_1, \dots, X_n}(E_1 \times \dots \times E_n) &= Tr((\rho_1 \otimes \dots \otimes \rho_n) (P^{X_1}(E_1) \otimes \dots \otimes P^{X_n}(E_n) )) \nonumber \\ 
&= Tr(\rho_1 P^{X_1}(E_1)) \cdots Tr(\rho_n P^{X_n}(E_n)) = \mu_{X_1}(E_1) \cdots \mu_{X_n}(E_n).
\end{align}

Other notions introduced in the Postulate (A4)-(A5) can be extended analogously to composite systems. The above descriptions can be generalized to infinitely many independent systems and in the case when the systems are also identical it is convenient to achieve this task on a Hilbert space endowed with a certain structure: the Fock space, to be introduced next.


\section{Elements of Quantum Stochastic Analysis} \label{sect_qsc}

We provide a minimal review of the basic ideas and results from Hudson-Parthasarathy (H-P) quantum stochastic calculus, which is a  bosonic\footnote{A stochastic calculus can also be constructed in the setting of a fermionic Fock space and in fact such
calculus is related to the one based on the Fock space \citep{hudson1986unification}.} Fock space stochastic calculus based on the creation, conservation and annihilation operators of quantum field theory. The goal of our review is to convince, at least at a formal level, the readers that a quantum version of stochastic calculus can be developed in parallel with the classical calculus.  For details of the calculus, we refer to the monographs \cite{meyer2006quantum} and \cite{parthasarathy2012introduction}. Again, we follow \cite{parthasarathy2012introduction} closely in the following.

Recall that  a classical stochastic process is  a family of random variables (classical observables) on  a probability space $(\Omega, \mathcal{F}, P)$ indexed by $t \in \RR^+ := [0,\infty)$. To see how one could formulate the concept of quantum stochastic process, let us first consider a family of commuting observables, $\{X(t), t \in T\}$ where $T \subset \RR^+$ is a time interval. Since this family can be simultaneously diagonalized, we are allowed to consider observables of the form $\sum_{j=1}^n \xi_j X(t_j)$ for any finite set $\{t_1, \dots, t_n\} \subset T$ and real constants $\xi_j$ ($j=1,\dots,n$) and define the {\it joint characteristic function} of $\vecc{X}:=(X(t_1),\dots,X(t_n))$, i.e.  Fourier transform of the joint probability distribution $\mu_{X(t_1), \dots, X(t_n)}$ in $\RR^n$:
\begin{equation}
\langle e^{i \vecc{\xi} \cdot \vecc{X}} \rangle := \langle u | e^{i \sum_{j=1}^n \xi_j X(t_j)} | u \rangle,
\end{equation}
where $|u\rangle$ is a state and $\vecc{\xi}=(\xi_1,\dots,\xi_n) \in \RR^n$. Then the family $\{\mu_{X(t_1), \dots, X(t_n)} : \{t_1, \dots, t_n\} \subset T, \ n=1,2,\dots \}$ of all finite-dimensional distributions is consistent.  Therefore, it follows from Kolmogorov's theorem that the family $\{X(t), t \in T\}$ defines a stochastic process. Note that the correspondence
$\vecc{\xi} \to e^{i \vecc{\xi} \cdot \vecc{X}}$ is a unitary representation of $\RR^n$.

An important example of classical stochastic process that can be constructed via the above procedure is the Wiener process, with respect to which a stochastic integral can be defined \citep{karatzas2012brownian}. Adapting the above view point to a family of commuting (operator-valued and not necessarily bounded) quantum observables with respect to a class of states in a Hilbert space, one can construct quantum stochastic processes \citep{accardi1982quantum}. A particular quantum analogue of the above unitary representation   (the Weyl representation) will be important in the case when the observables are unbounded. We will focus on the construction of quantum analogue of Wiener process  within the setting of a bosonic Fock space.


\subsection{Bosonic Fock Space}


A Fock space describes states of a quantum field consisting of an indefinite number of identical particles. It is a crucial object in the formalism of second quantization used to study quantum many-body systems. The main idea of second quantization is to specify quantum states by the number of particles occupying the states, rather than labeling each particle with its state, thereby eliminating redundant information concerning identical particles and allowing an efficient description of quantum many-body states. From the perspective of quantum stochastic modeling, it is a natural space\footnote{The theory of Fock space provides a convenient framework to study  not only quantum fields, but also other objects, such as the Carleman linearization techniques in nonlinear dynamical systems \citep{kowalski1991nonlinear,kowalski1994methods} and classical stochastic mechanics \citep{baez2012quantum}.} to support the quantum noise, describing the effective action of the environment on a system of interest.

A system of identical particles is described by either a totally symmetric wave function
(invariant under exchange of any two coordinates) or a totally asymmetric wave function. This gives rise to two distinct types of particles: bosons in the former case and fermions in the latter case. We are only interested in description for bosonic systems.

\begin{defn} The {\it bosonic Fock space}, over the one-particle space $\mathcal{H}$, is defined as the countable direct sum: 
\begin{equation} \Gamma(\mathcal{H}) =  \CC \oplus \mathcal{H} \oplus \mathcal{H}^{\circ 2} \oplus \dots \oplus \mathcal{H}^{\circ n} \oplus \dots, \end{equation} where $\CC$, denoting the one-dimensional space of complex scalars, is called the {\it vacuum subspace} and $\mathcal{H}^{\circ n}$, denoting the symmetric tensor product of $n$ copies of $\mathcal{H}$, is called the {\it $n$-particle subspace}. Any element in an $n$-th particle subspace is called an {\it $n$-particle vector}. For any $n$ elements $|u_1\rangle, |u_2\rangle, \dots, |u_n\rangle$  in $\mathcal{H}$, the vector $\otimes_{j=1}^n |u_{j}\rangle$ is known as the {\it finite particle vector (or Fock vector)}.  The dense linear manifold $\mathcal{F}(\mathcal{H})$ of all finite particle vectors is called the {\it finite particle domain}. 
\end{defn} 

Since the particles constituting the noise space (and in each of the $n$-particle space) are bosons, in order to describe the $n$-particle state (i.e. to belong to the $n$-particle space, $\mathcal{H}^{\circ n}$), a Fock vector has to be symmetrized:   \begin{equation} 
|u_1\rangle \circ |u_2 \rangle \circ \dots \circ |u_n\rangle =  \frac{1}{n!} \sum_{\sigma \in \mathcal{P}_n} |u_{\sigma(1)}\rangle \otimes |u_{\sigma(2)}\rangle \otimes \dots \otimes |u_{\sigma(n)}\rangle, \end{equation}
where $\mathcal{P}_n$ is the set of all permutations, $\sigma$, of the set $\{1,2,\dots,n\}$. The $n$-particle space is  invariant under the action of the permutation group $\mathcal{P}_n$. 


Important elements of the bosonic Fock space, $\Gamma(\mathcal{H})$, are the {\it exponential vectors}: 
\begin{equation} |e(u)\rangle =  1 \oplus |u \rangle \oplus \frac{|u \rangle^{\otimes 2}}{\sqrt{2!}} \oplus \dots \oplus \frac{|u \rangle^{\otimes n}}{\sqrt{n!}} \oplus \dots, \end{equation}
where $|u \rangle \in \mathcal{H}$ and $|u\rangle^{\otimes n}$ denotes the tensor product of $n$ copies of $|u\rangle$.
We call $|\Omega \rangle:= |e(0)\rangle$  the {\it Fock vacuum vector}, which corresponds to the state with no particles. Note that $|\psi(u)\rangle = e^{-\langle u| u \rangle/2} |e(u)\rangle$ is a unit vector. The pure state with the density operator $|\psi(u) \rangle \langle \psi(u)|$ is called the {\it coherent state} associated with $|u\rangle$.   In the special case when $\mathcal{H} = \CC$, the coherent states on $\Gamma(\mathcal{H}) = \CC \oplus \CC \oplus \cdots$  are sequences of the form: \begin{equation}|\psi(\alpha)\rangle = e^{-|\alpha|^2/2} \left(1, \alpha, \frac{\alpha^2}{\sqrt{2!}}, \cdots, \frac{\alpha^{n}}{\sqrt{n!}} \dots \right).\end{equation}

We collect some basic properties of exponential vectors in the following. 

\begin{prop} {\bf Basic properties of exponential vectors.} \label{prop_expvec}
\begin{itemize}
\item[(i)] For all $|u\rangle$, $|v \rangle \in \mathcal{H}$, the exponential vectors satisfy the following  scalar product formula:
\begin{equation} \langle e(u)|e(v) \rangle = e^{\langle u |  v \rangle}, \end{equation} with the same notation for scalar products in appropriate spaces.  
\item[(ii)] The map $|u\rangle \mapsto |e(u) \rangle$ from $\mathcal{H}$ into $\Gamma(\mathcal{H})$ is continuous. 
\item[(iii)] The set $\{ |e(u)\rangle : |u\rangle \in \mathcal{H}\}$ of all exponential vectors is linearly independent and total in $\Gamma(\mathcal{H})$, i.e. the smallest closed subspace containing the set is the whole space $\Gamma(\mathcal{H})$.
\item[(iv)]  Let $S$ be a dense set in $\mathcal{H}$. Then the linear manifold $\mathcal{E}(S)$ generated by $M := \{|e(u)\rangle : |u\rangle \in S\}$ is dense in $\Gamma(\mathcal{H})$. For every map $T: M \to \Gamma(\mathcal{H})$, there exists a unique linear operator $T'$ on $\Gamma(\mathcal{H})$ with domain $\mathcal{E}(S)$ such that $T' |e(u)\rangle = T|e(u)\rangle$ for all $|u\rangle \in S$.  
\item[(v)] Let $\mathcal{H}_i$ be Hilbert spaces, $|u_i\rangle \in \mathcal{H}_i$ $(i=1,\dots,n)$  and $\mathcal{H} = \oplus_{i=1}^n \mathcal{H}_i$. Then there exists a unique unitary isomorphism $U: \Gamma(\mathcal{H}) \to \Gamma(\mathcal{H}_1) \otimes \dots \otimes \Gamma(\mathcal{H}_n)$ satisfying the relation:
\begin{equation}
U |e(u_1 \oplus \dots \oplus u_n)\rangle = |e(u_1)\rangle \otimes \dots \otimes |e(u_n)\rangle,
\end{equation} 
for every $|u_i\rangle \in \mathcal{H}_i$. 
\end{itemize}
\end{prop}
\begin{proof}
(i) follows from a straightforward computation. (ii) follows from the proof in Corollary 19.5 in \cite{parthasarathy2012introduction}. (iii)-(v) follow from Proposition 19.4, Corollary 19.5 and Proposition 19.6 in \cite{parthasarathy2012introduction} respectively.  One key ingredient in showing (iv)-(v) is Proposition 7.2 in \cite{parthasarathy2012introduction}.
\end{proof}

When  $S = \mathcal{H}$ in (iv), we call $\mathcal{E} = \mathcal{E}(\mathcal{H})$ the {\it exponential domain} in $\Gamma(\mathcal{H})$. By $(iv)$, $\mathcal{E}$ is dense in $\Gamma(\mathcal{H})$. Therefore, any linear operator on $\Gamma(\mathcal{H})$ can be determined by its action on the exponential vectors.

It turns out that bosonic Fock spaces have many interesting connections with Gaussian stochastic processes (see \cite{janson1997gaussian} and  Example 19.8-19.12 in \cite{parthasarathy2012introduction}). We only mention one such connection: that with the Wiener process in classical probability. 

\begin{ex} (From Example 19.9 in \cite{parthasarathy2012introduction}) \label{expmart_iso}
Consider the Hilbert spaces $L^2(\RR^{+})$, $\Gamma(L^2(\RR^{+}))$, where $\RR^{+} = [0,\infty)$ is equipped with Lebesgue measure, and $L^2(\mu)$, where $\mu$ is the probability measure of the standard Wiener process $\{W(t), t \geq 0\}$. For any complex-valued function $f \in L^2(\RR^{+})$, let $\int_0^\infty f dW$ denote the stochastic integral of $f$ with respect to the path $W$ of the Wiener process. Then there exists a unique unitary isomorphism (the Wiener-Segal duality transformation) $U: \Gamma(L^2(\RR^+)) \to L^2(\mu)$ satisfying:
\begin{equation}
[U |e(f)\rangle](W) = \exp\left\{\int_0^\infty f dW - \frac{1}{2} \int_0^\infty f(t)^2 dt \right\} =: e^W(f),
\end{equation}
and $$\langle e(f) | e(g) \rangle = e^{\langle f | g \rangle} = \mathbb{E}_\mu \overline{e^W(f)} e^W(g)$$   for all $f,g \in L^2(\RR^+)$. 

In particular, this implies that $|e(1_{[0,t]} f)\rangle$ can be identified with the exponential martingale, $\exp\{\int_0^t f dW - \frac{1}{2} \int_0^t f^2(s) ds \}$, in classical probability for every $t> 0$. This suggests that the operators of multiplication by an indicator function (with respect to time intervals), together with (v) in Proposition \ref{prop_expvec}, will be crucial  when building a theory of quantum stochastic integration.  
\end{ex}

\subsection{The Weyl Representation and Stochastic Processes in Bosonic Fock Spaces}

An important group in the theory of quantum stochastic calculus is the translation group on  the Hilbert space $\mathcal{H}$. Indeed, any Hilbert space, $\mathcal{H}$, being a vector space, is an additive group, which has a natural translation action on the set of all exponential vectors by $|u\rangle: |e(v)\rangle \mapsto |e(v+u)\rangle$, where $|u\rangle$, $|v\rangle \in \mathcal{H}$. By requiring this action to be scalar product preserving, we define the {\it Weyl operator} (displacement operator):
\begin{equation} \label{defn_weylop}
W(u) |e(v)\rangle = e^{-\frac{1}{2}\|u\|^2 - \langle u | v \rangle} |e(u+v)\rangle,
\end{equation}
for every $|v\rangle \in \mathcal{H}$. Note that $\langle W(u) e(v_1)| W(u) e(v_2)\rangle = \langle e(v_1) | e(v_2) \rangle$ for every $|v_1\rangle, |v_2\rangle \in \mathcal{H}$. By the totality of the set of all exponential vectors, it follows that there exists a unique unitary operator $W(u)$ in $\Gamma(\mathcal{H})$ satisfying the above formula for every $|u\rangle  \in \mathcal{H}$. 

\begin{thm} Let $\mathcal{H}$ be a complex separable Hilbert space. Let $W(u)$ be the Weyl operator defined in \eqref{defn_weylop}. 
The correspondence $|u\rangle \to W(u)$ from $\mathcal{H}$ into the set of unitary operators on $\mathcal{H})$ is strongly continuous and irreducible in $\Gamma(\mathcal{H})$, in the sense that there is no proper subspace in $\Gamma(\mathcal{H})$ that is invariant under all $W(u)$. 

Moreover, for every $|u_1\rangle, |u_2\rangle \in \mathcal{H}$, we have: 
\begin{align}
W(u_1) W(u_2) &= e^{-i Im(\langle u_1 | u_2 \rangle)} W(u_1+u_2), \label{proj_rep} \\ 
W(u_1) W(u_2) &= e^{-2i Im(\langle u_1 | u_2 \rangle)} W(u_2) W(u_1). \label{weyl_cr}
\end{align}
It follows that for every $|u\rangle \in \mathcal{H}$, the map $t \mapsto W(tu)$, $t \in \RR$, is a one-parameter group of unitary operators with the self-adjoint Stone generator $p(u)$, satisfying
\begin{equation}
W(tu) = e^{-it p(u)}, 
\end{equation}
for all $t \in \RR$. The observables $p(u)$ obey the following commutation relation: 
$$[p(u), p(v)] |e(w)\rangle = 2 i Im \langle u, v \rangle |e(w)\rangle $$ 
for all $|u\rangle, |v\rangle, |w\rangle \in \mathcal{H}$. 
\end{thm}
\begin{proof}
This is a special case of Theorem 20.10 in \cite{parthasarathy2012introduction}.  
\end{proof}

The correspondence $|u\rangle \to W(u)$ is called a {\it projective unitary representation}. The formula \eqref{proj_rep} implies that it is a homomorphism modulo a phase factor of unit modulus. The formula \eqref{weyl_cr} is known as the {\it Weyl commutation relation}. Such representation allows one to obtain a rich class of observables, which are the building blocks of the  calculus, on the Fock space $\Gamma(\mathcal{H})$. From these observables, one can then build quantum analogues of Wiener process. Illustrating this is the focus of this subsection.

\begin{rmk}
A more general group called the Euclidean group, which contains the translation group as a subgroup, would allow one to obtain a richer class of observables, including quantum analogue of L\'evy processes \citep{parthasarathy2012introduction}. This shows the power of the formalism, as it allows realization of processes such as Wiener process and Poisson process on the same space. However, since we are only interested in the stochastic integration theory with respect to  quantum analogue of the Wiener process in this paper, we omit further discussions on the general construction. For details, see \cite{parthasarathy2012introduction}. 
\end{rmk}



We now introduce a family of operators in terms of which not only computations involving the $p(u)$ become simplified but can also be related to operators familiar from quantum field theory. 

We define, for any $|u\rangle \in \mathcal{H}$,
\begin{equation}
q(u) = -p(iu) = p(-iu),  \ \ \ a(u) = \frac{1}{2}(q(u)+i p(u)), \ \ \ a^\dagger(u) = \frac{1}{2}(q(u)-ip(u)). 
\end{equation}

The  operators $a^{\dagger}(u)$ and $a(u)$ defined above are canonical observables on the bosonic Fock space, called the {\it creation operators} and {\it annihilation operators} associated to the vector $|u\rangle \in \mathcal{H}$, respectively. Following \cite{parthasarathy2012introduction}, we will refer to them as the {\it fundamental fields}. Note that $p(u) = i (a^\dagger(u) - a(u))$ and $q(u) = a^\dagger(u)+a(u)$. 

We collect some useful properties, which will be crucial for the development of quantum stochastic calculus, of these operators in the following. 

\begin{prop}
The domain of product of finitely many operators from the family $\{a(u), a^\dagger(u) : |u\rangle \in \mathcal{H}\}$ contains the exponential domain $\mathcal{E}$. Moreover, for any $|u\rangle, |v\rangle \in \mathcal{H}$, $\psi, \psi_1, \psi_2 \in \mathcal{E}$,
\begin{itemize}
\item[(i)] \begin{equation} 
a(u) |e(v)\rangle = \langle u | v \rangle |e(v)\rangle, \ \ \ \ a^\dagger(u) |e(v)\rangle = \sum_{n=1}^\infty \frac{1}{\sqrt{n!}} \sum_{r=0}^{n-1} |v\rangle^{\otimes r} \otimes |u\rangle \otimes |v\rangle^{\otimes (n-r-1)};
\end{equation}
\item[(ii)] the creation and annihilation operators are mutually adjoint, i.e. $\langle a^\dagger(u) \psi_1 | \psi_2\rangle = \langle \psi_1 | a(u) \psi_2 \rangle$; 
\item[(iii)] the restrictions of $a(u)$ and $a^\dagger(u)$ to $\mathcal{E}$ are antilinear and linear in $|u\rangle$ respectively. Moreover, they satisfy the {\it canonical commutation relations (CCRs)}: $[a(u), a(v)] \psi  = [a^{\dagger}(u), a^{\dagger}(v)] \psi = 0$ and $[a(u), a^{\dagger}(v)] \psi = \langle u | v \rangle \psi$.
\item[(iv)] 
\begin{equation} a^{\dagger}(u) |e(v)\rangle = \frac{d}{d\epsilon} |e(v+\epsilon u)\rangle \bigg|_{\epsilon = 0}, \end{equation}
\item[(v)] the linear manifold of all finite particle vectors is contained in the domain of $a(u)$ and $a^\dagger(u)$. Moreover, 
\begin{align} a(u) |\Omega\rangle &= 0, \\ 
a(u) |v\rangle^{\otimes n} &= \sqrt{n} \langle u | v \rangle |v\rangle^{\otimes (n-1)}, \\ 
a^{\dagger}(u)  |v\rangle^{\otimes n} &= \frac{1}{\sqrt{n+1}} \sum_{r = 0}^{n} |v\rangle^{\otimes r } \otimes |u\rangle \otimes |v\rangle^{\otimes (n-r)}.\end{align} 
\end{itemize}
\end{prop} 
\begin{proof} 
See  Proposition 20.12-20.14 in \cite{parthasarathy2012introduction}. The key idea to obtain the formula in (i), (iii)-(v) is to replace, in the definition of Weyl operator in \eqref{defn_weylop}, $u$ by $tu$, $t \in \RR$, and then differentiating with respect to $t$ at $t=0$, so that one obtains:
\begin{equation}
p(u) |e(v)\rangle = -i \langle u | v \rangle |e(v)\rangle + i \sum_{n=1}^{\infty} \frac{1}{\sqrt{n!}} \sum_{r=0}^{n-1} |v\rangle^{\otimes r} \otimes |u\rangle \otimes |v\rangle^{\otimes (n-r-1)}. 
\end{equation} The formula there then lead to (ii) and the statements about the domain of the operators. 
 \end{proof}

Note that in the special case $|u\rangle=|v\rangle$ in (i), we have $a(u) |\psi(u)\rangle = \langle u | u \rangle |\psi(u)\rangle$, which is an eigenvalue relation similar to the one that defines the coherent state as eigenvector of annihilation operator in quantum optics \citep{glauber1963coherent}. Since vectors of the form $|v\rangle^{\otimes n}$ linearly span the $n$-particle space, (v) shows that $a(u)$ maps the $n$-particle subspace into the $(n-1)$-particle subspace while $a^{\dagger}(u)$ maps the $n$-particle subspace into the $(n+1)$-particle subspace, justifying their names as annihilation and creation operators respectively.

\begin{rmk} {\bf Connection to quantum field theory.}
By working with appropriate basis in $\mathcal{H}$,  we can relate the above creation and annihilation operators (the fundamental fields) to those familiar from quantum field theory. For instance, choose and fix an orthonormal basis $\{|e_k\rangle, k=1,2,\dots\}$ and define $a_k = a(e_k)$, $a^\dagger_k = a^\dagger(e_k)$. Then on the exponential domain $\mathcal{E}$, these operators satisfy the commutation relations $[a_k, a_l] = [a^\dagger_k, a^\dagger_l] = 0$ and $[a_k, a^\dagger_l] = \delta_{kl}$. These are the commutation relations describing the ladder operators for a set of independent harmonic oscillators. 

Now, let us move to a different basis of single particle state, say the position eigenstates $|r\rangle$, corresponding to the unitary transformation $|r\rangle =  \sum_k \langle e_k | r \rangle |e_k \rangle$. Note that  $\langle e_k | r \rangle = \varphi_k^*(r)$, the complex conjugate of the wavefunction $\varphi_k(r)$. The transformation gives a new set of creation operators, $\psi^\dagger(r) = \sum_k \varphi^*_k(r) a^\dagger_k$, and annihilation operators, $\psi(r) = \sum_k \varphi_k(r) a_k$.  They satisfy the commutation relations $[\psi(r_1), \psi^\dagger(r_2)] = \delta(r_1-r_2)$, $[\psi(r_1), \psi(r_2)] = [\psi^\dagger(r_1), \psi^\dagger(r_2)] = 0$ for $r_1, r_2 \in \RR$, which are commutation relations describing operator-valued functions called the {\it quantum fields}.
\end{rmk}

Lastly, we investigate statistical features of  observables arising from Weyl representation and show that these observables can be viewed as quantum analogue of stochastic processes. From definition, we have:
\begin{equation} \label{char_func}
\bigg\langle \Omega \bigg| W\left(\sum_{j=1}^n t_j u_j\right) \bigg| \Omega \bigg\rangle = \exp{\left( -\frac{1}{2} \sum_{i,j} t_i t_j \langle u_i | u_j \rangle \right)}
\end{equation}
for $|u_j\rangle \in \mathcal{H}$, $t_j \in \RR$, $1 \leq j \leq n$. Let $\mathcal{H}_\RR$ be a real subspace of $\mathcal{H}$ such that $\mathcal{H} = \mathcal{H}_\RR \oplus i \mathcal{H}_\RR$. Then $\langle u | v \rangle \in \RR$ for $|u\rangle, |v\rangle \in \mathcal{H}_\RR$ and $\{W(u): |u\rangle \in \mathcal{H}_\RR\}$ is a commutative family of operators (due to \eqref{weyl_cr}). In particular, $\{p(u): |u\rangle \in \mathcal{H}_\RR \}$ is a commutative family of observables, and so from $\eqref{char_func}$, one has 
\begin{equation} 
\bigg\langle \Omega \bigg| \exp{ \left( - i \sum_{j=1}^n t_j p(u_j) \right) } \bigg| \Omega \bigg\rangle = \exp{\left( -\frac{1}{2} \sum_{i,j} t_i t_j \langle u_i | u_j \rangle \right)},
\end{equation}
which is the characteristic function of the $n$-dimensional Gaussian distribution with zero means and covariance matrix $(\langle u_i | u_j \rangle)_{i,j=1,\dots,n}$. Therefore, the above commutative family of observables in the Fock space realizes a zero mean classical Gaussian random field \citep{mandrekar2015stochastic} in the vacuum state $|\Omega\rangle$ (more generally, in the coherent states; see Proposition 21.1 in \cite{parthasarathy2012introduction}). 

\begin{ex} \label{ex_qprocess} (From Example 21.3 in \cite{parthasarathy2012introduction}) 
Let $S \subset \mathcal{H}$ be a real subspace and $P$ be a spectral measure on $\RR^{+}$ for which $S$ is invariant.  Let $|u_{t]}\rangle= P([0, t]) |u\rangle$,  $X_t = p(u_{t]})$ and $R(t) = \langle u | P([0, t]) |u \rangle$, where $|u\rangle \in S$ and $p(u_{t]})$ is the Stone generator of the map $t \mapsto W(u_{t]})$.  Then $\{X_t : t \in \RR\}$ is a family of commuting observables (with common domain $\mathcal{E}$) whose distribution in the vacuum state is a mean zero Gaussian process with independent increments and $cov(X_t, X_s) = R(\min(t,s))$. Note that when $ \langle u| P([0, t])|u \rangle = t$, $\{X_t: t \in \RR\}$ realizes the standard Wiener process in classical probability. 
\end{ex}

\subsection{Elements of Quantum Stochastic Calculus} 

Example \ref{ex_qprocess} suggests that we can turn the fundamental fields on a bosonic Fock space into continuous time quantum stochastic processes\footnote{See also the more abstract definition in \cite{accardi1982quantum}, which defines quantum stochastic process $J_{t}(X)$ as a family of  continuous $*$-homomorphisms on $(\mathcal{H}, \rho)$ indexed by $t$.}  provided that an appropriate time parameter is introduced in the fields. Then one could develop a quantum stochastic differential description for a large class of observable-valued maps $t \mapsto X(t)$ in terms of these quantum stochastic processes, in parallel with classical integration theory. This is the basic idea in the seminal work \cite{hudson1984quantum} and will be elaborated in the following.

To introduce the time parameter in the fundamental fields, we take the one-particle space to be $\mathcal{H} = L^2(\RR^+) \otimes \mathcal{Z} = L^2(\RR^+; \mathcal{Z})$, with its Borel structure and Lebesgue measure. Symmetrically second quantizing this space leads to the bosonic Fock space $\Gamma(L^2(\RR^+) \otimes \mathcal{Z})$.  Here $\mathcal{Z}$ is a complex separable Hilbert space, equipped with a complete orthonormal basis $(|z_{k}\rangle)_{ k \geq 1}$. The space $\mathcal{Z}$ is called the {\it multiplicity space} of the noise.  The space $\mathcal{H}$ is equipped with the scalar product:
\begin{equation}
\langle f | g \rangle = \int_0^\infty \langle f(t) | g(t) \rangle_{\mathcal{Z}} dt,
\end{equation}
and any element in it can be viewed as a norm square integrable function from $\RR^{+}$ into $\mathcal{Z}$. As we will be formulating a differential (in time) description of processes on the Fock space, $\RR^+$ represents the time semi-axis.

Physically, the dimension of $\mathcal{Z}$ is the number of field channels that ones can couple to a system. When $\mathcal{Z} = \CC$ (one-dimensional), the corresponding bosonic Fock space, $\Gamma(L^2(\RR^{+}))$, describes a single field channel \citep{nurdinlinear}. When $\mathcal{Z} = \CC^{d}$ and the $|z_{i}\rangle = (0, \dots , 0, 1, 0, \dots , 0)$ with $1$ in the $i$-th slot, $i = 1, 2,\dots, d$, is fixed
as a canonical orthonormal basis in $\CC^d$, the corresponding Fock space describes $d$ field channels coupled to the system. Since the dimension of $\mathcal{Z}$ can be infinite, it allows considering infinitely many field channels. To take advantage of this generality, we take the quantum noise space to be the bosonic Fock space $\Gamma(\mathcal{H})$ over $\mathcal{H} = L^2(\RR^+) \otimes \mathcal{Z}$ in the following.

To introduce quantum probabilistic analogues of stochastic integrals, one needs   an appropriate notion of time to formulate the notion of filtration and adapted processes. We consider the canonical spectral measure $P(\cdot)$ on $\RR^+$, defined by:  
\begin{equation}
(P(E) f)(t) = 1_{E}(t) f(t), 
\end{equation}
for $f \in \mathcal{H}$, where $1_{E}$ denotes the indicator function of a Borel subset $E \subset \RR^+$. One can interpret $P(\cdot)$ as a time observable in the Hilbert space $\mathcal{H}$ which, as a spectral measure, is continuous, i.e. $P(\{t \}) = 0$ for all $t$. 

Note that the $P(\cdot)$ are orthogonal projections, in terms of which a decomposition of the Hilbert space $\mathcal{H}$ as the direct sum of a closed subspace and its orthogonal complement can be obtained. We define:
$$\mathcal{H}_{t]} := Ran(P([0,t])),  \ \ \mathcal{H}_{[s,t]} := Ran(P([s,t])),  \ \ \mathcal{H}_{[t} := Ran(P([t,\infty))),$$ 
where $Ran$ denotes the range. Then for  $0 < t_1 < t_2 < \dots < t_n < \infty$, we have the decomposition: 
$$ \mathcal{H} = \mathcal{H}_{t_{1}]} \oplus \mathcal{H}_{[t_{1},t_{2}]} \oplus \dots \oplus \mathcal{H}_{[t_{n-1} ,t_{n}]}  \oplus \mathcal{H}_{[t_{n}}.$$

Now  let $\mathcal{H}_S$ be a fixed Hilbert space called the {\it initial Hilbert space} and consider the space $\mathcal{F} = \mathcal{H}_S \otimes \Gamma_{s}(\mathcal{H})$. Physically, one view $\mathcal{H}_S$ as describing a system of interest, $\Gamma(\mathcal{H})$ as describing a noise process (modeling, for instance, a heat bath) and $\mathcal{F}$ as the total space for the time evolution of the system in the presence of quantum noise. 

Denote: 
$$\mathcal{F}_{0]}=\mathcal{H}_S,  \ \ \mathcal{F}_{t]} = \mathcal{H}_S \otimes \Gamma(\mathcal{H}_{t]}), \ \ \mathcal{F}_{[s,t]} = \Gamma(\mathcal{H}_{[s,t]}), \ \  \mathcal{F}_{[t} = \Gamma(\mathcal{H}_{[t}). $$
Then, by (v) of Proposition \ref{prop_expvec},  for $0\leq t_{1} < \dots < t_{n} < \infty$, we have the following continuous tensor product factorization over $\RR^+$: 
$$\mathcal{F} = \mathcal{F}_{t_{1}]} \otimes \mathcal{F}_{[t_{1},t_{2}]} \otimes \dots \otimes \mathcal{F}_{[t_{n-1} ,t_{n}]}  \otimes \mathcal{F}_{[t_{n}}.$$ 

The identification above is based on the factorization of the exponential vectors: 
\begin{equation} |\psi \rangle \otimes |e(u) \rangle =  |\psi\rangle \otimes |e(u_{[0,t_1]}) \rangle \otimes  |e(u_{[t_1,t_2]}) \rangle \otimes \dots \otimes  |e(u_{[t_{n-1},t_n]}) \rangle \otimes |e(u_{[t_n,\infty)}) 
\rangle,\end{equation} where $|\psi\rangle \in \mathcal{H}_{S}$, $u_{A}(\tau) = P(A) u (\tau)$. Note that $\mathcal{F}_{t_1]}$, $\mathcal{F}_{[t_j, t_{j+1}]}$ ($j=1, \dots, n-1$) and $\mathcal{F}_{[t_{n}}$ embed naturally into $\mathcal{F}$ as subspaces by tensoring with the vacuum vectors in appropriate sectors of the total space.

The basic idea of H-P quantum stochastic calculus comes from this continuous tensor product factorization property of bosonic Fock space.  One can obtain a quantum analogue of the filtration by generalizing the viewpoint that filtrations in classical probability can be represented as a commutative algebra. 

\begin{defn}
A {\it filtration} $(\mathcal{B}_{t]})_{t \geq 0}$ in $\mathcal{F}$ is an increasing family of (von Neumann) algebras, where 
\begin{equation} \mathcal{B}_{t]} = \{ X \otimes 1_{[t} : X \in \mathcal{B}(\mathcal{F}_{t]}), 1_{[t} \text{ is the identity operator in } \mathcal{F}_{[t}  \}.  \end{equation}
\end{defn}

Roughly speakiing,  a process $\{X_t, t \geq 0\}$ is adapted to the filtration $(\mathcal{B}_{t]})_{t \geq 0}$, if  $X_{t}$ acts trivially on $\mathcal{F}_{[t}$, i.e. such that $X_t = X_{t} \otimes 1_{[t}$ for all $t$.  From now on, we assume the following for simplicity. All the operators in a bosonic Fock space have domains that include the exponential domain. Let $D_0 \subset \mathcal{H}_S$ be a dense subspace, $\mathcal{M} \subset \mathcal{H}$ be  linear manifold such that $P([s,t]) |u\rangle \in \mathcal{M}$ whenever $|u\rangle \in \mathcal{M}$ for every $0 \leq s < t < \infty$. Also, the linear manifold generated by all vectors of the form $f  e(u) := |f\rangle \otimes |e(u)\rangle$, $f:= |f\rangle \in D_0$, $u:=|u\rangle \in \mathcal{M}$, is contained in the domain of any operators in $\mathcal{F}$.   More precisely: 

\begin{defn}
\begin{itemize}
\item[(i)] A family $X = \{X_t : t \geq 0\}$ of operators in $\mathcal{F}$ is an {\it adapted process} if the map $t \to X_t f e(u)$ is measurable and there exists an operator $X'_t$ in $\mathcal{F}_{t]}$ such that 
\begin{equation}
X(t) fe(u) = (X'_t f e(u_{t]})) \otimes e(u_{[t}), 
\end{equation}
for all $t \geq 0$, $f \in D_0$ and $u \in \mathcal{M}$. 
Such an adapted process is called {\it regular} if the map $t \to X_t f e(u)$ is continuous for every $f \in D_0$, $u \in \mathcal{M}$. 
\item[(ii)] A map $m: t \to m_t$ from $\RR^+$ into $\mathcal{F}$ is a {\it martingale} if $m_t \in \mathcal{F}_{t]}$ for every $t$ and $P([0,s]) m_t = m_s$ for all $s < t$. 
\end{itemize}
\end{defn}

\begin{rmk}
Take $\mathcal{H}_S = \CC$, $\mathcal{N}_1 = \mathcal{B}(\mathcal{F})$, $\mathcal{N}_2 = \mathcal{B}(\mathcal{F}_{t]}) \otimes P_{\mathcal{F}_{[t}}$, where $P_{\mathcal{F}_{[t}}$ denotes the set of projections $P([t,\infty))$ into the space $\mathcal{F}_{[t}$. Then for $X \in \mathcal{N}_1$, there exists a unique operator $X' \in \mathcal{N}_2$ such that for all $u, v \in \mathcal{M}$,   $$\langle  e(u_{t]}), X'   e (v_{t]}) \rangle = \langle  e(u_{t]}) \otimes  e(u_{[t}), X  e(v_{t]}) \otimes e(v_{[t}) \rangle $$  
The map $X \mapsto X' \otimes P([t,\infty))$ can thus be viewed as a (quantum) conditional expectation  $E[\cdot|\mathcal{N}_2]$. Therefore, one can view $m_t$, the martingale in a bosonic Fock space, as a martingale in the more familiar form in classical probabilistic sense, i.e. $E(m_t|\mathcal{F}_{s]}) = m_s$ for all $s \leq t$.
\end{rmk}


An important class of martingales is of the form $|u_{t]}\rangle = P([0,t])|u\rangle$ for $|u\rangle \in \mathcal{M}$ (c.f. Example \ref{ex_qprocess}). We now introduce two families of regular, adapted processes  associated with this class of martingales (see other processes that can be studied in \cite{parthasarathy2012introduction}). These are the processes in a bosonic Fock space with respect to which stochastic integrals will be defined later. 

Any vector $|u\rangle \in \mathcal{H}$ may be regarded as a $\mathcal{Z}$-valued function.  For a fixed basis of $\mathcal{Z}$ (e.g. in the case when $\mathcal{Z}$ is the space  $\CC^d$ with the canonical basis $|z_k\rangle$), we set  $u_k(t)  = \langle z_k | u(t) \rangle_{\mathcal{Z}}$ for $k \geq 1$, where $\langle \cdot | \cdot \rangle_{\mathcal{Z}}$ denotes scalar product on $\mathcal{Z}$.  

\begin{defn} Let $\mathcal{H}_S = \CC$ so that $\mathcal{F} = \Gamma(L^2(\RR^+) \otimes \mathcal{Z})$. The {\it creation and annihilation processes} associated with the martingales $\{1_{[0,t]} \otimes |z_{k}\rangle \}_{k\geq 1}$ are linear operators with their domains equal the exponential domain $\mathcal{E}$ and:
\begin{equation} \label{cr_ann} A^{\dagger}_{k}(t) = a^{\dagger}(1_{[0,t]} \otimes z_{k}), \ \ \ \ \  \ A_{k}(t) = a(1_{[0,t]} \otimes z_{k}),\end{equation} for $k=1,2,\dots$, where $1_{[0,t]}$ denotes indicator function of $[0,t]$ as an element of $L^2(\RR^{+})$. 
\end{defn}

Each $A_{k}$ (respectively, $A_{k}^{\dagger}$) is defined on a distinct copy of the Fock space $\Gamma(L^2(\RR^{+}))$ and therefore, the $A_{k}$'s (respectively, $A_k^{\dagger}$) are commuting. Physically, each of them represents a  single channel of quantum noise input coupled to the system. Note that in the special case $\mathcal{Z} = \CC$, the above construction only gives a single pair of creation and annihilation process and in the case $\mathcal{Z} = \CC^{d}$, we have $d$ pairs of creation and annihilation processes associated with $d$ distinct noise inputs. The actions of the $A_k(t)$ on the exponential vectors are given by the eigenvalue relations: 
\begin{equation} \label{reg1}
A_k(t)|e(u)\rangle = \left(\int_0^t u_k(s) ds \right) |e(u)\rangle, 
\end{equation}
and the $A_k^{\dagger}(t)$ are the  corresponding adjoint processes:
\begin{equation} \label{reg2}
\langle e(v) | A_k^{\dagger}(t) | e(u) \rangle = \left(\int_0^t \overline{v_k(s)} ds \right) \langle e(v) | e(u) \rangle. \end{equation}
The above processes, which are time integrated versions of instantaneous creation and annihilation operators are two of the three kinds of {\it fundamental noise processes} introduced by Hudson and Parthasarathy.
They satisfy an integrated version of the CCR: $[A_k(t), A_l^{\dagger}(s)] = \delta_{kl} \mbox{min}(t,s) $, $[A_k(t), A_l(s)] = [A_k^{\dagger}(t), A_l^{\dagger}(s)] = 0$. 

For each $k$, their `future pointing' infinitesimal time increments, $dA^{\#}_k(t):=A_k^{\#}(t+dt) - A^{\#}_k(t)$, where $\#$ denotes either creation or annihilation processes, with respect to the time interval $[t,t+dt]$, are {\it independent} processes. The independence is due to the fact that time increments with respect to non-overlapping time intervals are commuting since they are {\it adapted} with respect to $\mathcal{F}$, i.e. they act non-trivially on the factor $\mathcal{F}_{[t,t+dt]}$ of the space $\mathcal{F} =\mathcal{F}_{t]}\otimes \mathcal{F}_{[t,t+dt]} \otimes \mathcal{F}_{[t+dt}$ and trivially, as identity operator on the remaining two factors. In other words, for a fixed $k$, 
\begin{align}
dA_{k}^{\#}(t) | e(u) \rangle &= (A_k^{\#}(t+dt) - A^{\#}_k(t)) |e(u) \rangle \\ 
&= e(u_{[0,t]}) \otimes  a^{\#}(1_{[t,t+dt]} \otimes z_k)  e(u_{[t,t+dt]}) \otimes  e(u_{[t+dt,\infty)}), 
\end{align}
where the operators $a^{\#}$ are defined in \eqref{cr_ann}. Therefore, any Hermitian noise processes $M(t)$ that are appropriate combinations of the $A_k^{\#}(t)$ (for instance, the quantum Wiener processes introduced later in \eqref{qwp}) have independent time increments, i.e. if we define the characteristic function of $M$ with respect to the coherent state, $|\psi(u) \rangle$, as $\varphi_{M}(\lambda) := \langle \psi(u)| e^{i\lambda M} |\psi(u)\rangle$, then for any two times $s \leq t$, we see that their joint characteristic function with respect to the coherent states is the product of individual characteristic functions: 
\begin{align}
\varphi_{M(s),M(t)-M(s)}(\lambda_s,\lambda_t) &:= \langle \psi(u)| e^{i\lambda_s M(s) + i \lambda_t (M(t)-M(s))} |\psi(u)\rangle \nonumber \\ 
&= \varphi_{M(s)}(\lambda_s) \varphi_{M(t)-M(s)}(\lambda_t). 
\end{align}
This property is a quantum analog of the notion of processes with independent increments in classical probability. 


\begin{rmk} \label{qft}
In quantum field theory, the operators  $A^{\dagger}_{k}(t)$ and $A_{k}(t)$ are called the smeared field operators and are usually written formally as:
\begin{equation} A_k(t) = \int_0^t b_k(s) ds, \ \ \ \  \ \ A_{k}^{\dagger}(t) = \int_0^t b^{\dagger}_k(s) ds, \end{equation}
where the $b_k(t) = \frac{1}{\sqrt{2 \pi}}\int_{\RR} \hat{b}_k(\omega) e^{-i\omega t} d\omega$ and $b_k^{\dagger}(t)=\frac{1}{\sqrt{2 \pi}}\int_{\RR} \hat{b}^{\dagger}_k(\omega) e^{i\omega t} d\omega$ are the idealized Bose field processes  satisfying the singular CCR: $[b_k(t), b^{\dagger}_l(s)] = \delta_{kl} \delta(t-s)$ \citep{GardinerBook}. Physically, since $b^{\dagger}_{k}(s)$ creates a particle at time $s$ through the $k$th noise channel, $A_{k}^{\dagger}(t)$ creates a particle that survives up to time $t$. These formal expressions for the annihilation and creation processes are simpler to work with than the more regular integrated processes defined in \eqref{reg1}-\eqref{reg2}. As remarked on page 39 of \cite{nurdinlinear}, the more fundamental processes from the underlying
physics point of view are the quantum field  processes, not the rigorously defined, more regular integrated processes.  \end{rmk}

Exploiting the structure of bosonic Fock space and the properties of the fundamental noise processes outlined above, Hudson and Parthasarathy developed and studied {\it quantum stochastic integrals} with respect to these fundamental processes  for a suitable class of adapted integrand processes, in analogy with the constructions in the classical It\^o theory. These integrals are, for instance, of the form:
\begin{align}
M_k(t) &:= \int_0^t F_k(s) dA_k(s) + G_k(s)dA_k^\dagger(s) \label{qsi} \\
&= F_k(t_1) \otimes (A_k(t_2 \wedge t)-A_k(t_1 \wedge t))+ G_k(t_1)\otimes (A_k^\dagger(t_2 \wedge t) - A_k^\dagger(t_1 \wedge  t)),
\end{align}
defined on $\mathcal{F}_{t]}$ for $t > 0$, where the adapted operator-valued processes $F_k(s) := F_k(t_1) 1_{[t_1,t_2)}(s)$ and $G_k(s) := G_k(t_1) 1_{[t_1,t_2)}(s)$ are step functions of $s$, $t_1 \wedge t_2$ denotes minimum of $t_1$ and $t_2$,
and $A_k$, $A_k^\dagger$ are the annihilation and creation processes on $\mathcal{F}$. In parallel with the construction in classical It\^o calculus, the above stochastic integral can be extended to include integrands that belong to a large class of adapted operator-valued processes\footnote{In the case where the integrands are unbounded operators, mathematically rigorous studies of these integrals are technically difficult. See, for instance, \cite{fagnola2006quantum}.}. We will only sketch selected important results of the calculus in the following. For rigorous statements (including existence and uniqueness results etc.), see \cite{parthasarathy2012introduction}.

The most important result of the calculus is the {\it quantum It\^o formula}\footnote{The trick to derive quantum It\^o formula is to study product of quantum stochastic integrals sandwiched between coherent states by applying the properties of the fundamental field processes. All the basic ideas can be found in Section 25 of \cite{ parthasarathy2012introduction} or the seminal paper \cite{hudson1984quantum}.}, which describes how the classical Leibnitz formula for the time-differential of a product of two functions gets corrected when these functions depend explicitly on the fundamental processes. In the vacuum state, the quantum It\^o formula can be summarized by:
\begin{equation} \label{qif}
dA_k(t) dA^{\dagger}_l(t) = \delta_{kl} dt \end{equation} and all other products of differentials that involve $dA_k(t)$, $dA_k^{\dagger}(t)$ and $dt$ vanish. This can be viewed as a chain rule with Wick ordering \citep{streater2000classical} and as a quantum analogue of the classical It\^o formula. Also, when the $F_k(s)$ and $G_k(s)$ in \eqref{qsi} are bounded (so that multiplications of operators are free of domain issue), we have:
$$d(M_i M_j) = (dM_i)M_j +M_i(dM_j) + dM_i dM_j,$$
where the It\^o correction term $dM_i dM_j$ is evaluated according to the rule \eqref{qif}.

In particular, with respect to the initial vacuum state, the field quadratures $W^0_k(t) = A_k(t) +A_k^{\dagger}(t)$ ($k=1,2,\dots$) are mean zero Hermitian Gaussian processes with variance $t$. Therefore, they can be viewed as quantum analogue of classical Wiener processes and their formal time derivatives, $dW^0_k(t)/dt = b_k(t) + b^{\dagger}_k(t)$, are quantum analogues of the classical white noises. If one takes $\mathcal{Z} = \CC^{d}$, then $(W^0_1,W^0_2,\dots,W^0_{d})$ is a collection of commuting processes and thus form a quantum analogue of $d$-dimensional classical Wiener process in the vacuum state. Moreover, one has $dW^0_i(t) dW^0_j(t) = \delta_{ij}dt$, which is the classical It\^o correction formula for Wiener process. These results hold for a more general class of field observables: 

\begin{defn} {\bf Quantum Wiener processes.} Let $k$ be a positive integer. For $\theta_k \in \RR$ (phase angle), we call the following operator-valued processes on the bosonic Fock space $\Gamma(\mathcal{H})$:
\begin{equation} \label{qwp}
W^\theta_k(t) = e^{-i\theta_k } A_k(t) + e^{i\theta_k} A^{\dagger}_k(t).
\end{equation}
{\it quantum Wiener processes}.
\end{defn}
Quantum Wiener processes are quantum analogue of the classical Wiener processes and they satisfy the quantum It\^o formula: 
$$dW_i^\theta(t) dW_j^\theta(t) =\delta_{ij} dt,$$
for every $\theta \in \RR$. Stochastic integrals with respect to the quantum Wiener processes can therefore be viewed as quantum analogues of the It\^o stochastic integral in the classical theory. The classical It\^o theory of stochastic calculus is included in the quantum calculus by using the Wiener-Segal idenfification of the bosonic Fock space with the $L_2$ space of Wiener process.

\begin{rmk} Following Remark \ref{qft}, one can introduce the notion of {\it quantum colored noise} \citep{belavkin1995world,xue2017modelling}. For $k=1,2,\dots$, define \begin{equation} b^{\dagger}_{g,k}(t) := \frac{1}{\sqrt{2 \pi}} \int_{\RR} \hat{b}_k^{\dagger}(\omega) e^{i\omega t} \hat{g}(\omega) d\omega, \ \ \ \ \ \  \ b_{g,k}(t) := \frac{1}{\sqrt{2 \pi}} \int_{\RR} \hat{b}_k(\omega) e^{-i\omega t} \overline{\hat{g}(\omega)} d\omega,\end{equation}
where $\hat{g}(\omega)$ and $\hat{b}_k(\omega)$ denote the Fourier transform of $g(t)$ and $b_k(t)$ respectively. Note that $b^{\dagger}_{g,k}(t)$ is the inverse Fourier transform of $\hat{b}_k^{\dagger}(\omega) \hat{g}(\omega)$ and so by the convolution theorem we have: \begin{equation} b^{\dagger}_{g,k}(t) = \frac{1}{\sqrt{2 \pi}} \int_{\RR} g(t-s) b_k^{\dagger}(s) ds = \frac{1}{\sqrt{2 \pi}} \int_{\RR} g(t-s) dA_k^{\dagger}(s),\end{equation} where the $A_{k}^{\dagger}(s)$ are creation processes. This is reminiscent of the formula for classical colored noise defined via filtering of white noise \citep{lindgren2006lectures}: 
\begin{equation} \int_{\RR} \gamma(t-s)1_{\{s\leq t \}}(s) dB_{s} = \int_{-\infty}^t \gamma(t-s) dB_{s},  \end{equation} where $\gamma(t)$ describes the filter and $B=(B_s)$ is a classical Wiener process. In the limit $g \to 1$ (flat spectrum limit), the $b^{\dagger}_{g,k}(t)$  converge to the fundamental noise process $b_k^{\dagger}(t) = dA_k^{\dagger}/dt$. Similar remarks apply to $b_{g,k}(t)$ and to appropriate linear combinations of  $b_{g,k}(t)$ and  $b^{\dagger}_{g,k}(t)$, which therefore deserve to be called quantum colored noise processes. \end{rmk}






\section{Open Quantum Systems}
In reality, no quantum system is completely isolated from its surrounding. In other words, every quantum system (denoted $S$) is intrinsically in contact with its surroundings (environment or bath), denoted $B$, which usually has infinitely many degrees of freedom. This intrinsic openness of quantum systems is in contrast to classical systems. Similarly to the classical case, a popular approach to study open quantum systems are models of system-reservoir type. 
If $\mathcal{H}_S$ denotes the Hilbert space associated with the system, and $\mathcal{H}_B$ the Hilbert space associated with the bath, then the formal Hamiltonian for the total system, which takes into account the interaction between the system and bath, is 
\begin{equation}
H = H_S \otimes I + I \otimes H_B + H_I
\end{equation}
 on the total space $\mathcal{H}_S \otimes \mathcal{H}_B$. In the above, $H_S$ is the Hamiltonian describing the system, $H_B$ is the Hamiltonian describing the bath, $H_I$ is the Hamiltonian specifying the their interaction, and $I$ denotes identity operator on an understood space. The effective action of the reservoir on the system can often be modeled as a quantum noise. This line of thinking has been fruitful to study many concrete open quantum systems and has attracted increasing interest these days, particularly from researchers working in the field of quantum optics \citep{GardinerBook} and quantum information \citep{nielsen2010quantum,ohya2011mathematical}. Useful formalisms to study open quantum systems are fundamental, given the  unprecedented\footnote{For instance, Haroche and Wineland have developed an experiment to study the quantum mechanics of light trapped between two mirrors. They show that the quantum of light -- the photon -- can be controlled with atoms at an astonishing level of precision. Their work was recognized by a Nobel prize in 2012 \cite{haroche2013nobel}. Their experiments were explained using models and methods of open quantum systems in \cite{rouchon2014models}. } progress in the development of techniques to measure and manipulate  quantum systems while keeping their essential quantum features.

Usually one is only interested in the dynamics of the open system which is, say, accessible for measurement. A variety of methods and formalisms has been developed and employed to study the open system. The standard formalism for the investigation of the dynamics of the open system is the {\it master equation}, which describes
the evolution of the reduced density matrix of the open system, obtained by taking the partial trace over the degrees of freedom of the environment, i.e. the reduced density matrix is of the form $\rho_S(t) = Tr_B(\rho(t))$, where $\rho(t) = e^{-iHt/\hbar} (\rho_S \otimes \rho_B) e^{iHt/\hbar}$ is the density matrix of the total system and $Tr_B$ denotes the partial trace. Note that the initial total density matrix, $\rho(0) = \rho_S \otimes \rho_B$, is in a factorized form and evolves unitarily, but the reduced density matrix $\rho_S(t)$ generally does not evolves unitarily.  
 
Early works in the open quantum systems literature focused on Markovian\footnote{Non-Markovian descriptions, in particular characterization and quantification of non-Markovianity, are equally important and in fact are topics of active research these days \citep{de2017dynamics,bylicka2013non}.}  descriptions of an open quantum system. These descriptions provide approximate but reasonably accurate and tractable models for many open quantum systems. There are two approaches in such description: the axiomatic approach and the constructive approach.  

In the constructive approach, one tries to derive the mathematically correct form of Markovian master equations from first principles. The starting point is a microscopic model of system-reservoir type described above.  One obtains Markovian master equation as an approximation to the exact reduced dynamics of the open system. In the physical theory of open quantum systems, one often derives the so-called Born-Markov master equation for the reduced density operator of the system under uncontrolled approximations \citep{SchlosshauerBook}. This leads to a breakdown of positivity of the density operator describing the state of the particle associated to a violation of the Heisenberg's uncertainty principle. In order to free of such problem, one naive remedy is to add a correction term, justifiable in certain parameter regimes, to bring the master equation into a physically correct Markovian form (the Lindblad form that we will discuss below) \citep{lampo2016lindblad}. On the other hand, in the mathematical theory of open quantum systems, one derive the physically correct Markovian form from a microscopic model by taking well justified limits such as the weak coupling limits \citep{gough2006quantum,Der_notes} (or the stochastic limits \citep{accardi2002quantum}), repeated interaction limits \citep{attal2006repeated}, among others \citep{bouten2015trotter}.

The axiomatic approach focus on deriving the mathematically correct form of Markovian master equations from the theory of completely positive maps \citep{gorini1978properties}. Mathematical properties of the master equations and their solutions are studied. This leads to the theory of quantum dynamical semigroups and their dilations\footnote{It is well known from the theory of one-parameter semigroups that a one-parameter contraction semigroup on a Hilbert space can be expressed as a reduction of a unitary group \citep{nagy2010harmonic}. This unitary group is called the dilation of the semigroup. This fact can be generalized to a quantum dynamical semigroup.}.

We now briefly discuss the axiomatic\footnote{We refer to \cite{Rivas2012} for a nice mathematical introduction.} approach. A key observation is that any physically reasonable dynamical map, $\Lambda$, must map physical states to physical states, including those states of the system considered as part of a larger system. This is captured by the notion of complete positivity, which is stronger than mere positivity. For instance, the transpose map on matrices is positive but not completely positive. For precise definition, see \cite{alicki1987general}. Completely positive maps were studied already in the 50s and the celebrated Stinespring representation theorem leads to a general form of completely positive dynamical map, called the Kraus decomposition (or operator-sum representation): 
\begin{equation}
\Lambda \rho = \sum_k W_k \rho W_k^*,
\end{equation}
where $\rho$ is a state (for instance the reduced density operator $\rho_S(t)$) and the $W_k$ are (Kraus) operators such that $\sum_k W_k^* W_k = I$.

On the other hand, unity preservation of states is important to allow probabilistic interpretation. These lead to the notion of quantum dynamical semigroup (or quantum  Markov semigroup), i.e. a dynamical map $\Lambda_t$ which is completely positive and preserves the unity (trace-preserving) for all times $t \geq 0$. If $\Lambda_t$ is continuous, then we can define its generator $\mathcal{L}^*$ such that $\Lambda_t = e^{t \mathcal{L}^*}$ and write down the master equation for $\rho(t)$.
The celebrated result of Lindblad \citep{lindblad1976generators} (and also of Gorini, Kossakowski and Sudarshan) provides the most general form of such master equations, called the {\it Lindblad master equations (LMEs)} with a bounded\footnote{For open systems with infinite dimensional Hilbert spaces, the generators of the quantum dynamical semigroup are generally unbounded. However, in this case the Lindblad form often makes sense. Then the operators $H_e$ and $L_k$ can be unbounded and the sum over $k$ can be replaced by an integral. To our knowledge, there is only one paper \citep{siemon2017unbounded}  that studies the case of unbounded generator.} generator:
\begin{equation} \label{lme}
\frac{d}{dt} \rho(t) = \mathcal{L}^*(\rho) = -\frac{i}{\hbar}[H_e, \rho]  + \sum_{k}  L_{k} \rho L_{k}^{*} - \frac{1}{2} \sum_{k } \{ L_{k}^{*} L_{k}, \rho \},
\end{equation}
where $H_e = H_e^*$  and the $L_k$ ({\it Lindblad operators}) are bounded  operators. The right hand side of the LME above consists of three contributions. The term $-\frac{i}{\hbar}[H_e, \rho]$ gives the unitary contribution, the term $\sum_{k}  L_{k} \rho L_{k}^{*}$ can be interpreted as quantum jumps and the term $\frac{1}{2} \sum_{k } \{ L_{k}^{*} L_{k}, \rho \}$ represents dissipation. It is the dissipation term that make the solution of LME non-unitary. The choice of bounded operators $H_e$ and the $L_k$ is not unique and the sum over $k$ can be replaced by an integral. 

The LME can be viewed as quantum analogue of the Fokker-Planck equation for transition probability density.  The result of Lindblad\footnote{It is remarkable that the Lindblad form can be derived in analogy with the classical Markovian conditions on the generator $Q$ of the stochastic matrix $e^{tQ}$ for classical Markov process, at least in finite dimensional systems. This was done by Kossakowski, who arrived at a set of equivalent conditions in the quantum case \citep{Rivas2012}. } et. al. above says that the semigroup $\Lambda_t = e^{\mathcal{L}^*t}$, where $\mathcal{L}^*$ is given in \eqref{lme}, is a quantum dynamical semigroup. Quantum dynamical semigroups are  generalization of classical Markov semigroups  to the quantum setting. 

Hudson-Parthasarathy (H-P) theory of stochastic integration produces a dilation of quantum dynamical semigroups via their quantum stochastic differential equations (QSDEs).  In contrast to closed quantum system's unitary evolution, interaction with an environment leads to randomness in the unitary evolution of an open quantum system. Using the quantum It\^o formula,  the general form of a unitary, reversible, Markovian evolution for a system interacting with an environment described by the fundamental noise processes can be deduced. The unitary evolution operator, $V(t)$, of the whole system, in the interaction picture with respect to the free field dynamics, is found to satisfy an It\^o SDE of the following form:
\begin{align}
dV(t) &= \left[ \left( -\frac{i}{\hbar} H_e -  \frac{1}{2} \sum_{k} L^{\dagger}_k L_k \right) dt + \sum_{k} \left( L_k^{\dagger} dA_k(t) - L_k dA_k^{\dagger}(t) \right) \right] V(t), \\
V(0) &= I, 
\end{align}
associated to the system operators $(H_e, \{L_k\})$, where $H_e = H_e^{\dagger}$ is an effective Hamiltonian and the $L_k$ are Lindblad coupling operators. It can be viewed as a noisy Schrodinger equation\footnote{One can also derive SDEs for the wave function; see \cite{parthasarathy2017quantum,barchielli2009quantum}.}. The choice of the operators $(H_e, \{L_k\})$ depends on physical systems on hand.

The evolution of a noisy system observable, $X$, initially defined on $\mathcal{H}_{S}$, can also be obtained. By applying the quantum It\^o formula, one can deduce that its evolution, $j_t(X) = V(t)^{\dagger} (X \otimes I) V(t)$ (Evans-Hudson flow), on $\mathcal{F}$ is described by the following Heisenberg-Langevin equation\footnote{This is often formulated in the so-called SLH framework in the quantum control and modeling literature \citep{xue2017modelling}.} :
\begin{align}
dj_t(X) &= j_t(\mathcal{L}(X))dt + \sum_k \bigg( j_t([X,L_k])dA_k^{\dagger}(t) + j_t([L_k^{\dagger},X])dA_k(t) \bigg), \\ 
j_{0}(X) &= X \otimes I, \end{align}
where $\mathcal{L}$ is the Lindblad generator:
\begin{equation}
\mathcal{L}(X) = \frac{i}{\hbar}[H_e,X] + \frac{1}{2} \sum_k ([L_k^{\dagger},X]L_k + L_k^{\dagger} [X,L_k] ). \end{equation}

One can also obtain the evolution of field observables in  this way and study the relation between input and output field processes \citep{GardinerBook}. We call such equation for an observable a {\it quantum stochastic differential equation} (QSDE) and its solution is a {\it quantum stochastic process}, which is a noncommutative analogue of classical stochastic process. 

To obtain the Lindblad master equation (LME) for reduced system density operator, $\rho_S(t)$, we first take the Fock vacuum conditional expectation of $j_t(X)$ to obtain the evolution of the reduced system observable, $T_t(X)$, defined via \begin{equation} \langle \psi| T_t(X) | \phi \rangle = \langle \psi \otimes e(u) | j_t(X) | \phi \otimes e(v) \rangle, \end{equation} so that $dT_{t}(X) = T_{t}(\mathcal{L}(X)) dt$,  then the Lindblad master equation:
\begin{equation} d\rho_S(t) = \mathcal{L}^{*}(\rho_{S}(t)) dt\end{equation} is obtained by duality. In this way, one sees that the Evans-Hudson flow is a dilation of the quantum dynamical semigroup.
The equation for $T_t(X)$ above can be seen as quantum analogue of the backward Kolmogorov equation. Indeed, for simple instances, one sees that the restriction of $\mathcal{L}$ to a commutative algebra coincides with infinitesimal generator of classical Markov processes (see Proposition 3.2 in \cite{fagnola1999view}).

The next example illustrate how one can apply the above formalism to study models in quantum optics. For more examples, see the text \cite{nurdinlinear} or the recent review paper \cite{combes2017slh}.

\begin{ex}{\bf Two-level Atom Interacting with a Radiation Field \citep{haroche2006exploring}.}
Consider the small system to be a two-level atom, which is described by a Hilbert space $\mathcal{H}_S = \CC^{2}$ and the Hamiltonian $H_e = \hbar \Omega \sigma_{+} \sigma_{-}$, where $\sigma_{+}$ and $\sigma_{-}$ are the raising and lowering operator respectively.  It interacts with a radiation field in equilibrium at a temperature $T$. The evolution of the atom can be described effectively by two Lindblad operators $L_{1} = \sqrt{\gamma(N+1)} \sigma_{-}$ and $L_{2} = \sqrt{\gamma N} \sigma_{+}$, which describe the energy exchanges between the atom and the field. The first Lindblad operator describes the processes of spontaneous and stimulated emission, where the atom loses energy into the field, while the second one describes the absorption, where the atom gains energy from the field. Here $\gamma$ is the rate at which the atom loses or gains energy when the radiation field is at the temperature $T$ and $N$ is the mean number of photons in the radiation field at the resonant frequency $\Omega$. Using the above physical choice of $H_e$ and $L_k$ ($k=1,2$), one can write down an equation to describe evolution of dynamical variables of interest. 
\end{ex}

\begin{ex} {\bf A Lindblad model of damped quantum harmonic oscillator.} Let $X$ and $P$ denotes the position and momentum operator on $\mathcal{H}_S$. Let $\gamma > 0$ be constant and $T>0$ be the temperature. We take $$H_{e} =  \frac{P^{2}}{2M} + \frac{1}{2} kX^{2} + \frac{ \gamma}{M} \{X, P\}$$  and a single Lindblad operator $$L = \frac{1}{\hbar} \sqrt{4 k_{B} T \gamma} X + \frac{i}{M} \sqrt{\frac{\gamma}{4 k_{B} T}} P.$$ 
Then we can write down the QSDE for $X_{t} = j_t(X)$ and $P_{t} = j_t(P)$, where $I$ is identity operator on the boson Fock space $\mathcal{F}$:
\begin{align}
dX_{t} &= \frac{P_{t}}{M} dt - \frac{\hbar}{M}  \sqrt{\frac{\gamma}{4 k_{B} T}} dB^{(1)}_{t}, \\
dP_{t} &= -kX_{t} dt - \frac{2 \gamma}{M} P_{t} dt + \sqrt{4 k_{B} T \gamma} dB^{(2)}_{t},
\end{align}
where $$B^{(1)}_{t} = A^{\dagger}(t) + A(t), \ \ B^{(2)}_{t} = i(A^{\dagger}(t) - A(t))$$ are (noncommuting) quantum Wiener processes (recall $(A(t), A^\dagger(t))$ are the fundamental operator processes of Hudson-Parthasarathy). Therefore, our choice of $(H_e, L)$ gives quantum analogue of the classical Langevin equation, modulo the appearance of a quantum noise term in the equation for $X_t$.   
\end{ex}


\subsection{A Hamiltonian Model for Open Quantum Systems} 
Perhaps a straightforward model for open quantum system is a quantized version of the open system considered in the previous section, in which case the Hamiltonian \eqref{Hamiltonian_cos} becomes an operator. We will see that indeed this is the case by studying the descriptions for a system, its environment and their interaction in the quantum mechanical setting. Basic references on these are \cite{GardinerBook,BreuerBook,SchlosshauerBook} (for physics) and \cite{attal2006open} (for mathematics). 

Before we describe our model, we provide some physical motivations. An important class of open quantum systems is the
{\it quantum Brownian motion (QBM)} \citep{BreuerBook,GardinerBook,Rivas2012}. In the standard form, the model for QBM consists of a particle moving in one spatial dimension and interacting linearly with an environment in thermal equilibrium. One candidate model for QBM is the Caldeira-Leggett model \citep{Caldeira1983a}, a prototype
of microscopic Hamiltonian model where the environment is modeled by a collection of non-interacting harmonic oscillators. Such model has been used widely to study decoherence \citep{SchlosshauerBook} (a process in which quantum coherence is lost and the quantum system is brought into
a classical state)  and quantum dissipation phenomena \citep{Weiss}. A detailed study of QBM, in particular the memory effects and modeling of the environment by quantum noises, is important to understand, for instance, how one could exploit the interaction with the environment
to design efficient quantum thermal machines \citep{goychuk2016molecular}  as well as to create entanglement and superpositions of quantum states \citep{Plenio2002,Kraus2008}. 

The original Caldeira-Leggett model is not realistic from experimental point of view, as it does not take into account the spatial
inhomogeneity of the environment. Spatial inhomogeneity occurs, for instance, in the setup of a quantum impurity particle interacting with Bose-Einstein condensates (BECs) \citep{Massignan2015}, where the inhomogeneity is due to a harmonic potential trapping the particle. One would like to have a generalized model that takes into account such inhomogeneity and studies in detail its quantum dynamics.

We consider an open system where the quantum Brownian particle is coupled to an equilibrium heat bath. The particle interacts with the heat bath via a coupling, which is a function that can be nonlinear in the system's position, in which case the particle is subject to inhomogeneous damping and diffusion \citep{Weiss,Massignan2015,lampo2016lindblad}.  The model can be viewed as a field version of the generalized Caldeira-Leggett model studied in \cite{Massignan2015,lampo2016lindblad}, a generalization of the  spinless thermal Pauli-Fierz Hamiltonian with dipole type interactions \citep{pauli1938theorie, derezinski1999asymptotic,derezinski2001spectral}, or a quantum analog of the classical Hamiltonian field model. It is a fundamental model which not only allows simple analytic treatments and provides physical insights, but also realistically models many open qantum systems --- for instance, an atom in an electromagnetic field. 

As the heat bath is an infinitely extended quantum system made up of identitical particles on a bosonic Fock space, the formalism of second quantization is convenient for its description. The idea is that given an operator $J$ from a Hilbert space $\mathcal{H}$ to another Hilbert space $\mathcal{K}$, we can extend it naturally to an operator $\Gamma(J)$ from the boson Fock space $\Gamma(\mathcal{H})$ to the boson Fock space $\Gamma(\mathcal{K})$.  More precisely, if $J$ is the given operator on $\mathcal{H}$, so is $J^{\otimes n}$ on the $n$-particle space $\mathcal{H}^{\circ n}$ for every $n$. Therefore, the operator $\Gamma(J)$, called the {\it second quantization} of $J$, defined by
\begin{equation}
\Gamma(J)(u_1 \circ \cdots \circ u_n) = Ju_1 \circ \cdots \circ Ju_n,
\end{equation}
for $n \in \NN$, 
or formally,
\begin{equation} 
\Gamma(J) = I \oplus J \oplus J^{\otimes 2} \oplus \dots \oplus J^{\otimes n} \oplus \dots
\end{equation}
on $\Gamma(\mathcal{H})$ is an operator satisfying $\Gamma(J) |e(u)\rangle = |e(Ju)\rangle$ for every $|u\rangle \in \mathcal{H}$. The identity-preserving correspondence $J \mapsto \Gamma(J)$ is called the {\it second quantization map} and satisfies:
\begin{equation}
\Gamma(J^*) = \Gamma(J)^*,  \  \ \  \Gamma(J_1 J_2) = \Gamma(J_1) 
\end{equation}

Therefore, $\Gamma(J)$ is a self-adjoint, positive, projection or unitary operator whenever $J$ is. In particular, if if $(U_t: t \in \RR)$ is a strongly continuous one-parameter group of unitary operators, then so is $(\Gamma(U_t): t \in \RR)$. In this case, if $U_t = e^{-itH}$ for some self-adjoint operator $H$ on $\mathcal{H}$, then $\Gamma(U_t) = e^{-itH'}$ is a strongly continuous unitary group on $\Gamma(\mathcal{H})$ generated by a self-adjoint operator $H'$. We denote its generator $H'$ as $d\Gamma(H)$ and call it the {\it differential second quantization} of $H$, whose action on the $n$-particles subspace is given by:
\begin{equation}
d\Gamma(H)(u_1 \circ \cdots \circ u_n) = \sum_{k=1}^n u_1 \circ \cdots \circ H u_k \circ \cdots \circ u_n.
\end{equation}
In the special case when $H = I$ (identity operator), $d\Gamma(I)$ is called the number operator. Note that $\Gamma(e^{-itH}) = e^{-it d\Gamma(H)}$.

We describe the particle, the heat bath and their interaction in the model more precisely in the following. The Brownian particle is a quantum mechanical system, denoted $\mathcal{S}$, with energy operator $H_{S}$ on the  Hilbert space $\mathcal{H}_{S} := L^2(\RR)$. It  is subjected to a confining, smooth potential $U(X)$.  The infinite heat bath\footnote{For a rigorous introduction to the heat bath (ideal quantum gas), we refer to the lecture notes \cite{merkli2006ideal}. We will not pursue the rigorous approach here.}, denoted $\mathcal{B}$, is a field of mass-less bosons at a positive temperature. It is described by the triple $(\mathcal{H}_{B}, \rho_{\beta}, H_{B} )$, where $\mathcal{H}_{B} := \Gamma(L^2(\RR^+))$, $\RR^+ = [0,\infty)$, is the bosonic Fock space over $L^2(\RR^+)$ (momentum space), $H_{B}$ is the Hamiltonian of the heat bath defined on $\mathcal{H}_{B}$ and $\rho_{\beta} = e^{-\beta H_B}/Tr(e^{-\beta H_B})$ is the Gibbs thermal state at an inverse temperature $\beta = 1/(k_{B}T)$. We take $H_{B} = d\Gamma(H_B^1)$, the differential second quantization of the energy operator $H_{B}^{1}$ which acts in the one-particle frequency space $L^2(\RR^+)$ as: \begin{equation}(H_{B}^{1} \phi)(w) = \epsilon(w) \phi(w),\end{equation} 
where  $\epsilon(w)$ is the energy of a boson with frequency $w \in \RR^{+}$. The function $\epsilon(w)$ is the dispersion relation for the bath, which in our case, is a linear one, i.e. $\epsilon(w) = \hbar w$.  The equilibrium frequency distribution of bosons at an inverse temperature $\beta$ is given by the Planck's law:  
\begin{equation}\nu_\beta(w) = \frac{1}{\exp{\left( \beta \epsilon(w) \right)}-1}.\end{equation}

The full dynamics of the model is described by the Hamiltonian: 
\begin{equation}
H = H_{S} \otimes I + I \otimes H_{B} + H_{I} + H_{ren} \otimes I, \end{equation} where $H_S$ and $H_B$ are Hamiltonians for the particle and the heat bath respectively, given by
\begin{equation}
H_{S} = \frac{P^2}{2m} + U(X), \ \ \ H_{B} = \int_{\RR^{+}} \hbar \omega   b^{\dagger}(\omega)b(\omega) d\omega,\end{equation}
$H_{I}$ is the interaction Hamiltonian given by
\begin{equation}
H_{I} = - f(X) \otimes \int_{\RR^{+}} [c(\omega) b^{\dagger}(\omega)+\overline{c(\omega)}b(\omega) ] d\omega, \end{equation} 
and $H_{ren}$ is the renormalization Hamiltonian given by 
\begin{equation}
H_{ren} = \left( \int_{\RR^{+}} \frac{|c(\omega)|^2}{\hbar \omega} d\omega \right) f(X)^2. \end{equation}
Here $X$ and $P$ are the particle's position and momentum operators,  $m$ is the mass of the particle, $U(X)$ is a smooth confining potential, $b(\omega) $ and $b^{\dagger}(\omega)$ are the bosonic annihilation and creation operator of the boson of frequency $\omega$ respectively on $L^2(\RR^+)$ and they satisfy the usual canonical commutation relations (CCR): $[b(\omega),b^{\dagger}(\omega')] = \delta(\omega-\omega'), \ \ [b(\omega),b(\omega')]=[b^{\dagger}(\omega), b^{\dagger}(\omega')] = 0.$  We assume that the operator-valued function $f(X)$ is positive and can be expanded in a power series, and $c(\omega)$ is a complex-valued coupling function (form factor) that specifies the strength of the interaction with each frequency of the bath.  It determines the spectral density of the bath and therefore the model for damping and diffusion of the particle. The heat bath is initially in the Gibbs thermal state, $\rho_{\beta} = e^{-\beta H_B}/Tr(e^{-\beta H_B})$,  at an inverse temperature $\beta = 1/(k_{B}T)$. We will refer to the model specified by the above Hamiltonian as the {\it QBM model}. 

The renormalization potential $H_{ren}$ is needed to ensure that the bare potential acting on the particle is $U(X)$ and that the Hamiltonian can be written in a positively defined form: $H=H_{S} \otimes I + H_{B-I}$, where $H_{B-I}$ is given by
\begin{equation} \label{trans_inv_ch6}
H_{B-I} = \int_{\RR^{+}} \hbar \omega \left(b(\omega)-\frac{c(\omega)}{\hbar \omega} f(X) \right)^{\dagger} \left(b(\omega)-\frac{c(\omega)}{\hbar \omega} f(X) \right) d\omega. \end{equation}

Lastly, we discuss the Gibbs thermal state $\rho_\beta$, in particular the derivation of quantum fluctuation-dissipation relation assuming that it is the initial state. Assume in the following that $\rho_\beta$ is of trace class. In the case of thermodynamic limit (i.e. when the limit to an infinitely extended system with infinite volume and infinitely many degrees of freedom is already passed to), the $\rho_\beta$ has infinite trace but the results derived below can still be made sense of \citep{merkli2006ideal}. 

Let $A$ and $B$ be two observables on $\mathcal{H}_B$. Denote $\tau_t(A) = A(t)= e^{i  H_B t/\hbar} A e^{-i H_B t/\hbar}$ and similarly for $\tau_t(B)$. Then, 
\begin{equation}
\langle A \tau_t(B) \rangle_\beta = \frac{Tr(A e^{i  H_B t/\hbar} B e^{-(\hbar \beta + i t)  H_B/\hbar})}{Tr(e^{-\beta H_B})} = \frac{Tr(B e^{-(\hbar \beta + i t)  H_B/\hbar}A e^{i  H_B t/\hbar}  )}{Tr(e^{-\beta H_B})},
\end{equation}
where we have used cyclicity of trace in the last line above. Taking the boundary value at $t = i \hbar \beta$, we have $\langle A \tau_t(B) \rangle_\beta \large|_{t = i \hbar \beta} = \langle B A \rangle_\beta$. This is the {\it Kubo-Martin-Swinger (KMS) condition}, which completely
characterizes the expectation $\langle \cdot \rangle_\beta$ and so gives an alternate definition of equilibrium states.  In particular, setting $A = B  = q$ (for instance, the position observable) and assuming $\langle q \rangle_\beta = 0$, 
\begin{equation}
C^{-}(t) := \langle q q(t) \rangle_{\beta} = \langle q(t) q(i\hbar \beta) \rangle_\beta =  \langle q(t-i\hbar \beta) q \rangle_\beta =: C^{+}(t-i \hbar \beta),
\end{equation}
where we have used the time translation invariance of correlation function in the last line above. Taking the Fourier transform gives 
\begin{equation} \label{ft_r}
\tilde{C}^{-}(\omega) = \tilde{C}^{+}(\omega) e^{-\hbar \omega \beta},
\end{equation}
where $\tilde{F}$ denotes Fourier transform of $F$. 

We now derive quantum fluctuation-dissipation relation of Callen and Welton \citep{callen1951irreversibility}. Set $C^{\pm}(t) = S(t) + i A(t)$, where $S(t) = \langle \{q(t),q(0)\}/2 \rangle_\beta $ is the symmetric correlation function and $A(t) = -i \langle [q(t),q(0)]/2 \rangle_\beta$ is the anti-symmetric correlation function. Similarly, set $\tilde{C}^{\pm}(\omega) = \tilde{S}(\omega) + i \tilde{A}(\omega)$ for its Fourier transform.  

Next we recall some notions from linear response theory. Define the response function (or generalized susceptibility) $\chi(t) = -2 \theta(t) A(t)/\hbar$, where $\theta(t)$ is the step function specifying causality, and define the dynamical susceptibility, $\chi''(\omega)$,  as the imaginary part of $\tilde{\chi}(\omega)$. Then $\chi''(\omega) = i \tilde{A}(\omega)/\hbar = (\tilde{C}^+(\omega) - \tilde{C}^-(\omega))/(2\hbar)$. Then using \eqref{ft_r}, we obtain the result of Callen-Welton:
\begin{equation}
\tilde{\chi}''(\omega) = \frac{1}{2\hbar} (1-e^{-\hbar \omega \beta}) \tilde{C}^+(\omega).
\end{equation}
This result implies that the symmetric correlation function the observable is related (in the Fourier domain) to the anti-symmetric correlation function as:
\begin{equation}
\tilde{S}(\omega) = i \coth(\hbar \omega \beta/2) \tilde{A}(\omega) = \hbar \coth(\hbar \omega \beta/2) \tilde{\chi}''(\omega) .
\end{equation}

\subsection{Heisenberg-Langevin Equations} \label{approach}

In this section, following \cite{lim2018small} we derive the Heisenberg equations of motion for the QBM model and  study  the stochastic force term appearing in the equation. This will pave the way to model the action of the heat bath on the particle by appropriate quantum colored noises introduced in the next sections. Our final goal is the construction of dissipative non-Markovian Heisenberg-Langevin equations driven by appropriate thermal noises, which are built from H-P fundamental noise processes. From now on, $I$ denotes identity operator on an understood space and $1_{A}$ denotes  indicator function of the set $A$. 

Recall that the particle's position evolves according to $\tau_t(X \otimes I) =: X(t)$ and momentum evolves according to $\tau_t(P \otimes I) =: P(t)$, where $\tau_t(\vecc{O}) = e^{i H t/\hbar} \vecc{O} e^{-i H t/\hbar}$ for an observable $\vecc{O}$ of the total system.  
Define the particle's velocity, $V(t) = \frac{P(t)}{m}$ and note that $f'(X) = -i[f(X),P]/\hbar$.   Let
\begin{equation}b(\omega) = \sqrt{\frac{\omega}{2\hbar}}\left(x(\omega)+\frac{i}{\omega}p(\omega) \right), \ \ \ \  b^{\dagger}(\omega) = \sqrt{\frac{\omega}{2\hbar}}\left(x(\omega)-\frac{i}{\omega}p(\omega) \right),  \end{equation}
\begin{equation}
[x(\omega),p(\omega')] = i\hbar \delta(\omega-\omega')I,\end{equation} where we have normalized the masses of all bath oscillators. 

The Heisenberg equation of motion gives \begin{align} \dot{X}(t) &= \frac{i}{\hbar} [H, X(t)] = \frac{P(t)}{m}, \\ 
\dot{P}(t) &= \frac{i}{\hbar} [H, P(t)] \nonumber \\ 
&= -U'(X(t)) + f'(X(t)) \int_{\RR^{+}} d\omega c(\omega) \sqrt{\frac{2\omega}{\hbar}} x_{t}(\omega) - 2f(X(t)) f'(X(t)) \int_{\RR^{+}} r(\omega) d\omega, \\  \dot{x}_{t}(\omega) &= \frac{i}{\hbar} [H, x_{t}(\omega)] = p_{t}(\omega), \ \ \omega \in \RR^{+}, \\  \dot{p}_{t}(\omega) &= \frac{i}{\hbar} [H, p_{t}(\omega)] = - \omega^2 x_{t}(\omega) + \sqrt{\frac{2\omega}{\hbar}} c(\omega) f(X(t)), \ \ \omega \in \RR^{+},  \end{align} where $r(\omega) = |c(\omega)|^2/(\hbar \omega)$ and $f'(X)= [f(X),P ]/(i\hbar)$. 

Next we eliminate the bath degrees of freedom from the equations for $X(t)$ and $P(t)$. Solving for $x_{t}(\omega)$, $\omega \in \RR^{+}$, gives: \begin{equation}x_{t}(\omega) = \underbrace{x_{0}(\omega) \cos(\omega t) + p_{0}(\omega) \frac{\sin(\omega  t)}{\omega}}_{x^{0}_{t}(\omega)} + \int_{0}^{t} \frac{\sin( \omega (t-s))}{\omega} \sqrt{\frac{2\omega}{\hbar}} c(\omega) f(X(s)) ds. \end{equation} Substituting this into the equation for $P(t)$ results in:  \begin{align} \dot{P}(t) &= -U'(X(t)) + f'(X(t)) \int_{\RR^{+}} d\omega  c(\omega) \sqrt{\frac{2\omega}{\hbar}} x^{0}_{t}(\omega) \nonumber \\ 
&\ \ \  +  \frac{2}{\hbar} f'(X(t)) \int_{\RR^{+}} d\omega  |c(\omega)|^2  \int_{0}^{t} ds \sin(\omega  (t-s)) f(X(s)) - 2 f(X(t)) f'(X(t)) \int_{\RR^{+}} d\omega r(\omega). \end{align} 

Using integration by parts, we obtain \begin{equation}\int_{0}^{t} ds \sin( \omega(t-s)) f(X(s)) = \frac{f(X(t))}{\omega} - f(X)\frac{\cos(\omega t)}{\omega} - \int_{0}^{t} \frac{\cos(\omega(t-s))}{\omega} \frac{d}{ds}\left(f(X(s)) \right)ds   \end{equation} and therefore, \begin{align} \dot{P}(t) &= -U'(X(t)) + f'(X(t)) \underbrace{\int_{\RR^{+}} d\omega  c(\omega) (b^{\dagger}_t(\omega) + b_t(\omega) ) }_{\zeta(t)} \nonumber \\ 
&\ \ \ \ -  f'(X(t))  \int_{0}^{t} ds \underbrace{\int_{\RR^{+}} d\omega  2 r(\omega) \cos(\omega(t-s))}_{\kappa(t-s)}  \frac{d}{ds}\left(f(X(s)) \right) \nonumber \\ 
&\ \ \ \ -  f'(X(t)) f(X) \underbrace{\int_{\RR^{+}} d\omega  2r(\omega)  \cos(\omega t)}_{\kappa(t)},  \end{align} where \begin{equation}\frac{d}{ds}\left(f(X(s)) \right) = \frac{i}{\hbar}[H, f(X(s))] = \frac{\{ f'(X(s)), P(s)\}}{2 m},\end{equation} $b_t(\omega) = b(\omega)e^{-i\omega t}$, $b_t^{\dagger}(\omega) = b^{\dagger}(\omega) e^{i\omega t}$ and $\{\cdot, \cdot\}$ denotes anti-commutator.

Therefore, starting from the Heisenberg equations of motion and eliminating the bath variables, we obtain the following equations for the particle's observables: 
\begin{align}
\dot{X}(t) &= V(t), \label{qle1} \\ 
m  \dot{V}(t) &= -U'(X(t)) - f'(X(t)) \int_{0}^{t} \kappa(t-s) \frac{\{ f'(X(s)), V(s)\}}{2} ds \nonumber \\ 
&\ \ \ \ \ \ + f'(X(t)) \cdot (\zeta(t)-f(X)\kappa(t)), \label{qle2}
\end{align}
where 
\begin{align}
\kappa(t) &= \int_{\RR^{+}} d\omega \frac{2 |c(\omega)|^2}{\hbar \omega} \cos(\omega t)  = \int_{\RR^{+}} d\omega \frac{2 J(\omega)}{\omega} \cos(\omega t) 
\end{align}
is the {\it memory kernel},
\begin{align}
\zeta(t) &= \int_{\RR^{+}} d\omega c(\omega) (b^{\dagger}(\omega)e^{i\omega t}+b(\omega)e^{-i\omega t})   
\end{align} 
is a {\it stochastic force} whose correlation function depends on the coupling function, $c(\omega)$, and the distribution of the initial bath variables, $b(\omega)$ and $b^{\dagger}(\omega)$ -- let us remind that we initially consider a thermal Gibbs state. The term $f'(X(t))f(X) \kappa(t)$ is the initial slip term. The initial position and velocity are given by $X$ and $V$ respectively.

The above equations are exact, non-Markovian and operator-valued.  Note that in the damping term which is nonlocal in time, we have an anti-commutator, which does not appear in the corresponding classical equation or in the equation for the linear QBM model (where $f(X)=X$). The presence of the anti-commutator is thus a quantum feature of the inhomogeneous damping. The linear QBM model is exactly solvable. The properties of the solutions of the corresponding Heisenberg-Langevin equation\footnote{The rigorous study of the Heisenberg-Langevin equations, even in the case $f(X)=X$, is technically very difficult and there are only few works \citep{de1988quantum,maassen1984return} in the literature that treat them.}  have been studied in standard references on open quantum systems.

The initial preparation of the total system, which fixes the statistical properties of the bath operators and of the system's degrees of freeedom, turns the force $\zeta(t)$ into a stochastic one \citep{hughes2006dynamics}. We  specify a preparation procedure to fix the properties of the stochastic force. To this end, we absorb the initial slip term into the stochastic force, defining:
\begin{equation} \xi(t) := \zeta(t) - f(X) \kappa(t). \end{equation}
With this, in the nonlinear coupling case, the equation for the particle's velocity is driven by the multiplicative noise $f'(X(t)) \xi(t)$. From now on, we refer to $\xi(t)$ as the {\it quantum noise}. The statistics of $\xi(t)$ depends on the distributions of the initial bath variables $(b(\omega), b^{\dagger}(\omega))$ and the initial system variable $f(X)$. 


The main difference between classical and quantum systems lie in the statistical nature of the noise. Denoting by $E_{\beta}$ the expectation with respect to the thermal Gibbs state $\rho_{\beta}$ at the temperature $T$,  we have
\begin{align}
&E_\beta[(b^{\dagger}(\omega)e^{i\omega t}+b(\omega)e^{-i\omega t})  (b^{\dagger}(\omega')e^{i\omega' s}+b(\omega')e^{-i\omega' s})] \nonumber \\ &=\left[ (1+\nu_{\beta}(\omega)) e^{-i\omega(t-s)} + \nu_{\beta}(\omega) e^{i\omega(t-s)} \right] \delta(\omega-\omega'),
\end{align}
where $\nu_{\beta}(\omega)$ is given by the Planck's law \begin{equation} \nu_{\beta}(\omega) = \frac{1}{\exp{\left( \beta \hbar \omega \right)}-1}. \end{equation}

Since we absorbed the initial slip term into the stochastic force, $\xi(t)$ no longer has a stationary correlation when averaged with respect to $\rho_{\beta}$ \citep{hanggi1997generalized}. However, $\xi(t)$ is stationary and Gaussian when conditionally averaged with respect to $\rho'_{\beta} = e^{-\beta H_{B-I}}/Tr(H_{B-I}),$ where $H_{B-I}$  is the quadratic Hamiltonian defined in \eqref{trans_inv_ch6} and the average is conditioned on the initial position variable $X$.

The statistical properties of the quantum noise is fully specified by its two-time correlation function with respect to $\rho_\beta'$, given by:
\begin{align} E'_\beta[\xi(t) \xi(s)] &= \int_{\RR^{+}} d\omega \hbar J(\omega) \left(\coth\left(\frac{\hbar \omega}{2k_{B}T}\right)\cos(\omega(t-s)) - i\sin(\omega(t-s)) \right) \\ 
&=: D_{1}(t-s) - i D(t-s), \label{qfdt} \end{align}
where $D_{1}$ is the noise kernel given by
\begin{equation} D_{1}(t-s) := E'_\beta[\{\xi(t), \xi(s)\}/2],\end{equation} i.e. the symmetric correlation function of $\xi(t)$ with respect to $\rho_{\beta}'$, and $D$ is the dissipation kernel given by \begin{equation} D(t-s) := i E'_\beta[[\xi(t),\xi(s)]/2],\end{equation} which is related to linear susceptibility. Expanding, one gets for small $\hbar$ (or similarly, for large $T$), $E'_\beta[\xi(t) \xi(s)] =  k_{B}T \kappa(t) + O(\hbar),$ which is the classical Einstein's relation. Therefore, \eqref{qfdt} can be viewed as quantum analogue of the fluctuation-dissipation relation and a special case (and in time domain) of the quantum fluctuation-dissipation relation of Callen-Welton, since the noise kernel and dissipation kernel are related via:
\begin{equation}
\int_\RR dt \cos(\omega t) D_1(t) = \hbar \coth(\hbar \omega \beta/2) \int_\RR dt \sin(\omega t) D(t).
\end{equation}
Therefore, the multiscale structure of the quantum noises is far richer than that of classical noises even in the simplest (Ornstein-Uhlenbeck type) model.

\begin{rmk} {\bf On zero temperature systems.} For $T \to 0$ we have instead: 
\begin{equation} E'_{\beta}[\{\xi(t), \xi(s)\}/2] \to  \frac{-\hbar \Lambda^2}{2 \pi} (e^{-\Lambda (t-s)} \overline{Ei}(\Lambda (t-s)) + e^{\Lambda(t-s)} Ei(-\Lambda (t-s))), \end{equation}
where $Ei$ is the exponential integral function defined as follows: 
\begin{equation} -Ei(-x) = \hat{\gamma}(0,x) = \int_{x}^{\infty} e^{-t}/t dt. \end{equation} Here $\overline{Ei}(x) = \frac{1}{2}(Ei^{+}(x) + Ei^{-}(x))$, $Ei^{+}(x) = Ei(x+i0)$, $Ei^{-}(x) = Ei(x-i0)$. The symmetric correlation function obtained above can be interpreted as follows. As the temperature $T$ decreases, the Matsubara frequencies $\nu_{n}$ get closer to each other, so at zero temperature all of them contribute and the sum  may be replaced by an integral, which turns out to have expression in terms of the $Ei$ functions \citep{ingold}. In fact, in this case the symmetric correlation function decays polynomially for large times  \citep{jung1985long}. In other words, systems at zero temperature are strongly non-Markovian! 
\end{rmk}




\begin{rmk} {\bf On stochastic modeling of the quantum noise.}
A natural approach to study the quantum noise is to model it as a quantum stochastic process which admits a QSDE representation, along the line in the classical case. The existence of such processes that satisfy the KMS condition was studied rigorously in \cite{lewis1975existence}, after the notions of quantum stochastic process and stationarity were defined there. However, to our knowledge no rigorous studies on their QSDE representation have been performed. Of interest to us is to study QSDE representations for the class of ``quantum quasi-Markov'' stationary Gaussian process, by mimicking the theory for classical quasi-Markov processes. However, this extension of the studies to the quantum setting is not too straightforward and one would need to deal with some technicalities, for instance in the construction of appropriate representation Hilbert space \citep{araki1963representations,bratteli2012operator,bratteli2013operator} for the quantum process. 


\end{rmk}

\subsection{The QBM Model with an Ohmic Spectral Density} \label{model}
We consider the coupling function:
\begin{equation} \label{coupling}
c(\omega) = \sqrt{\frac{\hbar \omega }{\pi}\frac{\Lambda^2 }{\omega^2+\Lambda^2}},
\end{equation}
where  $\Lambda$ is a positive constant.  The bath spectral density is  given by:
\begin{equation}
J(\omega) := \frac{|c(\omega)|^2}{\hbar} = \frac{\omega}{\pi}\frac{\Lambda^2 }{\omega^2+\Lambda^2}, \end{equation}
which is the well-known Ohmic spectral density with a Lorentz-Drude cutoff \citep{BreuerBook}. 

Let us recall the Heisenberg-Langevin equation derived earlier. The
equations for the particle's observables read: 
\begin{align}
\dot{X}(t) &= V(t), \label{qle1_ch7} \\ 
m  \dot{V}(t) &= -U'(X(t)) - f'(X(t)) \int_{0}^{t} \kappa(t-s) \frac{\{ f'(X(s)), V(s)\}}{2} ds \nonumber \\ 
&\ \ \ \ \ \ + f'(X(t)) \cdot (\zeta(t)-f(X)\kappa(t)), \label{qle2_ch7}
\end{align}
where 
\begin{align}
\kappa(t) &= \int_{\RR^{+}} d\omega \frac{2 |c(\omega)|^2}{\hbar \omega} \cos(\omega t)  = \int_{\RR^{+}} d\omega \frac{2 J(\omega)}{\omega} \cos(\omega t) 
\end{align}
is the {\it memory kernel},
\begin{align}
\zeta(t) &= \int_{\RR^{+}} d\omega c(\omega) (b^{\dagger}(\omega)e^{i\omega t}+b(\omega)e^{-i\omega t})   
\end{align} 
is a {\it stochastic force}. 

We will work in a {\it Fock vacuum representation} using the H-P quantum stochastic calculus approach. In particular, our goal is to describe the quantum noise as a  quantum stochastic process satisfying certain QSDE such that the symmetric correlation function of the stochastic process with respect to the vacuum state on an enlarged Fock space coincides with that of  $\xi(t)$ with respect to $\rho_{\beta}'$. As a preparation to achieve this goal, we study $E_{\beta}'[\xi(t)\xi(s)]$ in the following. 
Recall that:
\begin{align} E'_\beta[\xi(t) \xi(s)] &= \int_{\RR^{+}} d\omega \hbar J(\omega) \left(\coth\left(\frac{\hbar \omega}{2k_{B}T}\right)\cos(\omega(t-s)) - i\sin(\omega(t-s)) \right) \\ 
&=: D_{1}(t-s) - i D(t-s), \end{align}
where $D_{1}$ is the noise kernel given by
\begin{equation} D_{1}(t-s) := E'_\beta[\{\xi(t), \xi(s)\}/2],\end{equation} i.e. the symmetric correlation function of $\xi(t)$ with respect to $\rho_{\beta}'$, and $D$ is the dissipation kernel given by \begin{equation} D(t-s) := i E'_\beta[[\xi(t),\xi(s)]/2],\end{equation} which is related to linear susceptibility. 

For our choice of $c(\omega)$ (see \eqref{coupling}), the memory kernel, $\kappa(t)$, is exponentially decaying with decay rate $\Lambda$, i.e. $\kappa(t) = \Lambda e^{-\Lambda t}.$ Moreover, one can compute, for $t>0$: 
\begin{align} \label{D1}
D_{1}(t) 
&= \frac{\hbar \Lambda}{2} \cot \left(\frac{\hbar \Lambda}{2k_{B}T} \right) \kappa(t) +  \sum_{n=1}^{\infty} \frac{2k_{ B}T \Lambda^2 \nu_{n}}{ \nu_{n}^2 - \Lambda^2} e^{-\nu_{n}t}, 
\end{align}
where $\nu_{n} = \frac{2 \pi n k_{B}T}{\hbar}$ are the bath Matsubara frequencies \citep{carlesso2016quantum}. Also, the dissipation kernel is 
\begin{equation} \label{D}
D(t) = \frac{\hbar \Lambda^3}{2} e^{-\Lambda t}. 
\end{equation}
In this paper, we consider the case $k_B T > \hbar \Lambda/\pi$, so that $\cot(\hbar \Lambda/2k_BT)$ and all the terms in the series in \eqref{D1} are positive. 



\subsection{QSDE's for Quantum Noise} \label{qnoise}

Guided by formula of the symmetric correlation function in \eqref{D1} and the plan outlined in Section \ref{approach}, we model the quantum noise by: 
\begin{equation}
\label{eq:noise}
 \sum_{k=0}^{\infty} \eta_{k}(t),
\end{equation}
  where the $\eta_{k}(t)$ are independent {\it quantum Ornstein-Uhlenbeck processes} (quantum analogue of the classical ones \citep{csaki1991infinite}), satisfying the SDEs: 
\begin{equation} \label{qou_ch7} d\eta_k(t) = -\alpha_{k} \eta_{k}(t) dt + \sqrt{\lambda_{k}} dW^\theta_{k}(t), \ \ \eta_k(0) = \eta_k.
\end{equation}
Here the $W_k^{\theta}$ are independent quantum Wiener processes defined earlier and the $\eta_{k}$ are  initial variables on a copy of Fock space. For a fixed $\theta$, independence and commutation for these processes can be achieved by realizing  the $\eta_{k}(t)$ on distinct copies of Fock space, i.e. \begin{equation} \sum_{k=0}^{\infty} \eta_k(t) = \eta_{0}(t) \otimes I \otimes I \otimes \dots + I \otimes \eta_1(t) \otimes I \otimes \dots + \dots  \end{equation} on $\bigotimes_{k=0}^{\infty} \Gamma(L^2(\RR^{+})) = \Gamma(L^2(\RR^+)\otimes \mathcal{K})$ where the multiplicity space $\mathcal{K}$ is a sequence space whose elements are of the form $(x_0, x_1, x_2, \dots)$, with each $x_i \in \CC$. From now on, each $\eta_k$ is understood to be \begin{equation} \underbrace{I \otimes \cdots \otimes I}_{k \text{  copies}} \otimes \eta_k \otimes I \otimes \cdots \end{equation} and similarly for each $W_{k}^{\theta}$. 

The formal solution to the SDE \eqref{qou_ch7} is given by:
\begin{equation} \eta_{k}(t) = \eta_k e^{-\alpha_{k} t} + \sqrt{\lambda_{k}} \int_{0}^{t} e^{-\alpha_{k}(t-s)} dW^{\theta}_{k}(s). \end{equation} 

Since there is a unique stationary solution of the SDEs \eqref{qou_ch7}, for all $k$ and $s \in [0,t]$:
\begin{equation} E''_{\infty}[\eta_{k}^2] = \frac{\lambda_k}{2 \alpha_k}, \ \ \ \ E''_\infty[\eta_k W^\theta_{k}(s)] = E''_{\infty}[W^\theta_k(s) \eta_k] =  0,\end{equation} where  $E''_{\infty}$ denotes expectation with respect to the vacuum state associated with $\Omega \otimes \Omega \otimes \cdots$ on the enlarged Fock space $\Gamma(L^2(\RR^{+})\otimes \mathcal{K})$. Then, with the parameters $\alpha_n$  and $\lambda_{n}$ defined by
\begin{equation}
\label{eq:alphan}
\alpha_{n} = \nu_n 1_{\{n \geq   1 \}} + \Lambda 1_{\{n=0\}} > 0 
\end{equation} and
\begin{equation}
\label{eq:lambdan}
\lambda_{n} = \frac{4 \nu_n^2 \Lambda^2 k_B T}{\nu_n^2-\Lambda^2}  1_{\{n \geq 1 \}} + \hbar \Lambda^3 \cot\left( \frac{\hbar \Lambda}{2k_{B}T} \right)  1_{\{n=0\}} > 0,
\end{equation} 
it can be verified that 
\begin{equation}
E''_{\infty}\left[\frac{\left\{\sum_k \eta_k(t), \sum_l \eta_{l}(s) \right\}}{2} \right] = D_{1}(t-s),
\end{equation} 
where $D_{1}$ is given in \eqref{D1}. 

Equations~\eqref{eq:alphan} and~\eqref{eq:lambdan} establish a link between the quantum noise as introduced in eqn.~\eqref{eq:noise} and the physical model of Section \ref{model}. We remark that there is freedom in the above construction of quantum noise, as the driving noise process, $(W_k^{\theta})$, is a family of quantum Wiener processes parametrized by $\theta$. On the one hand, the choice of the parameter should be fixed by physical considerations, i.e. by the nature of the field that the system couples to in the microscopic model. On the other hand, one would like to show that the quantum noises describe a Markovian system, so one should write the SDEs $\eqref{qou_ch7}$ in a H-P QSDE form.

To this end, let $\xi_{k}(t)$ and $\eta_{k}(t)$ be canonical conjugate bath observables that obey the commutation relation $[\xi_{j}(t), \eta_{k}(t) ] = i \hbar \delta_{jk} I$ for all $t \geq 0$. Suppose that the evolution of each pair $(\xi_k(t), \eta_{k}(t))$ is Markovian and can be described by the H-P QSDEs associated with $(H_k, L_k)$, where 
\begin{equation} H_{k} = \frac{\eta_k^2}{2} + \frac{\alpha_{k}}{4} \{ \xi_k, \eta_k \}, \ \  \ \ L_{k} = \frac{\sqrt{\lambda_k}}{\hbar} \xi_k + i \frac{\alpha_{k}}{2\sqrt{\lambda_{k}}} \eta_{k},\end{equation} where $\alpha_{k}$ and $\lambda_{k}$ are given as before. Therefore, they solve the  H-P QSDEs:
\begin{align}
d\xi_k(t) &= \eta_k(t) dt + \frac{\hbar \alpha_k}{2 \sqrt{\lambda_k}} dW^\pi_{k}(t), \label{qn1} \\
d \eta_{k}(t) &= -\alpha_{k} \eta_{k}(t) dt + \sqrt{\lambda_{k}} dW^{-\pi/2}_{k}(t), \label{qn2}
\end{align}
where \begin{equation} W^\pi_{k}(t) = -(A_k(t) + A^\dagger_k(t)), \ \ \ W^{-\pi/2}_{k}(t) = i(A_k(t) - A^\dagger_k(t))\end{equation} are noncommuting, conjugate quantum Wiener processes satisfying $[W_k^{\pi}(t), W_{k}^{-\pi/2}(s)] = 2i \delta(t-s) I$.  Modulo the negative factor, one can view the formal time derivatives of the $W^\pi_{k}(t)$ and $W^{-\pi/2}_{k}(t)$ as the noises arising from the position and momentum field observables respectively. We fix the freedom in our construction by taking the  Markovian system \eqref{qn1}-\eqref{qn2} as the model for noise. Therefore, we take $\sum_k \eta_k(t)$ to be the quantum colored noise that models the action of the heat bath on the evolution of the system's observables. 

Physically, one can think of our quantum noise model as equivalent to a model of infinitely many non-interacting ancillas that convert the white noise to colored noise through a channel at each Matsubara frequency \citep{xue2015quantum}. That one needs  infinitely many ancillas is due to the fact that there are infinitely many transition (Bohr) energies, each of which equals  the energy of a boson with a particular Matsubara frequency in the bath. According to our noise model, when a boson with the Matsubara frequency $\nu_{k}$ is created or annihilated, the energy transition does not occur instantaneously but happens on the time scale of $1/\alpha_k$ via a channel associated with $\nu_k$.



\section{Conclusion}
We have provided a basic introduction to quantum noise and open quantum systems, motivated by starting with classical noise and open classical systems. Our formulation is based on the elegant approach of \cite{parthasarathy2012introduction}, and we hope that such formulation would be more widely adopted by  physicists working in open quantum systems. The formulation could also be employed and/or extended to study data-driven approaches for open systems \citep{reddy2024data}, quantum stochastic thermodynamics \citep{strasberg2022quantum, lim2021anomalous}, quantum machine learning \citep{schuld2015introduction, olivera2023benefits} and the related problems \citep{dawid2022modern}.

\acks{The author would like to acknowledge the WINQ Fellowship and the Swedish Research Council (VR/2021-03648). He is also grateful to Jan Wehr and Maciej Lewenstein for stimulating discussions.}

\bibliographystyle{siam}

\begin{thebibliography}{166}
\providecommand{\natexlab}[1]{#1}
\providecommand{\url}[1]{\texttt{#1}}
\expandafter\ifx\csname urlstyle\endcsname\relax
  \providecommand{\doi}[1]{doi: #1}\else
  \providecommand{\doi}{doi: \begingroup \urlstyle{rm}\Url}\fi

\bibitem[Accardi et~al.(2002)Accardi, Lu, and Volovich]{accardi2002quantum}
L.~Accardi, Y.G. Lu, and I.~Volovich.
\newblock \emph{Quantum Theory and Its Stochastic Limit}.
\newblock Physics and astronomy online library. Springer Berlin Heidelberg,
  2002.
\newblock ISBN 9783540419280.

\bibitem[Accardi(2006)]{accardi2006could}
Luigi Accardi.
\newblock Could we now convince {E}instein?
\newblock In \emph{AIP Conference Proceedings}, volume 810, pages 3--18. AIP,
  2006.

\bibitem[Accardi et~al.(1982)Accardi, Frigerio, and Lewis]{accardi1982quantum}
Luigi Accardi, Alberto Frigerio, and John~T Lewis.
\newblock Quantum stochastic processes.
\newblock \emph{Publications of the Research Institute for Mathematical
  Sciences}, 18\penalty0 (1):\penalty0 97--133, 1982.

\bibitem[Adelman and Doll(1976)]{adelman1976generalized}
SA~Adelman and JD~Doll.
\newblock Generalized {L}angevin equation approach for atom/solid-surface
  scattering: General formulation for classical scattering off harmonic solids.
\newblock \emph{The Journal of chemical physics}, 64\penalty0 (6):\penalty0
  2375--2388, 1976.

\bibitem[Alicki(1987)]{alicki1987general}
Robert Alicki.
\newblock General theory and applications to unstable particles.
\newblock In \emph{Quantum Dynamical Semigroups and Applications}, pages 1--94.
  Springer, 1987.

\bibitem[Applebaum(2010)]{applebaum2010robin}
David Applebaum.
\newblock Robin {H}udson's pathless path to quantum stochastic calculus.
\newblock \emph{Communications on Stochastic Analysis}, 4\penalty0
  (4):\penalty0 2, 2010.

\bibitem[Araki and Woods(1963)]{araki1963representations}
Huzihiro Araki and EJ~Woods.
\newblock Representations of the canonical commutation relations describing a
  nonrelativistic infinite free {B}ose gas.
\newblock \emph{Journal of Mathematical Physics}, 4\penalty0 (5):\penalty0
  637--662, 1963.

\bibitem[Ariel and Vanden-Eijnden(2008)]{ariel2008strong}
Gil Ariel and Eric Vanden-Eijnden.
\newblock A strong limit theorem in the {K}ac--{Z}wanzig model.
\newblock \emph{Nonlinearity}, 22\penalty0 (1):\penalty0 145, 2008.

\bibitem[Arnold()]{arnold2013stochastic}
L.~Arnold.
\newblock \emph{Stochastic Differential Equations: Theory and Applications}.
\newblock Dover Books on Mathematics. Dover Publications.
\newblock ISBN 9780486482361.

\bibitem[Attal et~al.(2006)Attal, Joye, and Pillet]{attal2006open}
S.~Attal, A.~Joye, and C.A. Pillet.
\newblock \emph{Open Quantum Systems II: The Markovian Approach}.
\newblock Lecture Notes in Mathematics. Springer, 2006.
\newblock ISBN 9783540309925.

\bibitem[Attal and Pautrat(2006)]{attal2006repeated}
St{\'e}phane Attal and Yan Pautrat.
\newblock From repeated to continuous quantum interactions.
\newblock In \emph{Annales Henri Poincar{\'e}}, volume~7, pages 59--104.
  Springer, 2006.

\bibitem[Baez and Biamonte(2012)]{baez2012quantum}
John~C Baez and Jacob Biamonte.
\newblock Quantum techniques for stochastic mechanics.
\newblock \emph{arXiv preprint arXiv:1209.3632}, 2012.

\bibitem[Barchielli and Gregoratti(2009)]{barchielli2009quantum}
Alberto Barchielli and Matteo Gregoratti.
\newblock \emph{Quantum trajectories and measurements in continuous time: the
  diffusive case}, volume 782.
\newblock Springer, 2009.

\bibitem[Barchielli and Vacchini(2015)]{barchielli2015quantum}
Alberto Barchielli and Bassano Vacchini.
\newblock Quantum {L}angevin equations for optomechanical systems.
\newblock \emph{New Journal of Physics}, 17\penalty0 (8):\penalty0 083004,
  2015.

\bibitem[Belavkin et~al.(1995)Belavkin, Hirota, and Hudson]{belavkin1995world}
VP~Belavkin, O~Hirota, and RL~Hudson.
\newblock The world of quantum noise and the fundamental output process.
\newblock In \emph{Quantum Communications and Measurement}, pages 3--19.
  Springer, 1995.

\bibitem[Bellet(2006)]{bellet2006ergodic}
Luc~Rey Bellet.
\newblock Ergodic properties of {M}arkov processes.
\newblock In \emph{Open Quantum Systems II}, pages 1--39. Springer, 2006.

\bibitem[Bellman(1997)]{bellman1997introduction}
Richard Bellman.
\newblock \emph{Introduction to Matrix Analysis}, volume~19.
\newblock Siam, 1997.

\bibitem[Biane(2010)]{biane2010ito}
Philippe Biane.
\newblock It{\^o}’s stochastic calculus and {H}eisenberg commutation
  relations.
\newblock \emph{Stochastic Processes and their Applications}, 120\penalty0
  (5):\penalty0 698--720, 2010.

\bibitem[Bogachev(1998)]{bogachev1998gaussian}
Vladimir~Igorevich Bogachev.
\newblock \emph{Gaussian Measures}.
\newblock Number~62. American Mathematical Soc., 1998.

\bibitem[Bouten et~al.(2007)Bouten, Van~Handel, and
  James]{bouten2007introduction}
Luc Bouten, Ramon Van~Handel, and Matthew~R James.
\newblock An introduction to quantum filtering.
\newblock \emph{SIAM Journal on Control and Optimization}, 46\penalty0
  (6):\penalty0 2199--2241, 2007.

\bibitem[Bouten et~al.(2015)Bouten, Gohm, Gough, and Nurdin]{bouten2015trotter}
Luc Bouten, Rolf Gohm, John Gough, and Hendra Nurdin.
\newblock A {T}rotter-{K}ato theorem for quantum {M}arkov limits.
\newblock \emph{EPJ Quantum Technology}, 2\penalty0 (1):\penalty0 1, 2015.

\bibitem[Bouten(2008)]{Bouten2008}
Lue Bouten.
\newblock \emph{Applications of Quantum Stochastic Processes in Quantum
  Optics}, pages 277--307.
\newblock Springer Berlin Heidelberg, Berlin, Heidelberg, 2008.
\newblock ISBN 978-3-540-69365-9.
\newblock \doi{10.1007/978-3-540-69365-9_6}.

\bibitem[Bratteli and Robinson(2013)]{bratteli2013operator}
O.~Bratteli and D.W. Robinson.
\newblock \emph{Operator Algebras and Quantum Statistical Mechanics II:
  Equilibrium States Models in Quantum Statistical Mechanics}.
\newblock Theoretical and Mathematical Physics. Springer Berlin Heidelberg,
  2013.
\newblock ISBN 9783662090893.

\bibitem[Bratteli and Robinson(2012)]{bratteli2012operator}
Ola Bratteli and Derek~William Robinson.
\newblock \emph{Operator Algebras and Quantum Statistical Mechanics: Volume 1:
  C*-and W*-Algebras. Symmetry Groups. Decomposition of States}.
\newblock Springer Science \& Business Media, 2012.

\bibitem[Breuer and Petruccione(2007)]{BreuerBook}
H.P. Breuer and F.~Petruccione.
\newblock \emph{The Theory of Open Quantum Systems}.
\newblock OUP, Oxford, 2007.
\newblock ISBN 9780199213900.

\bibitem[Bylicka et~al.(2013)Bylicka, Chru{\'s}ci{\'n}ski, and
  Maniscalco]{bylicka2013non}
Bogna Bylicka, Dariusz Chru{\'s}ci{\'n}ski, and Sabrina Maniscalco.
\newblock Non-{M}arkovianity as a resource for quantum technologies.
\newblock \emph{arXiv preprint arXiv:1301.2585}, 2013.

\bibitem[Caldeira and Leggett(1983)]{Caldeira1983a}
A.O. Caldeira and A.J. Leggett.
\newblock Path integral approach to quantum {B}rownian motion.
\newblock \emph{Physica A: Statistical Mechanics and its Applications},
  121\penalty0 (3):\penalty0 587 -- 616, 1983.
\newblock ISSN 0378-4371.

\bibitem[Callen and Welton(1951)]{callen1951irreversibility}
Herbert~B Callen and Theodore~A Welton.
\newblock Irreversibility and generalized noise.
\newblock \emph{Physical Review}, 83\penalty0 (1):\penalty0 34, 1951.

\bibitem[Carlesso and Bassi(2017)]{carlesso2016quantum}
Matteo Carlesso and Angelo Bassi.
\newblock Adjoint master equation for quantum {B}rownian motion.
\newblock \emph{Physical Review A}, 95\penalty0 (5):\penalty0 052119, 2017.

\bibitem[Chorin and Hald(2009)]{chorin2009stochastic}
Alexandre~Joel Chorin and Ole~H Hald.
\newblock \emph{Stochastic Tools in Mathematics and Science}, volume~3.
\newblock Springer, 2009.

\bibitem[Combes et~al.(2017)Combes, Kerckhoff, and Sarovar]{combes2017slh}
Joshua Combes, Joseph Kerckhoff, and Mohan Sarovar.
\newblock The slh framework for modeling quantum input-output networks.
\newblock \emph{Advances in Physics: X}, 2\penalty0 (3):\penalty0 784--888,
  2017.

\bibitem[C{\'o}rdoba et~al.(2012)C{\'o}rdoba, Indei, and
  Schieber]{cordoba2012elimination}
Andr{\'e}s C{\'o}rdoba, Tsutomu Indei, and Jay~D Schieber.
\newblock Elimination of inertia from a generalized {L}angevin equation:
  applications to microbead rheology modeling and data analysis.
\newblock \emph{Journal of Rheology}, 56\penalty0 (1):\penalty0 185--212, 2012.

\bibitem[Cram{\'e}r and Leadbetter(2013)]{cramer2013stationary}
Harald Cram{\'e}r and M~Ross Leadbetter.
\newblock \emph{Stationary and related stochastic processes: Sample function
  properties and their applications}.
\newblock Courier Corporation, 2013.

\bibitem[Cs{\'a}ki et~al.(1991)Cs{\'a}ki, Cs{\"o}rg{\H{o}}, Lin, and
  R{\'e}v{\'e}sz]{csaki1991infinite}
E~Cs{\'a}ki, M~Cs{\"o}rg{\H{o}}, ZY~Lin, and P~R{\'e}v{\'e}sz.
\newblock On infinite series of independent {O}rnstein-{U}hlenbeck processes.
\newblock \emph{Stochastic Processes and their Applications}, 39\penalty0
  (1):\penalty0 25--44, 1991.

\bibitem[Dawid et~al.(2022)Dawid, Arnold, Requena, Gresch, P{\l}odzie{\'n},
  Donatella, Nicoli, Stornati, Koch, B{\"u}ttner, et~al.]{dawid2022modern}
Anna Dawid, Julian Arnold, Borja Requena, Alexander Gresch, Marcin
  P{\l}odzie{\'n}, Kaelan Donatella, Kim~A Nicoli, Paolo Stornati, Rouven Koch,
  Miriam B{\"u}ttner, et~al.
\newblock Modern applications of machine learning in quantum sciences.
\newblock \emph{arXiv preprint arXiv:2204.04198}, 2022.

\bibitem[De~Smedt et~al.(1988)De~Smedt, D{\"u}rr, Lebowitz, and
  Liverani]{de1988quantum}
P~De~Smedt, D~D{\"u}rr, JL~Lebowitz, and C~Liverani.
\newblock Quantum system in contact with a thermal environment: Rigorous
  treatment of a simple model.
\newblock \emph{Communications in Mathematical Physics}, 120\penalty0
  (2):\penalty0 195--231, 1988.

\bibitem[de~Vega and Alonso(2017)]{de2017dynamics}
In{\'e}s de~Vega and Daniel Alonso.
\newblock Dynamics of non-{M}arkovian open quantum systems.
\newblock \emph{Reviews of Modern Physics}, 89\penalty0 (1):\penalty0 015001,
  2017.

\bibitem[{Derezinski} and {De Roeck}(2007)]{Der_notes}
J.~{Derezinski} and W.~{De Roeck}.
\newblock {Reduced and Extended Weak Coupling Limit}.
\newblock \emph{ArXiv e-prints}, April 2007.

\bibitem[Derezi{\'n}ski and G{\'e}rard(1999)]{derezinski1999asymptotic}
J~Derezi{\'n}ski and Christian G{\'e}rard.
\newblock Asymptotic completeness in quantum field theory. {M}assive
  {P}auli-{F}ierz {H}amiltonians.
\newblock \emph{Reviews in Mathematical Physics}, 11\penalty0 (04):\penalty0
  383--450, 1999.

\bibitem[Derezi{\'n}ski and Jak{\v{s}}i{\'c}(2001)]{derezinski2001spectral}
Jan Derezi{\'n}ski and Vojkan Jak{\v{s}}i{\'c}.
\newblock Spectral theory of {P}auli--{F}ierz operators.
\newblock \emph{Journal of Functional Analysis}, 180\penalty0 (2):\penalty0
  243--327, 2001.

\bibitem[Doob and Doob(1953)]{doob1953stochastic}
Joseph~L Doob and Joseph~L Doob.
\newblock \emph{Stochastic Processes}, volume~7.
\newblock Wiley New York, 1953.

\bibitem[Dym and McKean(2008)]{dym2008gaussian}
Harry Dym and Henry~P McKean.
\newblock \emph{Gaussian Processes, Function Theory, and the Inverse Spectral
  Problem}.
\newblock Courier Corporation, 2008.

\bibitem[Emzir et~al.(2016)Emzir, Woolley, and Petersen]{emzir2016physical}
Muhammad~F Emzir, Matthew~J Woolley, and Ian~R Petersen.
\newblock On physical realizability of nonlinear quantum stochastic
  differential equations.
\newblock \emph{arXiv preprint arXiv:1612.07877}, 2016.

\bibitem[Evans(2012)]{evans2012introduction}
Lawrence~C Evans.
\newblock \emph{An Introduction to Stochastic Differential Equations},
  volume~82.
\newblock American Mathematical Soc., 2012.

\bibitem[Fagnola(1999)]{fagnola1999quantum}
Franco Fagnola.
\newblock Quantum {M}arkov semigroups and quantum flows.
\newblock \emph{Proyecciones}, 18\penalty0 (3):\penalty0 1--144, 1999.

\bibitem[Fagnola(2006)]{fagnola2006quantum}
Franco Fagnola.
\newblock Quantum stochastic differential equations and dilation of completely
  positive semigroups.
\newblock In \emph{Open Quantum Systems II}, pages 183--220. Springer, 2006.

\bibitem[Fagnola and Rebolledo(1999)]{fagnola1999view}
Franco Fagnola and Rolando Rebolledo.
\newblock A view on stochastic differential equations derived from quantum
  optics.
\newblock \emph{Stochastic models (Guanajuato, 1998)}, pages 193--214, 1999.

\bibitem[Filip et~al.(2019)Filip, Javeed, and Trefethen]{filipsmooth}
Silvio Filip, Aurya Javeed, and Lloydomgg~N Trefethen.
\newblock Smooth random functions, random {O}{D}{E}s, and {G}aussian processes.
\newblock \emph{SIAM Review, Vol 61}, 2019.

\bibitem[Ford et~al.(1965)Ford, Kac, and Mazur]{ford1965statistical}
GW~Ford, M~Kac, and P~Mazur.
\newblock Statistical mechanics of assemblies of coupled oscillators.
\newblock \emph{Journal of Mathematical Physics}, 6\penalty0 (4):\penalty0
  504--515, 1965.

\bibitem[Gardiner and Zoller(2004)]{GardinerBook}
C.~Gardiner and P.~Zoller.
\newblock \emph{Quantum Noise: A Handbook of Markovian and Non-Markovian
  Quantum Stochastic Methods with Applications to Quantum Optics}.
\newblock Springer Series in Synergetics. Springer, Berlin, 2004.
\newblock ISBN 9783540223016.

\bibitem[Gardiner(2009)]{gardiner2009stochastic}
Crispin Gardiner.
\newblock \emph{Stochastic Methods}, volume~4.
\newblock Springer Berlin, 2009.

\bibitem[Gel'fand and Vilenkin(2014)]{gel2014generalized}
I.M. Gel'fand and N.Y. Vilenkin.
\newblock \emph{Generalized Functions: Applications of Harmonic Analysis}.
\newblock Generalized Functions. Elsevier Science, 2014.
\newblock ISBN 9781483262246.

\bibitem[Gikhman and Skorokhod(1996)]{gikhman1996introduction}
I.I. Gikhman and A.V. Skorokhod.
\newblock \emph{Introduction to the Theory of Random Processes}.
\newblock Dover books on mathematics. Dover Publications, 1996.
\newblock ISBN 9780486693873.

\bibitem[Gittes and Schmidt(1997)]{gittes1997signals}
Frederick Gittes and Christoph~F Schmidt.
\newblock \emph{Signals and noise in micromechanical measurements}, volume~55.
\newblock Elsevier, 1997.

\bibitem[Glatt-Holtz et~al.(2018)Glatt-Holtz, Herzog, McKinley, and
  Nguyen]{GlattHoltz2018}
Nathan Glatt-Holtz, David Herzog, Scott McKinley, and Hung Nguyen.
\newblock The generalized {L}angevin equation with a power-law memory in a
  nonlinear potential well.
\newblock \emph{arXiv preprint arXiv:1804.00202}, 2018.

\bibitem[Glauber(1963)]{glauber1963coherent}
Roy~J Glauber.
\newblock Coherent and incoherent states of the radiation field.
\newblock \emph{Physical Review}, 131\penalty0 (6):\penalty0 2766, 1963.

\bibitem[Gorini et~al.(1978)Gorini, Frigerio, Verri, Kossakowski, and
  Sudarshan]{gorini1978properties}
Vittorio Gorini, Alberto Frigerio, Maurizio Verri, Andrzej Kossakowski, and ECG
  Sudarshan.
\newblock Properties of quantum {M}arkovian master equations.
\newblock \emph{Reports on Mathematical Physics}, 13\penalty0 (2):\penalty0
  149--173, 1978.

\bibitem[Gough(2006)]{gough2006quantum}
John Gough.
\newblock Quantum {S}tratonovich calculus and the quantum {W}ong-{Z}akai
  theorem.
\newblock \emph{Journal of Mathematical Physics}, 47\penalty0 (11):\penalty0
  113509, 2006.

\bibitem[Gough and James(2009)]{gough2009quantum}
John Gough and MR~James.
\newblock Quantum feedback networks: Hamiltonian formulation.
\newblock \emph{Communications in Mathematical Physics}, 287\penalty0
  (3):\penalty0 1109--1132, 2009.

\bibitem[Goychuk(2012)]{goychuk2012viscoelastic}
Igor Goychuk.
\newblock Viscoelastic subdiffusion: Generalized {L}angevin equation approach.
\newblock \emph{Advances in Chemical Physics}, 150:\penalty0 187, 2012.

\bibitem[Goychuk(2016)]{goychuk2016molecular}
Igor Goychuk.
\newblock Molecular machines operating on the nanoscale: From classical to
  quantum.
\newblock \emph{Beilstein Journal of Nanotechnology}, 7:\penalty0 328, 2016.

\bibitem[Gregoratti(2001)]{gregoratti2001hamiltonian}
M~Gregoratti.
\newblock The {H}amiltonian operator associated with some quantum stochastic
  evolutions.
\newblock \emph{Communications in Mathematical Physics}, 222\penalty0
  (1):\penalty0 181--200, 2001.

\bibitem[Gustafson and Sigal(2011)]{gustafson2011mathematical}
Stephen~J Gustafson and Israel~Michael Sigal.
\newblock \emph{Mathematical Concepts of Quantum Mechanics}.
\newblock Springer Science \& Business Media, 2011.

\bibitem[Hald and Kupferman(2002)]{hald2002asymptotic}
Ole~H Hald and Raz Kupferman.
\newblock Asymptotic and numerical analyses for mechanical models of heat
  baths.
\newblock \emph{Journal of Statistical Physics}, 106\penalty0 (5-6):\penalty0
  1121--1184, 2002.

\bibitem[H{\"a}nggi(1997)]{hanggi1997generalized}
Peter H{\"a}nggi.
\newblock Generalized {L}angevin equations: A useful tool for the perplexed
  modeller of nonequilibrium fluctuations?
\newblock In \emph{Stochastic Dynamics}, pages 15--22. Springer, 1997.

\bibitem[Haroche and Raimond(2006)]{haroche2006exploring}
S.~Haroche and J.M. Raimond.
\newblock \emph{Exploring the Quantum: Atoms, Cavities, and Photons}.
\newblock Oxford Graduate Texts. OUP Oxford, 2006.
\newblock ISBN 9780198509141.

\bibitem[Haroche(2013)]{haroche2013nobel}
Serge Haroche.
\newblock Nobel lecture: Controlling photons in a box and exploring the quantum
  to classical boundary.
\newblock \emph{Reviews of Modern Physics}, 85\penalty0 (3):\penalty0 1083,
  2013.

\bibitem[Hartmann(2011)]{hartmann2011balanced}
Carsten Hartmann.
\newblock Balanced model reduction of partially observed {L}angevin equations:
  an averaging principle.
\newblock \emph{Mathematical and Computer Modelling of Dynamical Systems},
  17\penalty0 (5):\penalty0 463--490, 2011.

\bibitem[Hida and Si(2008)]{hida2008lectures}
Takeyuki Hida and Si~Si.
\newblock \emph{Lectures on White Noise Functionals}.
\newblock World scientific, 2008.

\bibitem[Hsu(2002)]{hsu2002stochastic}
E.P. Hsu.
\newblock \emph{Stochastic Analysis on Manifolds}.
\newblock Contemporary Mathematics. American Mathematical Society, 2002.
\newblock ISBN 9780821808023.

\bibitem[Hudson(1985)]{Hudson1985}
Lindsay~Martin Hudson, Robin.
\newblock The classical limit of reduced quantum stochastic evolutions.
\newblock \emph{Annales de l'I.H.P. Physique théorique}, 43\penalty0
  (2):\penalty0 133--145, 1985.

\bibitem[Hudson and Parthasarathy(1986)]{hudson1986unification}
RL~Hudson and KR~Parthasarathy.
\newblock Unification of fermion and boson stochastic calculus.
\newblock \emph{Communications in Mathematical Physics}, 104\penalty0
  (3):\penalty0 457--470, 1986.

\bibitem[Hudson and Parthasarathy(1984)]{hudson1984quantum}
Robin~L Hudson and Kalyanapuram~R Parthasarathy.
\newblock Quantum {I}to's formula and stochastic evolutions.
\newblock \emph{Communications in Mathematical Physics}, 93\penalty0
  (3):\penalty0 301--323, 1984.

\bibitem[Hughes(2006)]{hughes2006dynamics}
Keith~H Hughes.
\newblock \emph{Dynamics of Open Quantum Systems}.
\newblock Collaborative Computational Project on Molecular Quantum Dynamics
  (CCP6), 2006.

\bibitem[Huisinga et~al.(2003)Huisinga, Sch{\"u}tte, and
  Stuart]{huisinga2003extracting}
Wilhelm Huisinga, Christof Sch{\"u}tte, and Andrew~M Stuart.
\newblock Extracting macroscopic stochastic dynamics: Model problems.
\newblock \emph{Communications on Pure and Applied Mathematics: A Journal
  Issued by the Courant Institute of Mathematical Sciences}, 56\penalty0
  (2):\penalty0 234--269, 2003.

\bibitem[Ikeda and Watanabe(2014)]{ikeda2014stochastic}
Nobuyuki Ikeda and Shinzo Watanabe.
\newblock \emph{Stochastic Differential Equations and Diffusion Processes},
  volume~24.
\newblock Elsevier, 2014.

\bibitem[Ingold()]{ingold}
Ingold.
\newblock \emph{Chapter 4: Dissipative Quantum Systems}.

\bibitem[It{\^o} et~al.(1954)]{ito1954stationary}
Kiyosi It{\^o} et~al.
\newblock Stationary random distributions.
\newblock \emph{Memoirs of the College of Science, University of Kyoto. Series
  A: Mathematics}, 28\penalty0 (3):\penalty0 209--223, 1954.

\bibitem[Janson(1997)]{janson1997gaussian}
Svante Janson.
\newblock \emph{Gaussian {H}ilbert spaces}, volume 129.
\newblock Cambridge University Press, 1997.

\bibitem[Jung et~al.(1985)Jung, Ingold, and Grabert]{jung1985long}
Roland Jung, Gert-Ludwig Ingold, and Hermann Grabert.
\newblock Long-time tails in quantum {B}rownian motion.
\newblock \emph{Physical Review A}, 32\penalty0 (4):\penalty0 2510, 1985.

\bibitem[Kahane(1997)]{kahane1997century}
Jean-Pierre Kahane.
\newblock A century of interplay between {T}aylor series, {F}ourier series and
  {B}rownian motion.
\newblock \emph{Bulletin of the London Mathematical Society}, 29\penalty0
  (3):\penalty0 257--279, 1997.

\bibitem[Kahane(1968)]{kahane1968some}
J.P. Kahane.
\newblock \emph{Some Random Series of Functions}.
\newblock Heath Mathematical Monographs. D. C. Heath, 1968.

\bibitem[Karatzas and Shreve(2012)]{karatzas2012brownian}
Ioannis Karatzas and Steven Shreve.
\newblock \emph{Brownian Motion and Stochastic Calculus}, volume 113.
\newblock Springer Science \& Business Media, 2012.

\bibitem[Khasminskii(2011)]{khasminskii2011stochastic}
Rafail Khasminskii.
\newblock \emph{Stochastic Stability of Differential Equations}, volume~66.
\newblock Springer Science \& Business Media, 2011.

\bibitem[Kowalski(1994)]{kowalski1994methods}
Krzysztof Kowalski.
\newblock \emph{Methods of Hilbert spaces in the theory of nonlinear dynamical
  systems}.
\newblock World Scientific, 1994.

\bibitem[Kowalski and Steeb(1991)]{kowalski1991nonlinear}
Krzysztof Kowalski and W-H Steeb.
\newblock \emph{Nonlinear dynamical systems and Carleman linearization}.
\newblock World Scientific, 1991.

\bibitem[Kraus et~al.(2008)Kraus, B\"uchler, Diehl, Kantian, Micheli, and
  Zoller]{Kraus2008}
B.~Kraus, H.~P. B\"uchler, S.~Diehl, A.~Kantian, A.~Micheli, and P.~Zoller.
\newblock Preparation of entangled states by quantum {M}arkov processes.
\newblock \emph{Phys. Rev. A}, 78:\penalty0 042307, Oct 2008.
\newblock \doi{10.1103/PhysRevA.78.042307}.

\bibitem[Kubo(1966)]{Kubo_fd}
R~Kubo.
\newblock The fluctuation-dissipation theorem.
\newblock \emph{Reports on Progress in Physics}, 29\penalty0 (1):\penalty0 255,
  1966.

\bibitem[Kupferman and Stuart(2004)]{kupferman2004fitting}
R~Kupferman and AM~Stuart.
\newblock Fitting {S}{D}{E} models to nonlinear {K}ac--{Z}wanzig heat bath
  models.
\newblock \emph{Physica D: Nonlinear Phenomena}, 199\penalty0 (3-4):\penalty0
  279--316, 2004.

\bibitem[Kupferman et~al.(2002)Kupferman, Stuart, Terry, and
  Tupper]{kupferman2002long}
R~Kupferman, AM~Stuart, JR~Terry, and PF~Tupper.
\newblock Long-term behaviour of large mechanical systems with random initial
  data.
\newblock \emph{Stochastics and Dynamics}, 2\penalty0 (04):\penalty0 533--562,
  2002.

\bibitem[Kupferman(2004)]{kupferman2004fractional}
Raz Kupferman.
\newblock Fractional kinetics in {K}ac--{Z}wanzig heat bath models.
\newblock \emph{Journal of Statistical Physics}, 114:\penalty0 291--326, 2004.

\bibitem[Lamb(1900)]{lamb1900peculiarity}
Horace Lamb.
\newblock On a peculiarity of the wave-system due to the free vibrations of a
  nucleus in an extended medium.
\newblock \emph{Proceedings of the London Mathematical Society}, 1\penalty0
  (1):\penalty0 208--213, 1900.

\bibitem[Lampo et~al.(2016)Lampo, Lim, Wehr, Massignan, and
  Lewenstein]{lampo2016lindblad}
Aniello Lampo, Soon~Hoe Lim, Jan Wehr, Pietro Massignan, and Maciej Lewenstein.
\newblock Lindblad model of quantum {B}rownian motion.
\newblock \emph{Physical Review A}, 94\penalty0 (4):\penalty0 042123, 2016.

\bibitem[{Lei} et~al.(2016){Lei}, {Baker}, and {Li}]{2016arXiv160602596L}
H.~{Lei}, N.~{Baker}, and X.~{Li}.
\newblock {Data-driven parameterization of the generalized {L}angevin
  equation}.
\newblock \emph{ArXiv e-prints}, June 2016.

\bibitem[Leimkuhler and Sachs(2018)]{leimkuhler2018ergodic}
Benedict Leimkuhler and Matthias Sachs.
\newblock Ergodic properties of quasi-{M}arkovian generalized {L}angevin
  equations with configuration dependent noise.
\newblock \emph{arXiv preprint arXiv:1804.04029}, 2018.

\bibitem[Lewis and Maassen(1984)]{lewis1984hamiltonian}
John~T Lewis and Hans Maassen.
\newblock Hamiltonian models of classical and quantum stochastic processes.
\newblock In \emph{Quantum Probability and Applications to the Quantum Theory
  of Irreversible Processes}, pages 245--276. Springer, 1984.

\bibitem[Lewis and Thomas(1975)]{lewis1975existence}
JT~Lewis and LC~Thomas.
\newblock On the existence of a class of stationary quantum stochastic
  processes.
\newblock \emph{Ann. Inst. Henri Poincare, Sect. A}, 22\penalty0 (3):\penalty0
  241--248, 1975.

\bibitem[Lim(2021)]{lim2021anomalous}
Soon~Hoe Lim.
\newblock Anomalous thermodynamics in homogenized generalized {L}angevin
  systems.
\newblock \emph{Journal of Physics A: Mathematical and Theoretical},
  54\penalty0 (15):\penalty0 155001, 2021.

\bibitem[Lim and Wehr(2019)]{lim2019homogenization}
Soon~Hoe Lim and Jan Wehr.
\newblock Homogenization for a class of generalized {L}angevin equations with
  an application to thermophoresis.
\newblock \emph{Journal of Statistical Physics}, 174:\penalty0 656--691, 2019.

\bibitem[Lim et~al.(2018)Lim, Wehr, Lampo, Garc{\'\i}a-March, and
  Lewenstein]{lim2018small}
Soon~Hoe Lim, Jan Wehr, Aniello Lampo, Miguel~{\'A}ngel Garc{\'\i}a-March, and
  Maciej Lewenstein.
\newblock On the small mass limit of quantum {B}rownian motion with
  inhomogeneous damping and diffusion.
\newblock \emph{Journal of Statistical Physics}, 170\penalty0 (2):\penalty0
  351--377, 2018.

\bibitem[Lim et~al.(2020)Lim, Wehr, and Lewenstein]{lim2020homogenization}
Soon~Hoe Lim, Jan Wehr, and Maciej Lewenstein.
\newblock Homogenization for generalized {L}angevin equations with applications
  to anomalous diffusion.
\newblock In \emph{Annales Henri Poincar{\'e}}, volume~21, pages 1813--1871.
  Springer, 2020.

\bibitem[Lindblad(1976)]{lindblad1976generators}
Goran Lindblad.
\newblock On the generators of quantum dynamical semigroups.
\newblock \emph{Communications in Mathematical Physics}, 48\penalty0
  (2):\penalty0 119--130, 1976.

\bibitem[Lindgren()]{lindgren2006lectures}
Georg Lindgren.
\newblock \emph{Lectures on stationary stochastic processes}.

\bibitem[Lindquist and Picci(2015)]{lindquist2015linear}
A.~Lindquist and G.~Picci.
\newblock \emph{Linear Stochastic Systems: A Geometric Approach to Modeling,
  Estimation and Identification}.
\newblock Series in Contemporary Mathematics. Springer Berlin Heidelberg, 2015.
\newblock ISBN 9783662457504.

\bibitem[Lowther()]{GeorgeLowtherBlog}
George Lowther.
\newblock \emph{SDEs with Locally Lipschitz Coefficients}.

\bibitem[{\L}uczka(2005)]{luczka2005non}
J~{\L}uczka.
\newblock Non-{M}arkovian stochastic processes: Colored noise.
\newblock \emph{Chaos: An Interdisciplinary Journal of Nonlinear Science},
  15\penalty0 (2):\penalty0 026107, 2005.

\bibitem[Lysy et~al.(2016)Lysy, Pillai, Hill, Forest, Mellnik, Vasquez, and
  McKinley]{lysy2016model}
Martin Lysy, Natesh~S Pillai, David~B Hill, M~Gregory Forest, John~WR Mellnik,
  Paula~A Vasquez, and Scott~A McKinley.
\newblock Model comparison and assessment for single particle tracking in
  biological fluids.
\newblock \emph{Journal of the American Statistical Association}, 111\penalty0
  (516):\penalty0 1413--1426, 2016.

\bibitem[Maassen(1984)]{maassen1984return}
Hans Maassen.
\newblock Return to thermal equilibrium by the solution of a quantum {L}angevin
  equation.
\newblock \emph{Journal of Statistical Physics}, 34\penalty0 (1-2):\penalty0
  239--261, 1984.

\bibitem[Maassen(1982)]{maassen1982class}
Johannes Dominicus~Maria Maassen.
\newblock \emph{On a class of quantum Langevin equations and the question of
  approach to equilibrium}.
\newblock PhD thesis, Rijksuniversiteit te Groningen, 1982.

\bibitem[Mandrekar and Gawarecki(2015)]{mandrekar2015stochastic}
V.S. Mandrekar and L.~Gawarecki.
\newblock \emph{Stochastic Analysis for Gaussian Random Processes and Fields:
  With Applications}.
\newblock Chapman \& Hall/CRC Monographs on Statistics \& Applied Probability.
  CRC Press, 2015.
\newblock ISBN 9781498707824.

\bibitem[Mao(2007)]{mao2007stochastic}
Xuerong Mao.
\newblock \emph{Stochastic Differential Equations and Applications}.
\newblock Elsevier, 2007.

\bibitem[Marsden and Ratiu(1995)]{marsden1995introduction}
Jerrold~E Marsden and Tudor~S Ratiu.
\newblock \emph{Introduction to Mechanics and Symmetry}, volume~48.
\newblock American Institute of Physics, 1995.

\bibitem[Massignan et~al.(2015)Massignan, Lampo, Wehr, and
  Lewenstein]{Massignan2015}
Pietro Massignan, Aniello Lampo, Jan Wehr, and Maciej Lewenstein.
\newblock Quantum {B}rownian motion with inhomogeneous damping and diffusion.
\newblock \emph{Phys. Rev. A}, 91:\penalty0 033627, Mar 2015.

\bibitem[{McKinley} and {Nguyen}(2017)]{mckinley2017}
S.~A {McKinley} and H.~D {Nguyen}.
\newblock {Anomalous diffusion and the Generalized Langevin Equation}.
\newblock \emph{ArXiv e-prints}, November 2017.

\bibitem[McKinley et~al.(2009)McKinley, Yao, and Forest]{mckinley2009transient}
Scott~A McKinley, Lingxing Yao, and M~Gregory Forest.
\newblock Transient anomalous diffusion of tracer particles in soft matter.
\newblock \emph{Journal of Rheology (1978-present)}, 53\penalty0 (6):\penalty0
  1487--1506, 2009.

\bibitem[Melbourne and Stuart(2011)]{melbourne2011note}
Ian Melbourne and AM~Stuart.
\newblock A note on diffusion limits of chaotic skew-product flows.
\newblock \emph{Nonlinearity}, 24\penalty0 (4):\penalty0 1361, 2011.

\bibitem[Merkli(2006)]{merkli2006ideal}
Marco Merkli.
\newblock The ideal quantum gas.
\newblock In \emph{Open Quantum Systems I}, pages 183--233. Springer, 2006.

\bibitem[Meyer(2006)]{meyer2006quantum}
Paul~A Meyer.
\newblock \emph{Quantum Probability for Probabilists}.
\newblock Springer, 2006.

\bibitem[Mori(1965)]{mori1965transport}
Hazime Mori.
\newblock Transport, collective motion, and {B}rownian motion.
\newblock \emph{Progress of Theoretical Physics}, 33\penalty0 (3):\penalty0
  423--455, 1965.

\bibitem[Nagy et~al.(2010)Nagy, Foias, Bercovici, and
  K{\'e}rchy]{nagy2010harmonic}
B{\'e}la~Sz Nagy, Ciprian Foias, Hari Bercovici, and L{\'a}szl{\'o} K{\'e}rchy.
\newblock \emph{Harmonic Analysis of Operators on Hilbert Space}.
\newblock Springer Science \& Business Media, 2010.

\bibitem[Neumann(1932)]{neumann1932mathematische}
Johann Neumann.
\newblock \emph{Mathematische grundlagen der quantenmechanik}.
\newblock Verlag von Julius Springer Berlin, 1932.

\bibitem[Nguyen(2018)]{nguyen2018small}
Hung~D Nguyen.
\newblock The small-mass limit and white-noise limit of an infinite dimensional
  generalized {L}angevin equation.
\newblock \emph{Journal of Statistical Physics}, 173\penalty0 (2):\penalty0
  411--437, 2018.

\bibitem[Nielsen and Chuang(2010)]{nielsen2010quantum}
Michael~A Nielsen and Isaac~L Chuang.
\newblock \emph{Quantum Computation and Quantum Information}.
\newblock Cambridge University Press, 2010.

\bibitem[Nurdin and Yamamoto(2017)]{nurdinlinear}
Hendra~I Nurdin and Naoki Yamamoto.
\newblock \emph{Linear Dynamical Quantum Systems}.
\newblock Springer, 2017.

\bibitem[Ohya and Volovich(2011)]{ohya2011mathematical}
Masanori Ohya and Igor Volovich.
\newblock \emph{Mathematical Foundations of Quantum Information and Computation
  and its Applications to Nano-and Bio-systems}.
\newblock Springer Science \& Business Media, 2011.

\bibitem[{\O}ksendal(2003)]{oksendal2003stochastic}
Bernt {\O}ksendal.
\newblock \emph{Stochastic Differential Equations}.
\newblock Springer, 2003.

\bibitem[Olivera-Atencio et~al.(2023)Olivera-Atencio, Lamata, and
  Casado-Pascual]{olivera2023benefits}
Mar{\'\i}a~Laura Olivera-Atencio, Lucas Lamata, and Jes{\'u}s Casado-Pascual.
\newblock Benefits of open quantum systems for quantum machine learning.
\newblock \emph{Advanced Quantum Technologies}, page 2300247, 2023.

\bibitem[{Ottobre} and {Pavliotis}(2011)]{Ottobre}
M.~{Ottobre} and G.~A. {Pavliotis}.
\newblock {Asymptotic analysis for the generalized Langevin equation}.
\newblock \emph{Nonlinearity}, 24:\penalty0 1629--1653, May 2011.
\newblock \doi{10.1088/0951-7715/24/5/013}.

\bibitem[Parthasarathy(2012)]{parthasarathy2012introduction}
K.R. Parthasarathy.
\newblock \emph{An Introduction to Quantum Stochastic Calculus}.
\newblock Modern Birkh{\"a}user Classics. Springer Basel, 2012.
\newblock ISBN 9783034805667.

\bibitem[Parthasarathy and Usha~Devi(2017)]{parthasarathy2017quantum}
KR~Parthasarathy and AR~Usha~Devi.
\newblock From quantum stochastic differential equations to {G}isin-{P}ercival
  state diffusion.
\newblock \emph{Journal of Mathematical Physics}, 58\penalty0 (8):\penalty0
  082204, 2017.

\bibitem[Pauli and Fierz(1938)]{pauli1938theorie}
Wolfgang Pauli and Markus Fierz.
\newblock Zur theorie der emission langwelliger lichtquanten.
\newblock \emph{Il Nuovo Cimento (1924-1942)}, 15\penalty0 (3):\penalty0
  167--188, 1938.

\bibitem[Pavliotis(2014)]{pavliotis2014stochastic}
G.A. Pavliotis.
\newblock \emph{Stochastic Processes and Applications: Diffusion Processes, the
  Fokker-Planck and Langevin Equations}.
\newblock Texts in Applied Mathematics. Springer New York, 2014.
\newblock ISBN 9781493913220.

\bibitem[Picci(1988)]{picci1988hamiltonian}
Giorgio Picci.
\newblock Hamiltonian representation of stationary processes.
\newblock In \emph{Contributions to Operator Theory and its Applications},
  pages 193--215. Springer, 1988.

\bibitem[Picci(1989)]{picci1989stochastic}
Giorgio Picci.
\newblock Stochastic aggregation.
\newblock In \emph{Proceedings of the Third German-Italian Symposium
  Applications of Mathematics in Industry and Technology}, pages 259--276.
  Springer, 1989.

\bibitem[Picci and Taylor(1992)]{picci1992generation}
Giorgio Picci and Thomas~J Taylor.
\newblock Generation of gaussian processes and linear chaos.
\newblock In \emph{Decision and Control, 1992., Proceedings of the 31st IEEE
  Conference on}, pages 2125--2131. IEEE, 1992.

\bibitem[Plenio and Huelga(2002)]{Plenio2002}
M.~B. Plenio and S.~F. Huelga.
\newblock Entangled light from white noise.
\newblock \emph{Phys. Rev. Lett.}, 88:\penalty0 197901, Apr 2002.
\newblock \doi{10.1103/PhysRevLett.88.197901}.

\bibitem[Reddy et~al.(2024)Reddy, Guenther, and Cho]{reddy2024data}
Sohail Reddy, Stefanie Guenther, and Yujin Cho.
\newblock Data-driven characterization of latent dynamics on quantum testbeds.
\newblock \emph{arXiv preprint arXiv:2401.09822}, 2024.

\bibitem[Reed and Simon(1981)]{reed1981functional}
M.~Reed and B.~Simon.
\newblock \emph{I: Functional Analysis}.
\newblock Methods of Modern Mathematical Physics. Elsevier Science, 1981.
\newblock ISBN 9780080570488.

\bibitem[Reed and Simon(1972)]{reed1972methods}
Michael Reed and Barry Simon.
\newblock \emph{Methods of Mathematical Physics I: Functional Analysis}.
\newblock Academic Press New York, 1972.

\bibitem[Rey-Bellet(2006)]{rey2006open}
Luc Rey-Bellet.
\newblock Open classical systems.
\newblock In \emph{Open Quantum Systems II}, pages 41--78. Springer, 2006.

\bibitem[{Rivas} and {Huelga}(2012)]{Rivas2012}
{\'A}.~{Rivas} and S.~F. {Huelga}.
\newblock \emph{{Open Quantum Systems}}.
\newblock SpringerBriefs in Physics, 2012.

\bibitem[R{\"o}thinger and Rozanov(2012)]{rozanov1987stationary}
B.~R{\"o}thinger and Y.A. Rozanov.
\newblock \emph{Introduction to Random Processes}.
\newblock Springer Series in Soviet Mathematics. Springer Berlin Heidelberg,
  2012.
\newblock ISBN 9783642727177.

\bibitem[Rouchon(2014)]{rouchon2014models}
Pierre Rouchon.
\newblock Models and feedback stabilization of open quantum systems.
\newblock \emph{arXiv preprint arXiv:1407.7810}, 2014.

\bibitem[Sakurai and Napolitano(2017)]{sakurai2017modern}
Jun~John Sakurai and Jim Napolitano.
\newblock \emph{Modern quantum mechanics}.
\newblock Cambridge University Press, 2017.

\bibitem[Samorodnitsky and Taqqu(1994)]{samorodnitsky1994stable}
G.~Samorodnitsky and M.S. Taqqu.
\newblock \emph{Stable Non-Gaussian Random Processes: Stochastic Models with
  Infinite Variance}.
\newblock Stochastic Modeling Series. Taylor \& Francis, 1994.
\newblock ISBN 9780412051715.

\bibitem[Schlosshauer(2007)]{SchlosshauerBook}
M.A. Schlosshauer.
\newblock \emph{Decoherence and the Quantum-To-Classical Transition}.
\newblock The Frontiers Collection. Springer, 2007.
\newblock ISBN 9783540357759.

\bibitem[Schuld et~al.(2015)Schuld, Sinayskiy, and
  Petruccione]{schuld2015introduction}
Maria Schuld, Ilya Sinayskiy, and Francesco Petruccione.
\newblock An introduction to quantum machine learning.
\newblock \emph{Contemporary Physics}, 56\penalty0 (2):\penalty0 172--185,
  2015.

\bibitem[Schuss(2009)]{schuss1988stochastic}
Zeev Schuss.
\newblock \emph{Theory and Applications of Stochastic Processes: An Analytical
  Approach}, volume 170.
\newblock Springer Science \& Business Media, 2009.

\bibitem[Siegle et~al.(2010)Siegle, Goychuk, Talkner, and
  H{\"a}nggi]{siegle2010markovian}
Peter Siegle, Igor Goychuk, Peter Talkner, and Peter H{\"a}nggi.
\newblock Markovian embedding of non-{M}arkovian superdiffusion.
\newblock \emph{Physical Review E}, 81\penalty0 (1):\penalty0 011136, 2010.

\bibitem[Siemon et~al.(2017)Siemon, Holevo, and Werner]{siemon2017unbounded}
Inken Siemon, Alexander~S Holevo, and Reinhard~F Werner.
\newblock Unbounded generators of dynamical semigroups.
\newblock \emph{Open Systems \& Information Dynamics}, 24\penalty0
  (04):\penalty0 1740015, 2017.

\bibitem[Sinha and Goswami(2007)]{sinha2007quantum}
Kalyan~B Sinha and Debashish Goswami.
\newblock \emph{Quantum Stochastic Processes and Noncommutative Geometry},
  volume 169.
\newblock Cambridge University Press, 2007.

\bibitem[Stella et~al.(2014)Stella, Lorenz, and
  Kantorovich]{PhysRevB.89.134303}
L.~Stella, C.~D. Lorenz, and L.~Kantorovich.
\newblock Generalized {L}angevin equation: An efficient approach to
  nonequilibrium molecular dynamics of open systems.
\newblock \emph{Phys. Rev. B}, 89:\penalty0 134303, Apr 2014.
\newblock \doi{10.1103/PhysRevB.89.134303}.

\bibitem[Strasberg(2022)]{strasberg2022quantum}
Philipp Strasberg.
\newblock \emph{Quantum Stochastic Thermodynamics: Foundations and Selected
  Applications}.
\newblock Oxford University Press, 2022.

\bibitem[Streater(2000)]{streater2000classical}
Raymond~F Streater.
\newblock Classical and quantum probability.
\newblock \emph{Journal of Mathematical Physics}, 41\penalty0 (6):\penalty0
  3556--3603, 2000.

\bibitem[Stuart and Warren(1999)]{stuart1999analysis}
AM~Stuart and JO~Warren.
\newblock Analysis and experiments for a computational model of a heat bath.
\newblock \emph{Journal of Statistical Physics}, 97\penalty0 (3-4):\penalty0
  687--723, 1999.

\bibitem[Toda et~al.(2012)Toda, Kubo, Kubo, Toda, Saito, Hashitsume, and
  Hashitsume]{toda2012statistical}
M.~Toda, R.~Kubo, R.~Kubo, M.~Toda, N.~Saito, N.~Hashitsume, and N.~Hashitsume.
\newblock \emph{Statistical Physics II: Nonequilibrium Statistical Mechanics}.
\newblock Springer Series in Solid-State Sciences. Springer Berlin Heidelberg,
  2012.
\newblock ISBN 9783642582448.

\bibitem[Trentelman et~al.(2002)Trentelman, Stoorvogel, and
  Hautus]{trentelman2002control}
Harry~L Trentelman, Anton~A Stoorvogel, and Malo Hautus.
\newblock Control theory for linear systems.
\newblock \emph{Communications and Control Engineering Series, Springer}, 2002.

\bibitem[Tupper(2002)]{tupper2002constructing}
PF~Tupper.
\newblock Constructing stationary {G}aussian processes from deterministic
  processes with random initial conditions.
\newblock 2002.

\bibitem[Van~Kampen(1998)]{van1998remarks}
NG~Van~Kampen.
\newblock Remarks on non-{M}arkov processes.
\newblock \emph{Brazilian Journal of Physics}, 28\penalty0 (2):\penalty0
  90--96, 1998.

\bibitem[Weiss(2008)]{Weiss}
U.~Weiss.
\newblock \emph{Quantum Dissipative Systems}.
\newblock World Scientific, Singapore, 2008.

\bibitem[Wiseman and Milburn(2009)]{wiseman2009quantum}
Howard~M Wiseman and Gerard~J Milburn.
\newblock \emph{Quantum Measurement and Control}.
\newblock Cambridge University Press, 2009.

\bibitem[Xue et~al.(2015)Xue, James, Shabani, Ugrinovskii, and
  Petersen]{xue2015quantum}
Shibei Xue, Matthew~R James, Alireza Shabani, Valery Ugrinovskii, and Ian~R
  Petersen.
\newblock Quantum filter for a class of non-{M}arkovian quantum systems.
\newblock In \emph{Decision and Control (CDC), 2015 IEEE 54th Annual Conference
  on}, pages 7096--7100. IEEE, 2015.

\bibitem[Xue et~al.(2017)Xue, Nguyen, James, Shabani, Ugrinovskii, and
  Petersen]{xue2017modelling}
Shibei Xue, Thien Nguyen, Matthew~R James, Alireza Shabani, Valery Ugrinovskii,
  and Ian~R Petersen.
\newblock Modelling and filtering for non-{M}arkovian quantum systems.
\newblock \emph{arXiv preprint arXiv:1704.00986}, 2017.

\bibitem[Yaglom(1952)]{yaglom1952introduction}
Akiva~Moiseevich Yaglom.
\newblock \emph{Introduction to the Theory of Stationary Random Functions},
  volume~7.
\newblock Russian Academy of Sciences, Steklov Mathematical Institute of
  Russian Academy of Sciences, 1952.

\bibitem[Zwanzig(2001)]{zwanzig2001nonequilibrium}
R.~Zwanzig.
\newblock \emph{Nonequilibrium Statistical Mechanics}.
\newblock Oxford University Press, 2001.
\newblock ISBN 9780198032151.

\bibitem[Zwanzig(1973)]{Zwanzig1973}
Robert Zwanzig.
\newblock Nonlinear generalized {L}angevin equations.
\newblock \emph{Journal of Statistical Physics}, 9\penalty0 (3):\penalty0
  215--220, 1973.
\newblock ISSN 1572-9613.

\end{thebibliography}

\end{document}